\documentclass{birkjour}
 \usepackage{amsthm}
 \usepackage{amssymb}
 \usepackage{bbold}
 \usepackage{mathtools}
 \usepackage{xcolor}
 \usepackage{enumitem}
 \usepackage{verbatim}
 \usepackage{hyperref}
 
 \hypersetup{
    colorlinks,
    linkcolor={red!50!black},
    citecolor={blue!50!black},
    urlcolor={blue!80!black}
}

 \newtheorem*{thm*}{Theorem}
 \newtheorem{thm}{Theorem}[section]
 \newtheorem{cor}[thm]{Corollary}
 \newtheorem{lem}[thm]{Lemma}
 \newtheorem{prop}[thm]{Proposition}
 \theoremstyle{definition}
 \newtheorem{defn}[thm]{Definition}
 \theoremstyle{remark}
 \newtheorem{rem}[thm]{Remark}
 
 \numberwithin{equation}{section}
 
\newcommand{\idmat}{\mathbb{1}}
\newcommand{\idop}{\mathrm{id}}
\newcommand{\blt}{\mathcal{B}}
\newcommand{\states}{\mathcal{S}}
\newcommand{\trcl}{\mathcal{B}_1}
\newcommand{\norm}[1]{\left\lVert #1 \right\rVert}
\newcommand{\abs}[1]{\left\vert #1 \right\vert}
\newcommand{\tr}[1]{\mathrm{tr}\left[#1\right]} 
\newcommand{\ptr}[2]{\mathrm{tr}_{#1}\left[#2\right]}
\newcommand{\set}[2]{\left\{ #1 \,\middle|\, #2 \right\}}
\DeclarePairedDelimiter\bra{\langle}{\rvert}
\DeclarePairedDelimiter\ket{\lvert}{\rangle}
\DeclarePairedDelimiterX\braket[2]{\langle}{\rangle}{#1 \delimsize\vert #2}
\newcommand{\sq}[1]{`#1'}

\begin{document}

%
%
%
%
%
%
%
%
%

\title[\sq{Interaction-Free} Channel Discrimination]
 {\sq{Interaction-Free} Channel Discrimination}


\author{Markus Hasen\"ohrl}

\address{
Zentrum Mathematik\br
Garching Forschungszentrum\br
Boltzmannstr. 3\br
85748 Garching bei M\"unchen}

\email{m.hasenoehrl@tum.de}

\thanks{$^{1}$Department of Mathematics, Technical University of Munich, Garching, Germany\\$^2$Munich Center for Quantum Science and Technology (MCQST), Munich, Germany.}
\author{Michael M. Wolf}
\address{Zentrum Mathematik\br
Garching Forschungszentrum\br
Boltzmannstr. 3\br
85748 Garching bei M\"unchen}
\email{m.wolf@tum.de}



%

\begin{abstract}
In this work, we investigate the question, which objects one can discriminate perfectly by \sq{interaction-free} measurements. 
To this end, we interpret the Elitzur-Vaidman bomb-tester experiment as a quantum channel discrimination problem and generalize the notion of \sq{interaction-free} measurement to arbitrary quantum channels. Our main result is a necessary and sufficient criterion for when it is possible or impossible to discriminate quantum channels in an \sq{interaction-free} manner (i.e., such that the discrimination error probability and the \sq{interaction} probability can be made arbitrarily small). For the case where our condition holds, we devise an explicit protocol with the property that both probabilities approach zero with an increasing number of channel uses, $N$. More specifically, the \sq{interaction} probability in our protocol decays as $\frac{1}{N}$ and we show that this rate is the optimal achievable one. Furthermore, our protocol only needs at most one ancillary qubit and might thus be implementable in near-term experiments. For the case where our condition does not hold, we prove an inequality that quantifies the trade-off between the error probability and the \sq{interaction} probability. 
\end{abstract}

\maketitle
\tableofcontents

\newpage
\section{Introduction}
In 1993, Elitzur and Vaidman proposed their famous bomb-tester experiment \cite{Elitzur1993} to demonstrate that the arguably most intriguing property of quantum theory - superposition - can be exploited to detect an ultra-sensitive bomb in a black-box, in such a way that there is a non-vanishing probability that the bomb will not explode. Only two years later, Kwiat et al. \cite{PhysRevLett.74.4763} showed how to employ another fundamental phenomenon - the quantum Zeno effect \cite{QuantumZenoEffect} - to boost the probability that the bomb will not explode as close to $1$ as one pleases. These powerful ideas found applications in \sq{interaction-free} imaging \cite{InteractionFreeImaging, NonInvasiveElectron}, counterfactual quantum computation \cite{10.1007/3-540-49208-9_7, CounterfactualQC, HoustenDirectQunatum}, counterfactual communication \cite{PhysRevLett.110.170502} and cryptography \cite{CounterfactualCrypto} and even complexity theory \cite{QueryCompelxity}. Despite the great success, it became apparent that the aforementioned techniques, which we will generically call \sq{interaction-free} measurements, are subject to some fundamental limitations. Notably, it is impossible to learn the outcome of a decision problem solved by a quantum computer \cite{CounterfactualQC} without \sq{running} the computer in at least one of the two cases, and two optically semi-transparent objects cannot be discriminated in such a way that no photon gets absorbed \cite{PhysRevA.63.032105, PhysRevA.64.062303}.

Despite the results mentioned above, there seems to be no general theory that pinpoints what can or cannot be done perfectly with \sq{interaction-free} measurements. Encouraged by recent results that generalize the quantum Zeno effect\cite{WolfZenoGeneralized, Burgarth1, Burgarth2, Barankai2018GeneralizedQZ}, this work aims to remedy these shortcomings. To this end, we interpret the Elitzur-Vaidman bomb-tester experiment as a quantum channel discrimination problem and generalize the notion of \sq{interaction-free} measurement to quantum channels via two slightly different, but in the end largely equivalent models. The theory of quantum supermaps \cite{ChiribellaQunatumNetwork} then provides the right framework to consider all possible (causally ordered) discrimination strategies and thus to decide, when it is possible or impossible to discriminate two channels in an \sq{interaction-free} manner.

\paragraph{Organization of the paper} This article is structured as follows: In the remainder of this section, we are going to review the bomb-tester experiment in its versions by Elitzur and Vaidman and by Kwiat et al. and try to convey the idea, how the general model should look like. Armed with that rough understanding, we will be able to state and discuss the major results of this work. This is the content of Section \ref{ResultsSection}. In Section \ref{ModelSection}, we give a detailed derivation of our model. Our main result, a characterization of what is possible and impossible to do with \sq{interaction-free} measurements, is the combination of two pillars: a no-go theorem, in the form of an inequality, that tells us when it is impossible to discriminate two channels in an \sq{interaction-free} manner; and a quantum protocol that discriminates two channels in those cases that are not touched by the no-go theorem. A quantitative treatment of this protocol will be given in Section \ref{ConstructivePart}, while the main content of Section \ref{NOGOSection} is the no-go theorem. Also in Section \ref{NOGOSection}, we prove fundamental limits for the achievable decay rate of the \sq{interaction} probability.

\paragraph{The bomb-tester experiment}
We will briefly review the bomb-tester experiment in its original version by Elitzur and Vaidman and its iterative version by Kwiat et al.. 
\begin{figure}[htbp]
    \centering
    \includegraphics[width=.9\textwidth]{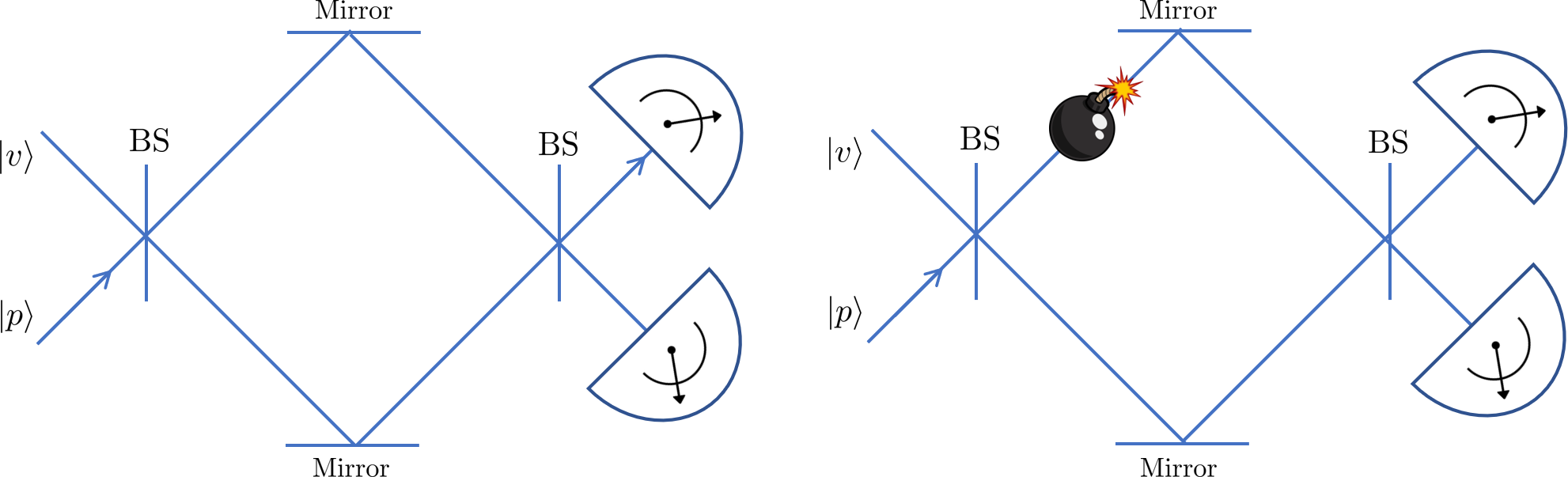}
    \caption{Elitzur-Vaidman bomb-tester experiment} \label{ElitzurFig}
\end{figure}
Suppose you got a box and you have been promised that inside of this box there is an ultra-sensitive bomb. By ultra-sensitive, we mean that the bomb will explode even if only one photon hits it. As you do not trust the deliverer, you want to check if there is a bomb inside the box. For some reason, the only way to obtain information about the content of the box is by shining light through it. Doing so, however, might trigger the bomb, which is what we want to avoid. Obviously, if photons were classical particles our task would be impossible. To circumvent this problem, Elizur and Vaidman proposed to put the box into the upper arm of a Mach-Zehnder interferometer, as depicted in Fig. \ref{ElitzurFig}. If we work only with a single photon, then this proposal can be stated abstractly as follows: The Hilbert space of the problem is $\mathcal{H} = \mathcal{H}_U \otimes \mathcal{H}_L$, where $\mathcal{H}_U = \mathcal{H}_L = \mathrm{span}\{v, p\}$ and the orthogonal unit vectors $v$ and $p$ denote the vacuum and one-photon states, respectively. The $50/50$ beamsplitter (BS) can be modeled as a unitary transformation $U$, defined by
\begin{align} \label{DefBeamsplitter}
    \begin{split}
    U v\otimes v &= v\otimes v \\
    U p\otimes v &= \cos(\theta)\, v\otimes p + \sin(\theta) \,p\otimes v\\
    U v\otimes p &= -\sin(\theta) \,v\otimes p + \cos(\theta) \,p\otimes v,
    \end{split}
\end{align}
where $\theta = 45^\circ$. Suppose we start with a photon in the lower input, then the initial state is $s_0 := \ket{v\otimes p}\bra{v\otimes p}$. There are two cases to analyze. On the one hand, if there is no bomb in the box, then the two beamsplitters rotate the state by $90^\circ$. Hence, the photon ends up in the upper output. On the other hand, if there is a bomb in the box, then the bomb acts as a measurement device in the upper path. There are three possible outcomes of the experiment. The first possibility is that the photon takes the upper path and thus causes the bomb to explode. This happens with a probability of 50\%. If the bomb does not explode, then, by the measurement postulate, the state of the system is still $s_0$. Since the second beamsplitter has a $50/50$ splitting ratio, the probability that we measure the photon in the upper output equals the probability that we measure the photon in the lower output. I.e., the probability for each of them is 25\%. The important point here is that in 25\% of the cases the photon ends up in the lower path. In that case, we can conclude that there is a bomb in the box, but the bomb has not been triggered. However, we only get this result in 25\% of the cases. 

\paragraph{Kwiat et al.'s iterative version}To increase the efficiency of this protocol, the crucial idea is to feed the output back to the input, (thus to let the photon go through the box many times) and to adjust the splitting ratio of the beamsplitters sensibly (see \cite{PhysRevLett.74.4763} for the experimental realization). The easiest way to analyze this proposal is to think of the feedback loop in a \sq{rolled out} way. That is, we look at this proposal as if we had $N$ copies of the Mach-Zehnder interferometer (where $N$ is the number of times we let the photon go through the box), in each of which the box is in the upper arm (see Fig. \ref{KwiatBombFig}).
\begin{figure}[htbp]
    \centering
    \includegraphics[width=.9\textwidth]{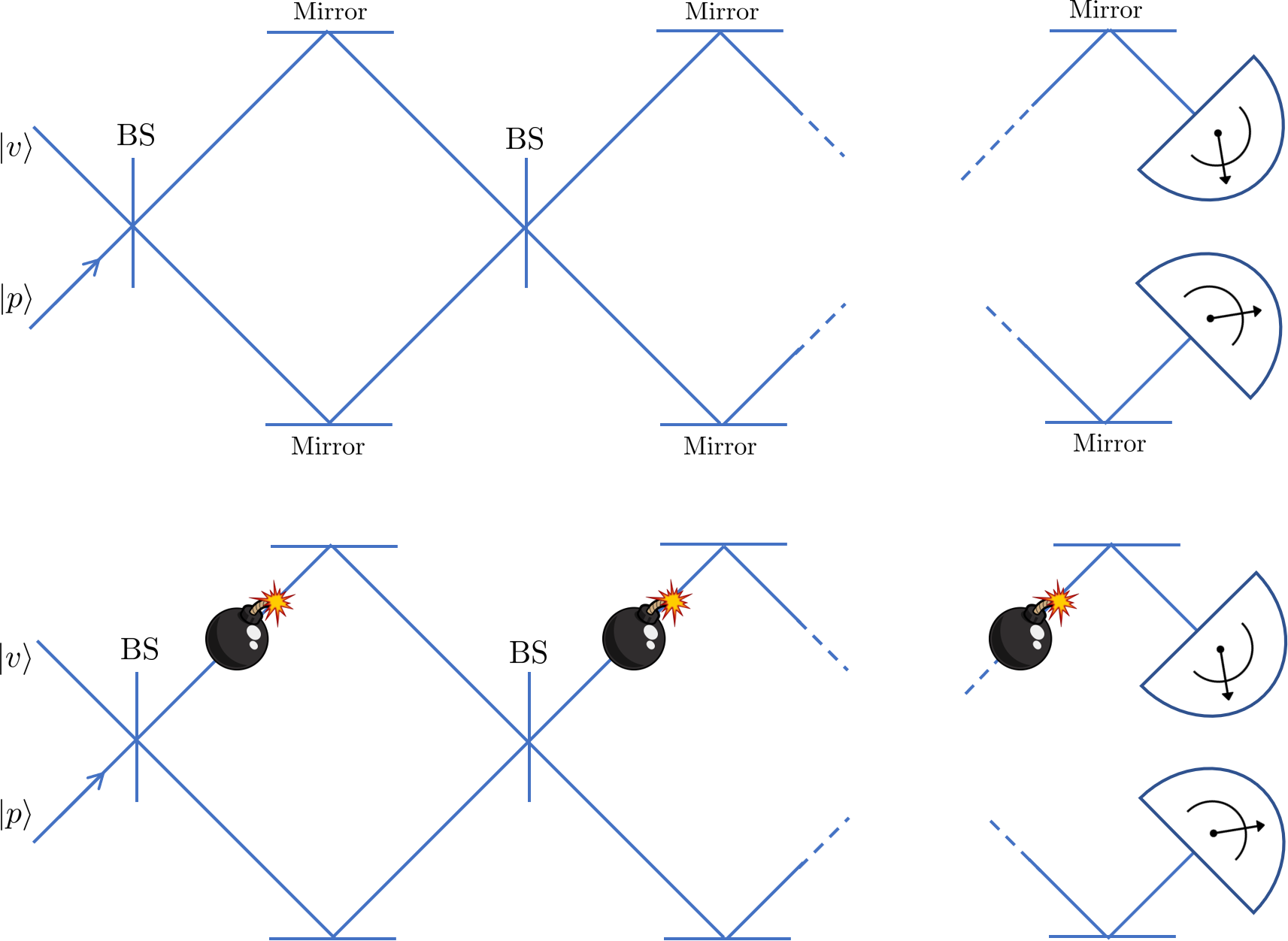} 
    \caption{Kwiat et al.'s version of the bomb-tester experiment} \label{KwiatBombFig}
\end{figure} We further choose the angle $\theta := \frac{90^\circ}{N}$ in \eqref{DefBeamsplitter}, which defines the action of the beamsplitters. Let us analyze this protocol: If there is no bomb in the box and the photon starts in the lower path, then the photon travels through $N$ beamsplitters, each of which rotates the state by an angle of $\frac{90^\circ}{N}$. So overall the state is rotated by $90^\circ$, which means that the photon will be in the upper output. For the case where there is a bomb in the box, let us calculate the probability that the photon always takes the lower path and therefore does not hit the bomb. For each of the beamsplitters, if the photon is in the lower path before the beamsplitter, then the probability that the photon will be in the lower path after the beamsplitter is given by $\cos^{2}(\theta)$. Since the bomb can be viewed as a measurement device, the probability that the photon always takes the lower path is simply the product of the probabilities at each beamsplitter. Hence, $P(\text{always lower path}) = \cos^{2N}(\theta)$. For $N \rightarrow \infty$, we have
\begin{align*}
    \cos^{2N}(\theta) = \left(1-\frac{\pi^2}{8N^2} + \mathcal{O}(N^{-4})\right)^{2N} = 1-\frac{\pi^2}{4N} + \mathcal{O}(N^{-2}) \xrightarrow{N \rightarrow \infty} 1.
\end{align*}
This simple calculation has the remarkable consequence that (when $N$ is large enough) the photon will always end up in the lower path and the bomb will not explode. Since the photon will always end up in the upper path, if there is no bomb in the box, this protocol enables us to tell (with probability approaching $1$), whether there is a bomb in the box, while simultaneously ensuring that the bomb will not be triggered. 
\paragraph{Interpretation as a channel discrimination problem} We have seen in the previous paragraph, how to discriminate between a completely transparent object (empty box) and an opaque object (bomb) such that the probability that a photon gets absorbed by the opaque object can be made as small as one pleases. This problem can be reinterpreted as a channel discrimination problem as follows: The channel corresponding to the transparent object is simply the identity channel ($T_{empty} := \idop$), while the action of the opaque object can be identified with the channel\footnote{$\trcl(\mathcal{H})$ denotes the set of traceclass operators on the Hilbert space $\mathcal{H}$ and $\states(\mathcal{H})$ denotes the set of density operators.} $T_{bomb} : \trcl(\mathcal{H}_U) \rightarrow \trcl(\mathcal{H}_U)$, defined by\begin{align*}
    T_{bomb}(\cdot) = \tr{\cdot} \ket{v}\bra{v}.
\end{align*}
\begin{figure}[htbp]
    \centering
    \includegraphics[width=.9\textwidth]{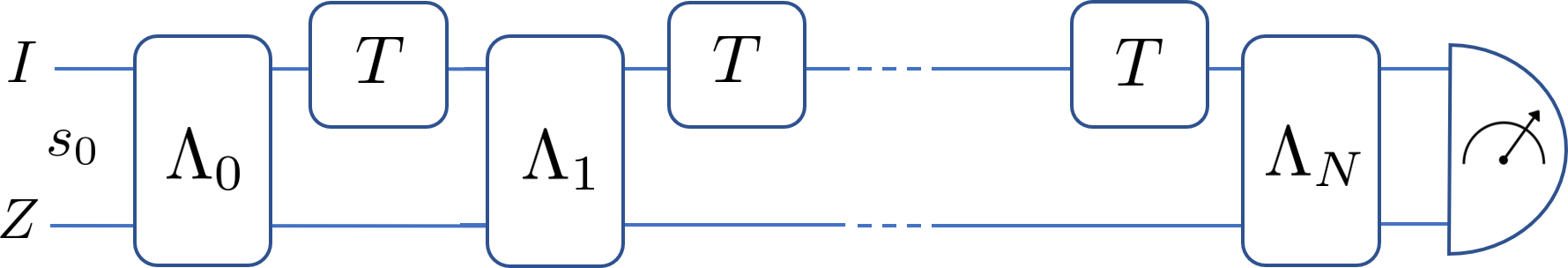}
    \caption{$N$-step discrimination strategy} \label{GeneralDiscriminationProtocol}
\end{figure}According to the theory of quantum combs\footnote{Quantum combs: also known as quantum supermaps, quantum strategies, \dots}, the most general (causally ordered) strategy to discriminate channels is given by the sequential scheme, depicted in Figure \ref{GeneralDiscriminationProtocol}. That is, if the channels to be discriminated act on the system $I$ ($I$ for interaction), then the most general discrimination strategy\footnote{This includes in particular coherent evolution, the use of entanglement, measurements, adaptive strategies, channels used in paralell, \dots} allowed by quantum theory can be described as follows: First one chooses an ancillary system $Z$ (which might be arbitrarily large) and an initial state $s_0 \in \states(\mathcal{H}_I \otimes \mathcal{H}_Z)$. Then we can apply a channel\footnote{Of course, the application of $\Lambda_0$ is redundant, since one could choose $s_0$ differently. Allowing to apply $\Lambda_0$, however, will simplify the notation.} $\Lambda_0 : \trcl(\mathcal{H}_I \otimes \mathcal{H}_Z) \rightarrow \trcl(\mathcal{H}_I \otimes \mathcal{H}_Z)$ to $s_0$. Afterwards, the unknown channel is applied to the system (i.e., if $T : \trcl(\mathcal{H}_I) \rightarrow \trcl(\mathcal{H}_I)$ is the unknown channel, then its application transforms the state $\Lambda_0(s_0)$ to $(T\otimes \idop)(\Lambda_0(s_0))$). Then we can transform the system by applying a channel $\Lambda_1 : \trcl(\mathcal{H}_I \otimes \mathcal{H}_Z) \rightarrow \trcl(\mathcal{H}_I \otimes \mathcal{H}_Z)$. Afterwards, we apply the unknown channel again, followed by an application of a channel $\Lambda_2 : \trcl(\mathcal{H}_I \otimes \mathcal{H}_Z) \rightarrow \trcl(\mathcal{H}_I \otimes \mathcal{H}_Z)$. We repeat this process $N$ times overall. In the end, our system is in a state $\rho_N^T \in \states(\mathcal{H}_I \otimes \mathcal{H}_Z)$, which depends on $T$.  Hence, by measuring, we can obtain information about the identity of $T$. Kwiat et al.'s protocol can be integrated in this formalism as follows: We identify the upper path with the system $I$ and the lower path with the system $Z$ and choose $s_0 := \ket{v \otimes p} \bra{v \otimes p}$. For $0 \leq i \leq N-1$, the channels $\Lambda_i$ are defined by $\Lambda_i(\cdot) := U\cdot U^\dagger = : \hat{U}(\cdot)$, with $\theta = \frac{90^\circ}{N}$ and we set $\Lambda_N := \idop$. It is then easy to calculate that 
\begin{align} \label{KwiatAsQCEval}
\begin{split}
    \rho_N^{T_{empty}} &= \hat{U}^N(\ket{v \otimes p} \bra{v \otimes p}) = \ket{p\otimes v}\bra{p \otimes v}\\
    \rho_N^{T_{bomb}} &= \left((T_{bomb}\otimes\idop)\circ\hat{U}\right)^N(\ket{v \otimes p} \bra{v \otimes p}) \\&=  \cos^{2N}(\theta) \ket{v\otimes p}\bra{v\otimes p} + (1-\cos^{2N}(\theta)) \ket{v\otimes v}\bra{v\otimes v},
\end{split}
\end{align}
where $\rho_N^{T_{empty}}$ and $\rho_N^{T_{bomb}}$ denote the output states of the protocol, when the unknown channel is $T_{empty}$ or $T_{bomb}$. An interesting aspect of the expressions \eqref{KwiatAsQCEval} is that one can read off the results of the last paragraph, since the states are orthogonal and since the probability that the bomb explodes is simply given by the coefficient of $\ket{v\otimes v}\bra{v\otimes v}$. To abstract from the bomb-tester experiment, we want to allow for arbitrary quantum channels and for arbitrary discrimination strategies (Figure \ref{GeneralDiscriminationProtocol}). In this more general setting, the concept of the output state does not change. What is not a priori clear is, what it means that something was \sq{interaction-free}. Since we want to allow for arbitrary strategies (for example, involving many photons in arbitrary superpositions), the output state does not, in general, contain the information if an interaction occurred. Therefore, we need to model separately what \sq{interaction-free} means for general discrimination strategies. A derivation of such a model based on some axioms takes some effort. We will, therefore, postpone this discussion until Section \ref{ModelSection}. For now, let us just describe the essential constituents. Firstly, for the notion of \sq{interaction-free} to have any meaning, there needs to be some way not to interact with the object in the box. We will thus assume, in analogy to the bomb-tester experiment, the existence of a \textit{vacuum state}. That is, we assume that for the channels under consideration, there exists a pure state $\ket{v}\bra{v} \in \states(\mathcal{H}_I)$ such that $\ket{v}\bra{v}$ gets mapped to a pure state by the channel and that if the channel is applied to $\ket{v}\bra{v}$, then there is no \sq{interaction} with the object in the box. This concept is formalized by the notion of a channel with vacuum. 
\begin{defn}[Channel with vacuum] \label{ChannelWithVacuumDef}
A \textit{channel with vacuum $v \in \mathcal{H}$} is a channel $T : \trcl(\mathcal{H}) \rightarrow \trcl(\mathcal{H})$ together with a unit vector $v \in \mathcal{H}$ such that $T(\ket{v}\bra{v})$ is pure. 
The unit vector $v$ is called the $\textit{vacuum}$ and the state $\ket{v}\bra{v} \in \states(\mathcal{H})$ is called the \textit{vacuum state}. 
\end{defn}
The notion of an object in the box already suggests that we should look at the given channel in the open system picture. To this end, we imagine a Daemon sitting in the box and trying to figure out if something else than the vacuum was sent through the box. To do so, we allow the Daemon to access the object in the box. In more mathematical terms, the Daemon has full access to the output of the conjugate channel \cite{King2007PropertiesOC}. An important implicit assumption underlying the discussion above is that the channels we look at can be applied several times (which means that the channel does not change) - a Markovianity assumption. Given just this Markovianity assumption, it is possible to determine the probability that, for a certain discrimination strategy, the Daemon will find out if at any point during the execution of the strategy, the channel was applied to something else than the vacuum state. We will call this probability the \textit{\sq{interaction} probability} (see Definition \ref{InteractionProbabilityDef}), denoted by $P_I^T(D)$, where $T$ denotes the channel and $D$ the discrimination strategy. The central notion of \textit{discrimination in an \sq{interaction-free} manner}, as formalized in Definition \ref{Interaction-FreeDiscriminationFormalDef}, is then defined by demanding that the discrimination error probability as well as the \sq{interaction} probability can be made arbitrarily small simultaneously. We finish this section by formalizing the notion of a discrimination strategy\footnote{Note that in this definition, we allow the input and output spaces to be different from $\mathcal{H}\otimes \mathcal{H}_Z$. This is solely for notational flexibility and has no physical significance.} and by fixing the notation.
\begin{defn}[Discrimination strategy]
An \textit{$N$-step discrimination strategy} is a tupel $(\mathcal{H}, \mathcal{H}_Z, \mathcal{H}_i, \mathcal{H}_o, s_0, \Lambda)$, where $\mathcal{H}$, $\mathcal{H}_Z$, $\mathcal{H}_i$ and $\mathcal{H}_o$ are Hilbert spaces, $s_0 \in \states(\mathcal{H}_i)$ is the initial state and $\Lambda := \{\Lambda_0, \Lambda_1, \dots, \Lambda_N\}$ is a set of channels, with $\Lambda_0 : \trcl(\mathcal{H}_i) \rightarrow \trcl(\mathcal{H} \otimes \mathcal{H}_Z)$, $\Lambda_n : \trcl(\mathcal{H} \otimes \mathcal{H}_Z) \rightarrow \trcl(\mathcal{H} \otimes \mathcal{H}_Z)$ for $1 \leq n \leq N-1$ and $\Lambda_N : \trcl(\mathcal{H} \otimes \mathcal{H}_Z) \rightarrow \trcl(H_o)$. \\
An $N$-step discrimination strategy induces the \textit{intermediate state map} $\rho : \blt(\trcl(\mathcal{H})) \times \{0, 1, 2, \dots, N\} \rightarrow \trcl(\mathcal{H}\otimes \mathcal{H}_Z) \cup \trcl(\mathcal{H}_o)$, defined by
\begin{align*}
\begin{split}
    \rho(T, 0) &= \Lambda_0(s_0)\\
    \rho(T, n) &= \Lambda_n\circ (T \otimes \idop) \circ \rho(T, n-1), \text{ for } 1 \leq n \leq N. 
\end{split}
\end{align*}
We will always write\footnote{The superscript should not be confused with the transpose.} $\rho^T_n$ for $\rho(T, n)$ and omit $\mathcal{H}_i$ and $\mathcal{H}_o$, if $\mathcal{H}_i = \mathcal{H}_o = \mathcal{H}\otimes\mathcal{H}_Z$.
\end{defn}

\paragraph{Notation}
Throughout, $\mathcal{H}$ (with some subscript) denotes a separable complex Hilbert space and in this paragraph, $\mathcal{X}$ and $\mathcal{Y}$ are Banach spaces.
The range of a map $f : \mathcal{X} \rightarrow \mathcal{Y}$ is denoted by $\mathrm{ran}(f) := \set{f(x)}{x \in \mathcal{X}}$. The kernel of $f$ is $\mathrm{ker}(f) := \set{x \in \mathcal{X}}{f(x) = 0}$. The dual space $\mathcal{X}^*$ of $\mathcal{X}$ is the set of bounded linear functionals on $\mathcal{X}$. The orthogonal complement of a linear subspace $\mathcal{V} \subseteq \mathcal{H}$ is denoted by $\mathcal{V}^\bot$. The open $\epsilon$-ball around $x_0\in \mathcal{X}$ is defined by $B_\epsilon(x_0) := \set{x \in \mathcal{X}}{\norm{x-x_0} < \epsilon}$ and the closed $\delta$-disc around $z_0 \in \mathbb{C}$ is denoted by $\mathbb{D}_{\delta}(z_0) := \set{z \in \mathbb{C}}{\abs{z-z_0} \leq \delta}$

The Banach space of bounded linear operators $\mathcal{X} \rightarrow \mathcal{X}$ is denoted by $\blt(\mathcal{X})$. The space of trace-class operators $\trcl(\mathcal{H})$ becomes a Banach space with trace-norm $\norm{\cdot}_1 := \tr{\abs{\cdot}}$. For $A \in \blt(\mathcal{H})$, the adjoint is denoted by $A^\dagger$ and the support of $A$ is defined by $\mathrm{supp}(A) := \mathrm{ker}(A)^\bot$. If $A^\dagger = A$, then $A$ is called self-adjoint. $A$ is called positive semi-definite, sometimes denoted by $A \geq 0$, if $A$ is self-adjoint and $\braket{\psi}{A\psi} \geq 0$ for all $\psi \in \mathcal{H}$. For a closed subspace $\mathcal{V} \subseteq \mathcal{H}$, we denote (in a slight abuse of notation) by $\blt(\mathcal{V}) \subseteq \blt(\mathcal{H})$ the bounded linear operators with range and support in $\mathcal{V}$ and by $\trcl(\mathcal{V})$ the trace-class operators with range and support in $\mathcal{V}$. 

A linear operator $T \in \blt(\trcl(\mathcal{H}))$ is called a quantum operation, if it is completely positive and trace non-increasing. If $T$ is even trace-preserving, then $T$ is called a (quantum) channel. If a quantum channel $T$ is written in the form $T(\cdot) = \ptr{E}{V\cdot V^\dagger}$, where $V : \mathcal{H} \rightarrow \mathcal{H}_E \otimes \mathcal{H}$ is an isometry and where $\mathrm{tr}_E$ is the partial trace, then $V$ is called a Stinespring isometry. The set of (quantum) states on $\mathcal{H}$ is given by $\states(\mathcal{H}) := \set{\rho \in \trcl(\mathcal{H})}{\rho \geq 0, \tr{\rho} = 1}$. The identity channel is denoted by $\idop$ and the unit matrix by $\idmat$.  
For positive semi-definite trace-class operators $\rho$ and $\sigma$, the fidelity is defined by $\sqrt{F}(\rho, \sigma) := \norm{\sqrt{\rho}\sqrt{\sigma}}_1$.
 
For $B \in \blt(\mathcal{X})$, the resolvent set is $\rho(B) := \set{z\in\mathbb{C}}{z-B \text{ is invertible}}$ and the spectrum is $\sigma(B) := \mathbb{C}\setminus\rho(B)$. The discrete spectrum of $B$ is the subset of isolated points of $\sigma(B)$ such that the corresponding Riesz projection has finite rank.  

\newpage
\section{Results} \label{ResultsSection}
To state and discuss our main results, we need one more concept, which is similar to that of a decoherence-free subspace\footnote{An isometric subspace is a decoherence-free subspace, if the range of the isometry is $\mathcal{V}$.}.
\begin{defn}[Isometric subspace] \label{DefIsometricSubspace}
Let $\mathcal{V}$ be a closed linear subspace of a Hilbert space $\mathcal{H}$. A channel $T : \trcl(\mathcal{H}) \rightarrow \trcl(\mathcal{H})$ is said to be \textit{isometric on $\mathcal{V}$}, if there exists an isometry $V : \mathcal{V} \rightarrow \mathcal{H}$, such that
\begin{align*}
    T\vert_{\trcl(\mathcal{V})}(\cdot) = V \cdot V^\dagger.
\end{align*}
If $T$ is isometric on $\mathcal{V}$, we call $\mathcal{V}$ an \textit{isometric subspace} w.r.t. $T$. 
\end{defn}

\noindent The significance of channels that are isometric on $\mathcal{V}$ is that they are the analogue to the identity channel in the bomb-tester case. To see why, note that $T|_{\trcl(\mathcal{V})}$ satisfies the Knill-Laflamme error-correcting conditions \cite{PhysRevLett.84.2525}. Hence, by composing $T|_{\trcl(\mathcal{V})}$ with an appropriate channel, we obtain the identity channel on $\trcl(\mathcal{V})$. Furthermore, as Lemma \ref{MaximalVacuumSubspaceLemma} proves in a language adapted to our model, the output of the conjugate channel of $T$ will be the same for all $\rho \in \trcl(\mathcal{V})$. In particular, if we have $v \in \mathcal{V}$, where $v$ is the vacuum, then even though $\rho \in \trcl(\mathcal{V})$ might be different form $\ket{v}\bra{v}$, the Daemon (having access to the conjugate channel only) has no chance to tell that something else than the vacuum has been sent through the box.

We are now ready to state our main result, which is an easy to check necessary and sufficient criterion that tells us when it is possible (or impossible) to discriminate two quantum channels in an \sq{interaction-free} manner.   
\begin{thm}[Main result - qualitative] \label{MainResult}
Let $\mathrm{dim}(\mathcal{H}) < \infty$. Two channels $T_A, T_B : \trcl(\mathcal{H}) \rightarrow \trcl(\mathcal{H})$ with vacuum $v \in \mathcal{H}$ can be discriminated in an \sq{interaction-free} manner, if and only if there exists a subspace $\mathcal{V} \subseteq \mathcal{H}$ with the following three properties:
\begin{enumerate}
    \item $v \in \mathcal{V}$ \label{MainProp1}
    \item At least one of the two channels is isometric on $\mathcal{V}$.\label{MainProp2}
    \item $T_A\vert_{\blt(\mathcal{V})} \neq T_B\vert_{\blt(\mathcal{V})}$\label{MainProp3}
\end{enumerate}
\end{thm}
\begin{rem}
At first glance it may seem to be hard to check whether such a subspace exists. This is not so, as one only needs to consider two candidates for $\mathcal{V}$, the so called maximal vacuum subspaces $\mathcal{V}_{T_A}$ and $\mathcal{V}_{T_B}$, which we define and study in \ref{MaximalVacuumSubspaceDef} and \ref{MaximalVacuumSubspaceLemma}.     
\end{rem}

\noindent Theorem \ref{MainResult} is the qualitative combination of two quantitative results. 
\paragraph{The constructive case} We consider the case, where there is a subspace $\mathcal{V}$, such that $\mathcal{V}$ contains the vacuum and one of the two channels is isometric on $\mathcal{V}$ and $T_A\vert_{\blt(\mathcal{V})} \neq T_B\vert_{\blt(\mathcal{V})}$. In this case, our main theorem says that we can discriminate the two channels in an \sq{interaction-free} manner. It turns out that one does not need complete information about the two channels to perform the discrimination task. To account for this, we consider the more general task, where we want to know to which one of two known, disjoint, sets of channels the unknown channel belongs. Of course, Theorem \ref{MainResult} puts some restrictions on how these sets may look like. 
Specifically we consider the following: Given a channel $T$ with vacuum $v \in \mathcal{V}$ that is isometric on $\mathcal{V}$, we take as our first set (a subset of) the set of channels that equal $T$, if we restrict their domains to $\trcl(\mathcal{V})$. The second set is less restricted in that we only assume that all channels must be channels with (the same) vacuum $v$ and that the restrictions to $\trcl(\mathcal{V})$ must not equal $T|_{\trcl(\mathcal{V})}$. It will then turn out that under these conditions, these two sets can be discriminated in an \sq{interaction-free} manner. Roughly speaking, this tells us that we can test whether the unknown channel is $T$ or some other channel, whoes identity is unknown. Put it yet another way: If the identity channel is interpreted as an empty box and every other channel as a non empty box, then our result says that one can always find out, (in an \sq{interaction-free} manner) if there is something or nothing in the box. Before we state this in mathematical terms, we need to define the \textit{discrimination error probability} for two sets. 
\begin{defn}[Error probability]\label{ErrorProbabilityDefn}
Let $\mathcal{C}_A, \mathcal{C}_B \subseteq \blt(\trcl(\mathcal{H}))$ be two sets of channels. For an $N$-step discrimination strategy $D$
and a two-valued POVM $\Pi = \{\pi_A, \pi_B\}$, the \textit{discrimination error probability} is defined by
\begin{align*}
    P_e(D,\Pi) &:= \frac{1}{2} \left[ \sup_{T \in \mathcal{C}_A} \tr{\pi_B \rho_N^{T}} + \sup_{T \in \mathcal{C}_B} \tr{\pi_A \rho_N^{T}} \right].
\end{align*}
\end{defn}

\begin{thm}[Discrimination strategy] \label{DiscriminationStroategyOutlineTheorem}
For $\mathrm{dim}(\mathcal{H}) < \infty$, let $\mathcal{C}_A, \mathcal{C}_B \subseteq \blt(\trcl(\mathcal{H}))$ be two closed sets of channels and $\mathcal{V}$ be a subspace of $\mathcal{H}$, such that
\begin{enumerate}
    \item For all $T \in \mathcal{C}_A \cup \mathcal{C}_B$, $T$ is a channel with vacuum $v \in \mathcal{V}$.  
    \item For all $T \in \mathcal{C}_A$, $T$ is isometric on $\mathcal{V}$.
    \item The set $\mathcal{C}_A\vert_{\trcl(\mathcal{V})} := \set{T\vert_{\trcl(\mathcal{V})}}{T \in \mathcal{C}_A}$ contains exactly one element.
    \item $\mathcal{C}_A\vert_{\trcl(\mathcal{V})}$ and $\mathcal{C}_B\vert_{\trcl(\mathcal{V})} := \set{T\vert_{\trcl(\mathcal{V})}}{T \in \mathcal{C}_B}$ are disjoint.
\end{enumerate} Then there exist a constant $C$, and for every $N \in \mathbb{N}$, an $N$-step discrimination strategy $D$ and a two-valued POVM $\Pi$, such that 
\begin{gather*}
P_e(D, \Pi) \leq \frac{C}{N^2}, \\
    P_I^{T_A}(D) = 0 \quad \text{ and } \quad P_I^{T_B}(D) \leq \frac{C}{N}, 
\end{gather*}
for all $T_A \in \mathcal{C}_A$ and all $T_B \in \mathcal{C}_B$, where $P_I$ denotes the \sq{interaction} probability. Thus, the sets $\mathcal{C}_A$ and $\mathcal{C}_B$ can be discriminated in an \sq{interaction-free} manner.
\end{thm} 
\begin{rem}
We will not only show the existence of the proclaimed strategy, but propose an explicit one. Our strategy needs only one ancillary qubit system in the worst case scenario (as does the Kwiat et al. protocol) and might thus be implementable in the near future. We also show that one cannot get rid of the ancillary qubit in a naive way. 
\end{rem}
\begin{rem}
Although Theorem \ref{DiscriminationStroategyOutlineTheorem} is formulated for finite-dimensional spaces, a key part of the proof works also in infinite-dimensional spaces (Theorems \ref{FundamentalTheoremPositivePart} and \ref{GeneralTheoremCompactSet}). 
\end{rem}
\begin{rem}
For two channels $T_A$ and $T_B$ with vacuum $v \in \mathcal{H}$, we can define the sets $\mathcal{C}_A := \{T_A\}$ and $\mathcal{C}_B := \{T_B\}$. If there is a subspace $\mathcal{V}$ such that the conditions 1-3 in the main theorem are fulfilled (and, w.l.o.g, $T_A$ is isometric on $\mathcal{V}$), then clearly $\mathcal{C}_A$ and $\mathcal{C}_B$ satisfy the hypothesis of Theorem \ref{DiscriminationStroategyOutlineTheorem} and thus $T_A$ and $T_B$ can be discriminated in an \sq{interaction-free} manner. This proves the direct part of Theorem \ref{MainResult}.
\end{rem}
Given the result of Theorem \ref{DiscriminationStroategyOutlineTheorem}, it is natural to ask, whether the bounds on the error probability and the discrimination probability have the optimal dependence on $N$. This is clearly not the case for the error probability, as is already evident from the bomb-tester experiment. For the \sq{interaction} probability, we were able to show (under a mild condition on $\mathcal{C}_A$ and $\mathcal{C}_B$) that $N^{-1}$ is indeed the best possible rate. We state this as a meta theorem (see Theorem \ref{RateLimitTheorem}).
\begin{thm*}
Subject to a condition stated in Theorem \ref{RateLimitTheorem}, there exists a constant $C > 0$ such that 
\begin{align*}
\max(P_I^{T_A}(D), P_I^{T_B}(D)) \geq C\,\frac{(1-2P_e(D,\Pi))^4}{N},
\end{align*}
for all $N$-step discrimination strategies $D$ and all two-valued POVM's $\Pi$.
\end{thm*}
\noindent The result above cannot hold unconditionally. If there is a subspace $\mathcal{V}$ such that $v \in \mathcal{V}$, \textit{both} channels are isometric on $\mathcal{V}$ and $T_A\vert_{\blt(\mathcal{V})} \neq T_B\vert_{\blt(\mathcal{V})}$, then we can restrict ourselves to probing the channel only with states in $\rho \in \trcl(\mathcal{V})$. Since the Daemon cannot tell the difference between these states, the \sq{interaction} probability is zero and the remaining problem is to discriminate two isometric channels. That problem can be solved with discrimination error probability equal to zero, in a finite number of steps \cite{PhysRevLett.87.177901}. We were unable to show that the case described above is the only one where the $N^{-1}$-rule can be violated, but this seems plausible. 
\noindent 
\paragraph{The no-go case} The second case to consider is where there exists no subspace satisfying all three properties of Theorem \ref{DefIsometricSubspace}. In other words, in this case the channels $T_A$ and $T_B$ must be such that whenever there is a subspace $\mathcal{V}$ that contains the vacuum and on which at least one of the two channels is isometric, then the two channels must necessarily be the same on that subspace\footnote{Unfortunately, this case seems to be the generic case. Indeed, on physical grounds (think of two semi-transparent objects) it is reasonable to assume that for both channels, the only isometric subspace that contains the vacuum is simply $\mathrm{span}\{v\}$ and that $\ket{v}\bra{v}$ is a fixed point.}. In this case, we were able to establish the following theorem.
\begin{thm}[No-go theorem] 
For $\mathrm{dim(\mathcal{H})} < \infty$, let $T_A, T_B : \trcl(\mathcal{H}) \rightarrow \trcl(\mathcal{H})$ be two channels with vacuum $v\in \mathcal{H}$. Suppose that no subspace satisfies the properties \ref{MainProp1}, \ref{MainProp2} and \ref{MainProp3} of Theorem \ref{MainResult} simultaneously.\\ Then there exists a constant $C > 0$, such that
\begin{align*}
    (1-2P_e(D, \Pi))^2 \leq  C \max(P_I^{T_A}(D), P_I^{T_B}(D)),
\end{align*}
for all finite-dimensional $N$-step discrimination strategies $D$ and all two-valued POVMs, $\Pi$.  
Hence, $T_A$ and $T_B$ cannot be discriminated in an \sq{interaction-free} manner.
\end{thm}
\noindent Clearly, this implies the converse in Theorem \ref{MainResult}.

As a byproduct, we obtained an inequality for the fidelity, which might be of independent interest.
\begin{prop} \label{FidelityInequality}
For $\mathrm{dim}(\mathcal{H}) < \infty$, let $T_A^\downarrow, T_B^\downarrow : \trcl(\mathcal{H}) \rightarrow \trcl(\mathcal{H})$ be quantum operations and let $\mathcal{V}$ be a subspace of $\mathcal{H}$ such that $T_A^\downarrow\vert_{\trcl(\mathcal{V)}} = T_B^\downarrow\vert_{\trcl(\mathcal{V)}}$ and $T_A^\downarrow\vert_{\trcl(\mathcal{V)}}$ is trace-preserving. Then
\begin{align*}
    \sqrt{F}(T^\downarrow_A(\rho), T^\downarrow_B(\sigma)) \geq \sqrt{F}(\rho, \sigma) - 2 \sqrt{F}(P^\bot \rho P^\bot, P^\bot \sigma P^\bot),
\end{align*}
for all $\rho, \sigma \geq 0$, where $P^\bot$ is the orthogonal projection onto $\mathcal{V}^\bot$.
\end{prop}

\newpage 
\section{The models} \label{ModelSection}

In this section, we propose two different, but in the end largely equivalent models that generalize the notion \sq{interaction-free} measurement to quantum channels. Since the sequential scheme, given in Figure \ref{GeneralDiscriminationProtocol}, is the most general causally ordered strategy allowed in quantum theory \cite{ChiribellaQunatumNetwork}, it suffices to define our notions for this kind of strategy. In both models, we assume the validity of Figure \ref{GeneralDiscriminationProtocol}. That is, we assume that the unknown channel $T$ does not change during the execution of the discrimination strategy - the Markovianity assumption. This is a relatively weak assumption, since we are in control of the duration between the individual channel invocations. This section consists of four subsections. In the first two subsections we derive our two models. The third subsection summarizes the former two by properly defining the quantities of merit and thereby setting the stage for a rigorous analysis in the later sections. In the fourth subsection, we compare the two models by deriving some elementary properties, which will be used later on.    

\subsection{The \sq{interaction} model}

In our first model, we interpret the term \sq{interaction-free} in an information theoretic way. That is, we imagine a Daemon sitting in the box trying to figure out, if we interacted with the interior of the box. In more technical terms, this means that the Daemon has full access to the output of the conjugate channel. Since the our task would be trivially infeasible otherwise, there must be a way not to interact with the box. Therefore, we only consider channels with vacuum. That is, we assume that for all channels under consideration there exists a distinguished pure state, the \textit{vacuum state}, $\ket{v}\bra{v}$. This state is assumed to have the following two important properties: First, if the vacuum state is sent through the channel, then the Daemon concludes that no interaction has occurred. Second, we assume that the channels under considerations map the vacuum state to a pure state. This assumption is physically reasonable as it means that the state of the probe system does not become entangled with the Daemon's system. If on the contrary, the probe system would become entangled with the Daemon's system, then there must have been an interaction and the term \sq{interaction-free} measurement would be inappropriate. We should mention, however, that the transmission model, which we are going to describe in the next section, does not use the \sq{\sq{vacuum maps to pure state}} assumption. This comes at the cost that the transmission functional is no longer a property of a channel (as the \sq{interaction} functional will turn out to be) but rather an object that has to be modeled separately. Together these two assumptions yield the definition of a channel with vacuum (Definition \ref{ChannelWithVacuumDef}). 
For a given channel $T$ with vacuum $v \in \mathcal{H}_I$ and an $N$-step discrimination strategy $D = (\mathcal{H}_I, \mathcal{H}_Z, \mathcal{H}_i, \mathcal{H}_o, s_0, \Lambda)$, we want to define the \textit{\sq{interaction} probability} $P_I^T(D)$ as the probability that the Daemon in the box encounters that, during the execution of $D$, something else than the vacuum state was sent through the channel. To define this probability, we need to specify how the Daemon can obtain information about what was sent through the channel. 
\begin{figure}[htbp]
    \centering
    \includegraphics[width=.9\textwidth]{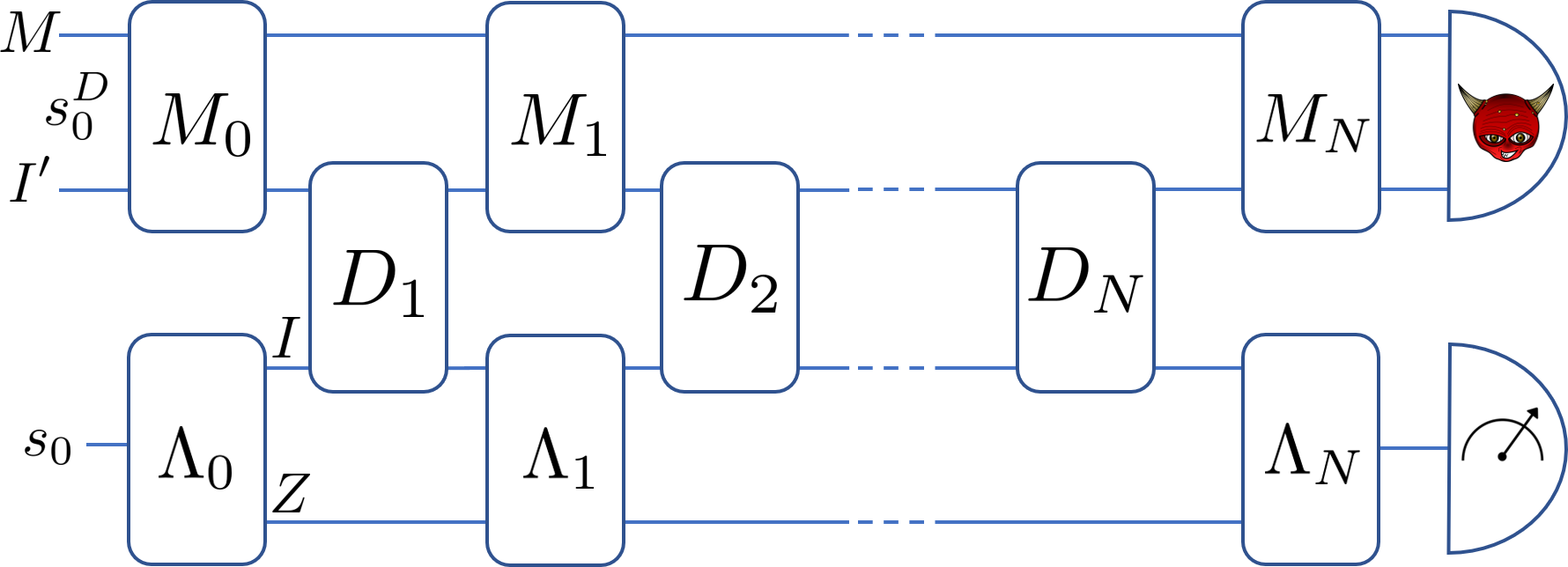}
    \caption{General Scenario}
    \label{Zenon1Picture}
\end{figure}A natural way to model this is by assuming that for each of the $N$ channel-uses (indexed by $n$) in the discrimination strategy, the Daemon is allowed to implement the channel $T$ via a channel $D_n : \trcl(\mathcal{H}_{I^\prime}\otimes \mathcal{H}_I) \rightarrow \trcl(\mathcal{H}_{I^\prime}\otimes \mathcal{H}_I)$, where $\mathcal{H}_{I^\prime}$ is the Hilbert space associated to a system $I^\prime$, which the Deamon controls. We further allow the Daemon to keep an arbitrarily large memory system $M$ (with Hilbert space $\mathcal{H}_M$) which he can manipulate freely (i.e., he can choose the channels $M_n$, defined below). The most general (causally ordered) scheme that can be obtained from the above description is depicted in Figure \ref{Zenon1Picture}. Mathematically, the Daemon's strategy is completely determined by an initial state $s_0^D \in \states(\mathcal{H}_M \otimes \mathcal{H}_{I^\prime})$ and channels $M_0, M_1, \dots, M_N : \trcl(\mathcal{H}_M \otimes \mathcal{H}_{I^\prime}) \rightarrow \trcl(\mathcal{H}_M \otimes \mathcal{H}_{I^\prime})$ and $D_1, D_2, \dots ,D_N : \trcl(\mathcal{H}_{I^\prime} \otimes \mathcal{H}_I) \rightarrow \trcl(\mathcal{H}_{I^\prime} \otimes \mathcal{H}_I)$. Given that data, the scheme in Figure \ref{Zenon1Picture} produces the output (or final) state $\rho_{F} \in \states(\mathcal{H}_M \otimes \mathcal{H}_{I^\prime} \otimes\mathcal{H}_o)$, defined by
\begin{align*} 
    \rho_{F} := (M_N \otimes \Lambda_N) (\idop \otimes D_N \otimes \idop) \dots  (M_1 \otimes \Lambda_1)  (\idop \otimes D_1 \otimes \idop)(M_0 \otimes \Lambda_0)(\chi_0), 
\end{align*}
where $\chi_0 := s_0^D \otimes s_0$. In the end, the Daemon will measure his system ($M+I^\prime$) and decide, based on the measurement outcome, if an interaction has occurred. The \sq{interaction} probability is then the probability that he detects such an interaction, if he chooses his strategy optimally within the given constraints. \\  
Before we can analyze what the Deamon's optimal strategy is, we still need to cast the assumptions that $D_n$ implements $T$, and that $T$ must be independent of the Deamon's strategy (Markovianity) into a mathematical form. Precisely, we assume that $D_n$ must be such that if the Daemon's system ($I^\prime$) and $I$ are uncorrelated, then the action on the system $I$ must be independent of the state of the system $I^\prime$. In formulas: We assume that
\begin{align} \label{SemiCausalCondition}
    \ptr{I}{D_n(\rho_{I^\prime}\otimes \rho_{I})} = T(\rho_I), \text{ for all } \rho_{I^\prime} \in \states(\mathcal{H}_{I^\prime}) \text{ and } \rho_I \in \states(\mathcal{H}_I). 
\end{align}
We now note that \eqref{SemiCausalCondition} is exactly the definition of a semi-causal channel, as introduced in \cite{CausalLocalizableOriginalPaper}. A structure theorem by Eggeling, Schlingemann
and Werner \cite{Eggeling_2002} tells us that semi-causal channels are semi-localizable. That is, $D_n$ can be written in the form
\begin{align*} 
    D_n(\rho_{I^\prime I}) = \ptr{E_n}{ (X_n \otimes \idop_I) (\idop_{I^\prime} \otimes \hat{V}_n)(\rho_{I^\prime I})},
\end{align*}
where $\hat{V}_n : \trcl(\mathcal{H}_I) \rightarrow \trcl(\mathcal{H}_{E_n} \otimes \mathcal{H}_I),$ defined by $\hat{V}_n(\cdot) = V_n \cdot V_n^\dagger$ is the quantum channel associated with a Stinespring isometry $V_n: \mathcal{H}_I \rightarrow \mathcal{H}_{E_n} \otimes \mathcal{H}_I$ of $T$ and $X_n : \trcl(\mathcal{H}_{I^\prime} \otimes \mathcal{H}_{E_n}) \rightarrow \trcl(\mathcal{H}_{I^\prime} \otimes \mathcal{H}_{E_n})$ is some channel. To proceed further in our search for the Daemon's optimal strategy, we make a few simplifying observations and definitions. Firstly, the unitary freedom in the Stinespring dilation $\hat{V}_n$ can be absorbed into the channel $X_i$. We can therefore assume, without loss of generality, that $\mathcal{H}_{E_1} = \mathcal{H}_{E_2} = \dots = \mathcal{H}_{E_N} =: \mathcal{H}_E$ and $\hat{V}_1 = \hat{V}_2 = \dots = \hat{V}_N =: \hat{V}$. Secondly, for $\rho \in \states(\mathcal{H}_M \otimes \mathcal{H}_{I^\prime} \otimes \mathcal{H}_I)$, we have
\begin{align*}
    (M_n \otimes \idop_{I})D_n(\rho) = \ptr{E_n}{( \left[M_n \otimes \idop_{E_n})(\idop_M \otimes X_n)\right]  \otimes \idop_I) (\idop_{M I^\prime} \otimes \hat{V}_n)(\rho)},
\end{align*}
which motivates the definition $\underline{X}_n := (M_n \otimes \idop_{E_n})(\idop_M \otimes X_n)$. In the following, we will adopt the convention that if some channel acts trivially on a tensor factor (i.e. as the identity), then we omit these tensor factors in the notation (e.g., $\underline{X}_i \otimes \idop_I$ becomes just $\underline{X}_i$). With the newly introduced notation, it follows from the definition of $\sigma_F$ that the state the Daemon obtains is 
\begin{align*}
    \ptr{IZ}{\rho_F} = \mathrm{tr}_{I Z} \, \Lambda_N \mathrm{tr}_{E_N} \underline{X}_N \hat{V}_N  \Lambda_{N-1} \mathrm{tr}_{E_{N-1}} \underline{X}_{N-1} \dots  \Lambda_{1} \mathrm{tr}_{E_1} \underline{X}_{1}\hat{V}_{1} M_0\Lambda_0 (\chi_0).
\end{align*}
We can commute the $\underline{X}_i$s and $\mathrm{tr}_{E_i}$s to the left. Thus, upon defining the channel $\Gamma : \trcl(\mathcal{H}_{E_N}\otimes\mathcal{H}_{E_{N-1}}\otimes \dots \otimes \mathcal{H}_{E_1}) \rightarrow \trcl(\mathcal{H}_{M}\otimes\mathcal{H}_{I^\prime})$ by
\begin{align*}
    \Gamma(\rho) = \mathrm{tr}_{E_N} \underline{X}_N \mathrm{tr}_{E_{N-1}}\underline{X}_{N-1} \dots \mathrm{tr}_{E_1 }\underline{X}_{1}M_0(s_0^D\otimes \rho),
\end{align*}
we have
\begin{align*}
    \ptr{IZ}{\rho_F} = \Gamma(\mathrm{tr}_{I Z}\, \Lambda_N \hat{V}_N \Lambda_{N-1} \hat{V}_{N-1} \dots  \Lambda_{1} \hat{V}_{1}\Lambda_0(s_0)). 
\end{align*}
To decide if the channel was ever applied to a state different from the vacuum state, the Daemon measures his state with a two valued POVM, $\{Q_1, Q_2\}$.
By convention, he will conclude that an interaction occurred (something else than the vacuum was sent through), if the event corresponding to $Q_2$ occurs. If the state sent through the channel is always the vacuum state, then the Daemon's final state is
\begin{align*}
    \Gamma(\left[\ptr{I}{V\ket{v}\bra{v}V^\dagger}\right]^{\otimes N}),
\end{align*}
where the tensor power is in the space $\mathcal{H}_{E_N} \otimes \mathcal{H}_{E_{N-1}} \otimes \dots \otimes \mathcal{H}_{E_1}$. Since the Daemon must not report an interaction, if the state was always the vacuum state, we demand
\begin{align*}
0 = \tr{Q_2 \Gamma(\left[\ptr{I}{V\ket{v}\bra{v} V^\dagger}\right]^{\otimes N})} = \tr{\Gamma^*(Q_2)  \left[\ptr{I}{V\ket{v}\bra{v} V^\dagger}\right]^{\otimes N}},
\end{align*}
where $\Gamma^*$ denotes the channel $\Gamma$ in the Heisenberg picture. Clearly, if $\Gamma^*(Q_2) = \idmat^{\otimes N} - P_v^{\otimes N}$, where $P_v$ is the orthogonal projection onto the support of $\ptr{I}{V\ket{v}\bra{v} V^\dagger}$, then this requirement is fulfilled. Since we want to choose the optimal strategy the Daemon can pursue, we want to set $\Gamma^*(Q_2) := \idmat^{\otimes N} - P_v^{\otimes N}$. We can always choose $\Gamma$ and $Q_2$ to satisfy the last equation, because this corresponds to the strategy where the Daemon simply stores all the states he obtains from the Stinespring dilation in each round. This justifies the graphical representation in Figure \ref{Zenon2Picture}. Since we defined the \sq{interaction} probability to be the probability that the Daemon concludes that an interaction occurred (if he acts optimally), we have 
\begin{align} \label{PiOpendef}
    P_I^T(D) := \tr{(\idmat^{\otimes N} - P_v^{\otimes N}) \; \mathrm{tr}_{I Z} \hat{V}_N \Lambda_{N-1} \hat{V}_{N-1} \dots  \Lambda_{1} \hat{V}_{1}(\rho_0^T)},
\end{align}
where $\rho_0^T := \Lambda_0(s_0)$ is the first intermediate state.
We remark that the definition of $P_I^T(D)$ does not depend on the particular choice of the Stinespring dilation, since the unitary freedom in the Stinespring isometries is compensated by the equal and opposite freedom in $P_v$. 
\begin{figure}[htbp]
    \centering
    \includegraphics[width=.9\textwidth, height = 3.5cm]{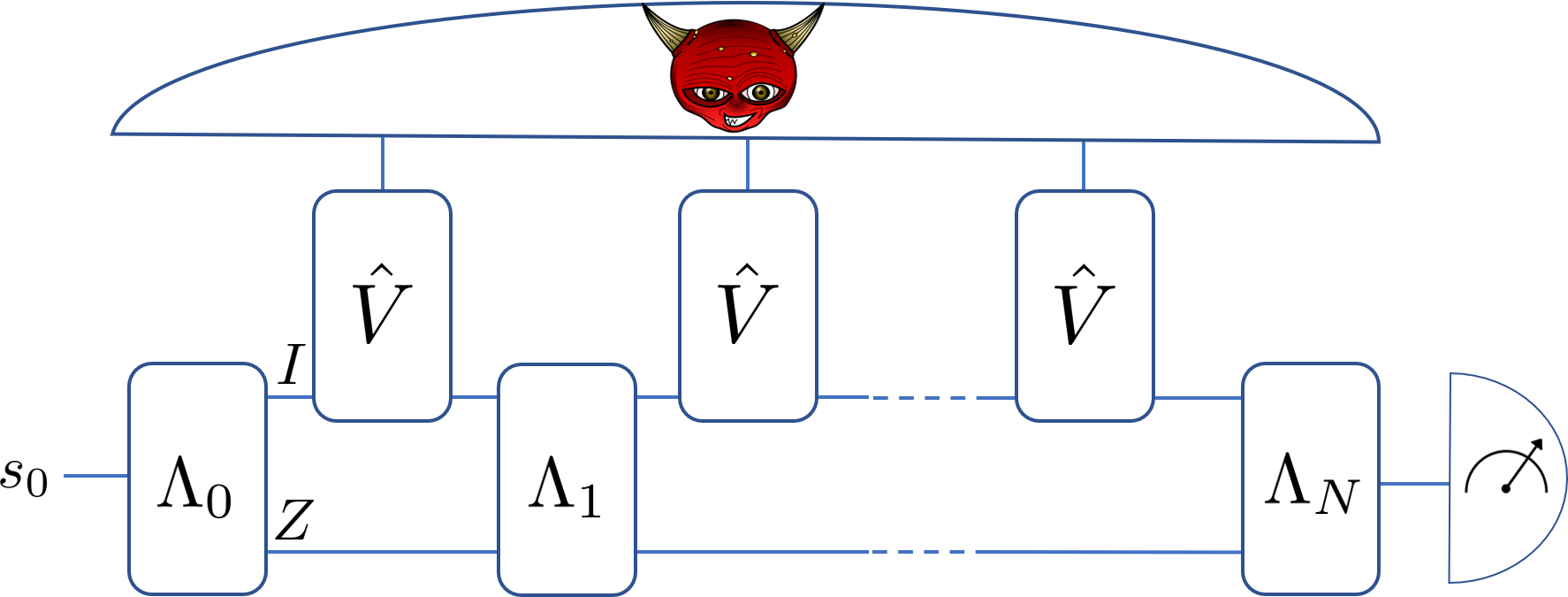}
    \caption{Scenario, when the Daemon's strategy is optimal.}
    \label{Zenon2Picture}
\end{figure}
We can simplify this expression a bit. We define $P_v^\bot := \idmat - P_v$ and note that
\begin{align*}
    \idmat^{\otimes N} - P_v^{\otimes N} &= \sum_{n = 0}^{N-1} \idmat^{\otimes N-n-1} \otimes P_v^\bot \otimes P_v^{\otimes n}\\
    P_v^\bot \otimes P_v^{\otimes K} &= \prod_{j = 0}^{K-1} P_v^\bot \otimes \idmat^{\otimes j} \otimes P_v \otimes \idmat^{\otimes  K - j - 1}.
\end{align*}
Using these two expressions and, excessively, that $\Lambda_n$ is trace-preserving, we obtain our final version for $P_I^T(D)$,
\begin{align*}
    P_I^T(D) &= \sum_{n=0}^{N-1} \tr{\idmat^{\otimes N-n-1} \otimes P_v^\bot \otimes P_v^{\otimes n} \; \mathrm{tr}_{I Z} \hat{V}_N \Lambda_{N-1} \hat{V}_{N-1} \dots  \Lambda_{1} \hat{V}_{1}(\rho_0^T)} \\
    &= \sum_{n=0}^{N-1} \tr{P_v^\bot \otimes P_v^{\otimes n} \; \mathrm{tr}_{I Z}  \hat{V}_{n+1} \Lambda_{n} \hat{V}_{n} \dots  \Lambda_{1} \hat{V}_{1}(\rho_0^T)}\\
    &= \sum_{n=0}^{N-1} \tr{P_v^\bot \mathrm{tr}_{I Z}( \hat{V}_{n+1} (  \Lambda_{n} ( \mathrm{tr}_{E_i} ( (P_v\otimes \idmat) \hat{V}_{n}( .. \mathrm{tr}_{E_1}( (P_v\otimes \idmat)\hat{V}_1(\rho_0^T) ..)} \\
    &= \sum_{n=0}^{N-1} \tr{P_v^\bot \mathrm{tr}_I \hat{V}(\mathrm{tr}_Z(\Lambda_i T^\downarrow \Lambda_{n-1} T^\downarrow \dots \Lambda_{1} T^\downarrow(\rho_0^T)))}\\ &= \sum_{n=0}^{N-1} \tr{P_v^\bot \mathrm{tr}_I \hat{V}(\ptr{Z}{\rho_n^{T^\downarrow}})}.
\end{align*}
In the second last line, we defined $T^\downarrow(\cdot) = \ptr{E}{(P_v\otimes \idmat) V \cdot V^\dagger}$ and $\rho_n^{T^\downarrow}$ is determined by the intermediate state map. We have thus succeeded in our goal to define the \sq{interaction} probability.
\begin{rem}
It is immediate from \eqref{PiOpendef} that an alternative expression for $P_I^T(D)$ is given by
\begin{align*}
    P_I^T(D) = \tr{\rho^{T^\downarrow}_N}.
\end{align*}
There are two reasons to prefer the lengthy version derived above. Firstly, it makes the connection between the \sq{interaction} model and the transmission model (defined below) explicit and thus allows us to treat these points of view on an equal footing. Secondly, it suggests to approach the problem by looking at the inputs of the individual channel uses, which turns out to be fruitful. 
\end{rem}

\subsection{The transmission model}

In our second model, we think of an interaction as something that does damage to the system in the box. As a guiding example, we think of a biological system - say a body cell. For the sake of the argument, assume that we want to use high-energetic radiation (e.g. x-ray) to resolve the inner structure of the cell. Of course, radiation might damage the cell, which is usually undesirable. A reasonable measure for how much damage has been done to a cell seems to be the number of x-ray photons that were absorbed by the cell. In other words, the damage is quantified by the amount of energy that got \textit{transmitted} from the probe system (x-ray) to the interior of the box (biological cell). Furthermore, if the cell is exposed to radiation several times, then the damage measure should be the sum of the number of photons that were absorbed each time. 
Let us now abstract away from this example. Assume that the system in the box is modeled quantum mechanically on a Hilbert space $\mathcal{H}_E$ and that the probe system is modeled on $\mathcal{H}_I$. Assume that initially the system $E$ is in the state $\rho_E \in \states(\mathcal{H}_E)$. If we probe the system with a state $\rho_I \in \states(\mathcal{H}_I)$, then the combined evolution is described by a (not necessarily unitary) channel $U : \trcl(\mathcal{H}_E\otimes\mathcal{H}_I) \rightarrow \trcl(\mathcal{H}_E\otimes\mathcal{H}_I)$. Thus, the state of the combined system after the evolution is given by
\begin{align*}
    \rho^\prime_{E I} = U(\rho_E\otimes \rho_I) 
\end{align*}
Now assume that, in analogy to the number of absorbed photons in the example above, there is some physical quantity (an observable) that got transmitted from the probe system to the interior of the box by the above process, and that this quantity is related to the damage done to the object in the box. We further assume that the process above can only cause damage and cannot repair the system in the box. Thus, the observable must be a positive semi-definite operator $\Theta$ on the Hilbert space $\mathcal{H}_E$. Hence, for a single shot experiment, the important object is the positive linear functional  $\mathfrak{t} : \trcl(\mathcal{H}_I) \rightarrow \mathbb{C}$, defined by
\begin{align*}
    \mathfrak{t}(\rho_I) = \tr{\Theta \, \ptr{I}{U(\rho_E\otimes \rho_I)}}
\end{align*}
For a general $N$-step discrimination strategy $D$ (with intermediate state map $\rho$), we assume that the transmitted quantity is extensive. Since the state of the part of the probe system that interacts with the interior of the box in the $n$th step is given by $\ptr{Z}{\rho_n^T}$ ($T$ is the channel defined by $T(\rho_I) = \ptr{E}{U(\rho_E\otimes \rho_I)}$), a good definition for the \textit{total transmission} $\mathfrak{T}_T(D)$ is
\begin{align*} 
    \mathfrak{T}_T(D) := \sum_{n = 0}^{N-1} \mathfrak{t}_T\left(\ptr{Z}{\rho_n^{T}}\right)
\end{align*}
We raise this to a principle by assuming that for every channel $T$ we have a positive linear functional $\mathfrak{t}_T$, which we call the \textit{transmission functional}, that models the damage done to the object. The total transmission then plays the same role for the transmission model as the \sq{interaction} probability does for the \sq{interaction} model. 

\subsection{Formal definition}

We cast the principles developed in the last sections into formal definitions.

\begin{defn}[\sq{Interaction} functional] 
Let $T : \trcl(\mathcal{H}) \rightarrow \trcl(\mathcal{H})$ be a channel with vacuum $v \in \mathcal{H}$ and let $V : \mathcal{H} \rightarrow \mathcal{H}_E \otimes \mathcal{H}$ be any Stinespring isometry of $T$. The positive linear functional $\mathfrak{i}_T : \trcl(\mathcal{H}) \rightarrow \mathbb{C}$, defined by
\begin{align*}
    \mathfrak{i}_T(\cdot) := \tr{P_v^\bot \ptr{\mathcal{H}}{V \cdot V^\dagger}},
\end{align*}
is called the \textit{\sq{interaction} functional} of $T$, where $P_v^\bot$ is the orthogonal projection onto the kernel of $\ptr{\mathcal{H}}{V \ket{v}\bra{v} V^\dagger}$.
\end{defn}

\begin{defn}[\sq{Interaction} probability] \label{InteractionProbabilityDef}
Let $T : \trcl(\mathcal{H}) \rightarrow \trcl(\mathcal{H})$ be a channel with vacuum $v \in \mathcal{H}$ and let $D = (\mathcal{H}, \mathcal{H}_Z, \mathcal{H}_i, \mathcal{H}_o, s_0, \Lambda)$ be an $N$-step discrimination strategy. The \textit{\sq{interaction} probability} is defined by
\begin{align*}
    P_I^T(D) := \sum_{n = 0}^{N-1} \mathfrak{i}_T\left(\ptr{Z}{\rho_n^{T^\downarrow}}\right),
\end{align*}
where the quantum operation $T^\downarrow : \trcl(\mathcal{H}) \rightarrow \trcl(\mathcal{H})$ is defined by
\begin{align}\label{DefinitionDownarrowChannel}
    T^\downarrow(\cdot) = \ptr{E}{(P_v\otimes \idmat) V \cdot V^\dagger},
\end{align}
and where $V: \mathcal{H} \rightarrow \mathcal{H}_E \otimes \mathcal{H}$ is any Stinespring isometry of $T$ and $P_v$ is the orthogonal projection onto the support of $\ptr{\mathcal{H}}{V \ket{v}\bra{v} V^\dagger}$. 
\end{defn}
\begin{defn} [\sq{Interaction-free} discrimination] \label{Interaction-FreeDiscriminationFormalDef}
Let $v \in \mathcal{H}$ and $\mathcal{C}_A, \mathcal{C}_B \subseteq \blt(\trcl(\mathcal{H}))$ be two sets of channels such that for all $T \in \mathcal{C}_A \cup \mathcal{C}_B$, $T$ is a channel with vacuum $v$.
We say that $\mathcal{C}_A$ and $\mathcal{C}_B$ can be \textit{discriminated in an \sq{interaction-free} manner}, if for every $\epsilon, \delta > 0$ there exists an $N$-step discrimination strategy $D$ and a two valued POVM $\Pi$ such that 
\begin{align*}
    P_e(D, \Pi) < \epsilon \quad\text{ and }\quad  P_I^T(D) < \delta, 
\end{align*}
for all $T \in \mathcal{C}_A \cup \mathcal{C}_B$.
\end{defn}
\begin{defn}[Channel with transmission functional] \label{DefinitionChannelWithTransmissionFunctional}
A \textit{channel with transmission functional} $\mathfrak{t}_T$ is a channel $T : \trcl(\mathcal{H}) \rightarrow \trcl(\mathcal{H})$ together with a positive linear functional $\mathfrak{t}_T \in \left(\trcl(\mathcal{H})\right)^*$. We call $\mathfrak{t}_T$ the \textit{transmission functional}.  
\end{defn}
\begin{defn} [Total transmission]
Let $T : \trcl(\mathcal{H}) \rightarrow \trcl(\mathcal{H})$ be a channel with transmission functional $\mathfrak{t}_T$. For an $N$-step discrimination strategy $D = (\mathcal{H}, \mathcal{H}_Z, \mathcal{H}_i, \mathcal{H}_o, s_0, \Lambda)$, the \textit{total transmission} is defined by
\begin{align*}
    \mathfrak{T}_T(D) := \sum_{n = 0}^{N-1} \mathfrak{t}_T\left(\ptr{Z}{\rho_n^{T}}\right).
\end{align*}
\end{defn}

\begin{defn} [Transmission-free discrimination]
Let $\mathcal{C}_A, \mathcal{C}_B \subseteq \blt(\trcl(\mathcal{H}))$ be two sets of channels such that for all $T \in \mathcal{C}_A \cup \mathcal{C}_B$, $T$ is a channel with transmission functional $\mathfrak{t}_T$.
We say that $\mathcal{C}_A$ and $\mathcal{C}_B$ can be \textit{discriminated in a transmission-free manner}, if for every $\epsilon, \delta > 0$ there exists an $N$-step discrimination strategy $D$ and a two valued POVM $\Pi$ such that 
\begin{align*}
    P_e(D, \Pi) < \epsilon \quad \text{ and }\quad \mathfrak{T}_T(D) < \delta, 
\end{align*}
for all $T \in \mathcal{C}_A \cup \mathcal{C}_B$.
\end{defn}

\subsection{Comparison of the models and elementary properties}

In this section, we clarify the relation between the transmission model and the \sq{interaction} model. As a rule of thumb, the transmission model can be thought of as a generalization of the \sq{interaction} model. Since we admit arbitrary positive linear functionals as transmission functionals, we have a much greater flexibility at modeling. For example, one could decide that out of the two objects to be discriminated, it does not matter (or is even desirable) if the second one gets destroyed. We should therefore set the transmission functional of the second channel to zero. This is something that is not possible in the \sq{interaction} model. On the other hand, the advantage of the \sq{interaction} model is that the \sq{interaction} probability has a very clear interpretation and that the \sq{interaction} functional is an intrinsic property of the channel. For the relation between these models, we note the following lemma.    

\begin{lem} \label{ComparisionIFTF}
Let $T : \trcl(\mathcal{H}) \rightarrow \trcl(\mathcal{H})$ be a channel with vacuum $v \in \mathcal{H}$ and let $\mathfrak{i}_T$ be its \sq{interaction} functional. If we interpret $T$ as a channel with transmission functional $\mathfrak{i}_{T}$, then 
\begin{align*}
    P_I^{T}(D) \leq \mathfrak{T}_{T}(D)
\end{align*}
for all $N$-step discrimination strategies $D$. 
\end{lem}
\begin{proof}
Immediate from the definition, since (by induction) $\rho^{T^\downarrow}_i \leq \rho^{T}_i$.
\end{proof}

The insight that should be gained from this lemma is that if we want to prove that a certain discrimination task can be done in an \sq{interaction-free} or in a transmission-free manner, then it suffices to tackle the problem in the transmission model. Thus, the results in Section \ref{ConstructivePart} will be formulated in terms of the transmission model. On the other hand, if we want to prove a no-go theorem, then it is sufficient to work in the \sq{interaction} model. At this point, there is a little detail that should not be swept under the rug, which is that it is possible that certain discrimination tasks can be performed with less resources, if one works in the \sq{interaction} model and not in the transmission model. We will not investigate this possibility any further. We close this section by introducing the concept of a \textit{maximal vacuum subspace}.   

\begin{defn}[Maximal vacuum subspace] \label{MaximalVacuumSubspaceDef}
Let $T : \trcl(\mathcal{H}) \rightarrow \trcl(\mathcal{H})$ be a channel with vacuum $v \in \mathcal{H}$ and let $V : \mathcal{H} \rightarrow \mathcal{H}_E \otimes \mathcal{H}$ be any Stinespring isometry of $T$. The subspace $\mathcal{V}_T$ of $\mathcal{H}$, defined by\footnote{$V^{-1}[\cdot]$ denotes the preimage operation.} 
\begin{align*}
    \mathcal{V}_T := V^{-1} \left[ \mathrm{supp}(\ptr{\mathcal{H}}{ V \ket{v}\bra{v} V^\dagger}) \otimes \mathcal{H} \right],
\end{align*}
is called the \textit{maximal vacuum subspace} of $T$.  
\end{defn}

\begin{lem}[Properties of maximal vacuum subspaces] \label{MaximalVacuumSubspaceLemma}
For $\mathrm{dim}(\mathcal{H}) < \infty$, let $T : \trcl(\mathcal{H}) \rightarrow \trcl(\mathcal{H})$ be a channel with vacuum $v \in \mathcal{H}$. The maximal vacuum subspace $\mathcal{V}_T$ has the following properties:
\begin{enumerate}
    \item $v \in \mathcal{V}_T$. \label{prop1}
    \item $T$ is isometric on $\mathcal{V}_T$. \label{prop2}
    \item If $T$ is isometric on a subspace $\mathcal{V}^\prime \subseteq \mathcal{H}$, then either $\mathcal{V}_T \cap \mathcal{V}^\prime = \{0\}$ or $\mathcal{V}^\prime \subseteq \mathcal{V}_T$.  \label{prop3}
    \item $\mathcal{V}_T$ is the union of all subspaces that contain $v$ and on which $T$ is isometric. \label{prop4}
    \item There exists a constant $C_T > 0$ such that $\mathfrak{i}_T(\rho) \geq C_T \tr{P^\bot \rho}$ for all $\rho \geq 0$, where $P^\bot$ is the projection onto $\mathcal{V}_T^\bot$. \label{prop5}
    \item For all $\rho \geq 0$, we have $\mathfrak{i}_T(\rho) \leq \tr{P^\bot \rho}$, where $P^\bot$ is the projection onto $\mathcal{V}_T^\bot$. \label{prop6}
\end{enumerate}
\end{lem}
\begin{rem}
The claims 1-4 and 6 remain true if one lifts the assumption that $\mathcal{H}$ is finite-dimensional. Claim 5, however, would then be wrong. 
\end{rem}
\begin{proof}
We start with the following observation: Let $V : \mathcal{H} \rightarrow \mathcal{H}_E \otimes \mathcal{H}$ be any Stinespring isometry of $T$. Since $T(\ket{v}\bra{v})$ is pure, $Vv$ must be a tensor product. Thus there are two unit vectors $v^\prime \in \mathcal{H}$ and $e \in \mathcal{H}_E$ such that
\begin{align*}
    Vv = e\otimes v^\prime. 
\end{align*}
Hence, $\ptr{\mathcal{H}}{V\ket{v}\bra{v}V^\dagger} = \ket{e}\bra{e}$ and
\begin{align} \label{SupportEqualsE}
    \mathrm{supp}(\ptr{\mathcal{H}}{V\ket{v}\bra{v}V^\dagger}) = \mathrm{span}\{e\}.
\end{align}
\noindent
\ref{prop1}) Clearly, $V v \in \mathrm{supp}(\ptr{\mathcal{H}}{V\ket{v}\bra{v}V^\dagger}) \otimes \mathcal{H}$. Thus, $v \in V^{-1}\left[V v\right] \subseteq \mathcal{V_T}$.\\ 
\ref{prop2}) For $\phi \in \mathcal{V}_T$, we have $V\phi = e\otimes \psi_\phi$ for a uniquely defined $\psi_\phi \in \mathcal{H}$. We define $U : \mathcal{V}_T \rightarrow \mathcal{H}$ by $U\phi := \psi_\phi$. It is easy to check, that $U$ is an isometry and that $T(\ket{\phi}\bra{\phi}) = U\ket{\phi}\bra{\phi}U^\dagger$. Since this holds for all $\phi \in \mathcal{V}_T$, $T$ is isometric on  $\mathcal{V}_T$.\\
\ref{prop3}) Suppose that $T$ is isometric on $\mathcal{V}^\prime$, with isometry $U^\prime : \mathcal{V}^\prime \rightarrow \mathcal{H}$. If $\mathrm{dim}(\mathcal{V}^\prime) \leq 1$ then the claim is trivially true. So we can assume that $\mathrm{dim}(\mathcal{V}^\prime) \geq 2$. Let $v_1$ and $v_2$ be two orthogonal unit vectors in $\mathcal{V}^\prime$. By assumption,
\begin{align*}
    T(\ket{v_i}\bra{v_i}) = \ptr{E}{V\ket{v_i}\bra{v_i}V^\dagger} = U^\prime\ket{v_i}\bra{v_i}U^{\prime\dagger},
\end{align*}
for $i \in \{1, 2\}$. As $U^\prime\ket{v_i}\bra{v_i}U^{\prime\dagger}$ is pure, there exists a pair of unit vectors $e_1, e_2 \in \mathcal{H}_E$ such that $Vv_i = e_i \otimes U^\prime v_i$. By linearity, we have
\begin{align*}
    0 &= T(\ket{v_1 + v_2}\bra{v_1 + v_2}) - \ptr{E}{V\ket{v_1 + v_2}\bra{v_1 + v_2}V^{\dagger}} \\&= U^\prime\ket{v_1 + v_2}\bra{v_1 + v_2}U^{\prime\dagger} - U^\prime\ket{v_1}\bra{v_1}U^{\prime\dagger} \\&- \braket{e_2}{e_1} U^\prime\ket{v_1}\bra{v_2}U^{\prime\dagger} - \braket{e_1}{e_2} U^\prime\ket{v_2}\bra{v_1}U^{\prime\dagger} - U^\prime\ket{v_2}\bra{v_2}U^{\prime\dagger} \\
    &= (1 - \braket{e_2}{e_1}) U^\prime\ket{v_1}\bra{v_2}U^{\prime\dagger} + (1 - \braket{e_1}{e_2}) U^\prime\ket{v_2}\bra{v_1}U^{\prime\dagger}.
\end{align*}
This can only be true, if $\braket{e_1}{e_2} = 1$, which is true only if $e_1 = e_2$. Thus, by transitivity, there is a unit vector $e^\prime \in \mathcal{H}_E$ such that $V v^\prime = e^\prime \otimes U^\prime v^\prime$, for all $v^\prime \in \mathcal{V}^\prime$. With the definition of $U$ in the proof of \ref{prop2}, we also have $V\phi = e\otimes U\phi$, for all $\phi \in \mathcal{V}_T$. Assume that $\mathcal{V}_T \cap \mathcal{V}^\prime \neq \{0\}$. For a unit vector $\hat{v} \in \mathcal{V}_T \cap \mathcal{V}^\prime$, the Cauchy-Schwarz inequality yields
\begin{align*}
    1 &= \abs{\braket{\hat{v}}{\hat{v}}} = \abs{\braket{V\hat{v}}{V\hat{v}}} = \abs{\braket{e^\prime}{e}} \abs{\braket{U^\prime \hat{v}}{U \hat{v}}} \\&\leq \abs{\braket{e^\prime}{e}} \norm{U^\prime \hat{v}}\norm{U \hat{v}}
    = \abs{\braket{e^\prime}{e}} \leq \norm{e^\prime}\norm{e} = 1.
\end{align*}
Hence, the Cauchy-Schwarz inequality is satisfied with equality, which implies that the vectors $e$ and $e^\prime$ differ only by a phase factor. In particular, $\mathrm{span}\{e\} = \mathrm{span}\{e^\prime\}$. Using \eqref{SupportEqualsE}, we have for any $v^\prime \in \mathcal{V}^\prime$ that $Vv^\prime = e^\prime \otimes U^\prime v^\prime \in \mathrm{supp}(\ptr{\mathcal{H}}{ V \ket{v}\bra{v} V^\dagger}) \otimes \mathcal{H}$. Consequently, $v^\prime \in \mathcal{V}_T$. As $v^\prime$ was arbitrary, this proves $\mathcal{V}^\prime \subseteq \mathcal{V}_T$ as claimed.
\\\ref{prop4}) If an isometric subspace $\mathcal{V}^\prime$ contains $v$, then (by \ref{prop1}) the intersection with $\mathcal{V}_T$ is non-trivial. Thus (by \ref{prop3}) $\mathcal{V}^\prime$ is a subspace of $\mathcal{V}_T$. Hence $\mathcal{V}_T$ contains all isometric subspaces and the claim follows as (by \ref{prop2}) $\mathcal{V}_T$ is isometric itself.\\
The following consideration is needed in the proof of \ref{prop5} as well as in the proof of \ref{prop6}. We define the projections $\hat{P} := P_v\otimes \idmat$ and $\hat{P}^\bot := \idmat - \hat{P}$, where $P_v := \ket{v}\bra{v}$. We further denote by $P$, the orthogonal projection onto $\mathcal{V}_T$ and define $P^\bot := \idmat - P$. 
In the following let $\rho \geq 0$. By definition, we have
\begin{align*}
    \mathfrak{i}_T(\rho) &= \tr{P_v^\bot \ptr{\mathcal{H}}{V\rho V^\dagger}} \\
    &= \tr{\hat{P}^\bot V\rho V^\dagger} \\
    &= \tr{\hat{P}^\bot V P\rho P V^\dagger} + \tr{\hat{P}^\bot V P\rho P^\bot V^\dagger} 
    \\&+ \tr{\hat{P}^\bot V P^\bot\rho P V^\dagger} + \tr{\hat{P}^\bot V P^\bot \rho P^\bot V^\dagger}.
\end{align*}
By definition, if $\psi \in \mathcal{V}_T$ then $\hat{P}^\bot V \psi = 0$. Thus $\hat{P}^\bot V P = 0$ as an operator. Hence, all summands except the last one vanish. Thus, we have
\begin{align} \label{IFinequality}
    \mathfrak{i}_T(\rho) = \tr{\hat{P}^\bot V P^\bot \rho P^\bot V^\dagger} = \tr{V^\dagger\hat{P}^\bot V P^\bot \rho P^\bot}
\end{align}
We can now prove \ref{prop5}. To this end, note that if $\tr{P^\bot \rho P^\bot} = 0$, then  the claim follows trivially. Otherwise, $\frac{P^\bot \rho P^\bot }{\tr{P^\bot \rho P^\bot}}$ is a density matrix and the spectral theorem implies that
\begin{align*}
    \frac{P^\bot\rho P^\bot}{\tr{P^\bot \rho P^\bot}} = \sum_i p_i \ket{\psi^\bot_i}\bra{\psi^\bot_i},
\end{align*}
with $p_i \geq 0$, $\sum_i p_i = 1$ and $\psi_i^\bot \in \mathcal{V}_T^\bot$. By convexity, we have
\begin{align*}
    \tr{\hat{P}^\bot V P^\bot \rho P^\bot V^\dagger} &= \tr{P^\bot \rho} \tr{\hat{P}^\bot V \frac{P^\bot\rho P^\bot}{\tr{P^\bot \rho P^\bot}} V^\dagger}  \\
    &\geq \tr{P^\bot \rho} \inf_{\substack{\psi^\bot \in \mathcal{V}_T^\bot\\ \norm{\psi^\bot} = 1}} \tr{\hat{P}^\bot V \ket{\psi^\bot}\bra{\psi^\bot} V^\dagger}
\end{align*}
If the infimum is strictly positive, then this is the $C_T$, we are looking for. To see that this is indeed the case, note that the set $\{\psi^\bot \in \mathcal{V}_T^\bot \mid| \norm{\psi^\bot} = 1\}$ is compact. Thus the infimum is actually a minimum. Assume for the sake of contradiction that $\tr{\hat{P}^\bot V \ket{\psi^\bot}\bra{\psi^\bot} V^\dagger} = 0$, for some unit vector $\psi^\bot \in \mathcal{V}_T^\bot$. Then $\braket{\hat{P}^\bot V \psi^\bot}{\hat{P}^\bot V \psi^\bot} = 0$ and consequently $\hat{P}^\bot V \psi^\bot = 0$. Hence, $V\psi^\bot \in \mathrm{supp}(\ptr{\mathcal{H}}{ V \ket{v}\bra{v} V^\dagger}) \otimes \mathcal{H}$ and $\psi^\bot \in \mathcal{V}_T$. As this is a contradiction, the claim follows.\\
To prove \ref{prop6}, we use H\"older's inequality for Schatten norms. Applying this inequality to the RHS of \eqref{IFinequality} yields
\begin{align*}
    \mathfrak{i}_T(\rho) \leq \norm{V^\dagger\hat{P}^\bot V}_\infty \norm{P^\bot \rho P^\bot}_1 = \tr{P^\bot \rho}.
\end{align*}
The last equality follows, since $V^\dagger\hat{P}^\bot V$ is an orthogonal projection (and thus has norm 1) and since $P^\bot \rho P^\bot \geq 0$. This proves the claim.
\end{proof}
\begin{rem}
Since, by the previous theorem, every subspace that is isometric w.r.t. $T$ and contains the vacuum, is contained in $\mathcal{V}_T$, checking the conditions in Theorem \ref{MainResult} reduces to checking whether
\begin{align*}
    T_A\vert_{\blt(\mathcal{V}_{T_A})} \neq T_B\vert_{\blt(\mathcal{V}_{T_A})} \quad \mathrm{or} \quad T_A\vert_{\blt(\mathcal{V}_{T_B})} \neq T_B\vert_{\blt(\mathcal{V}_{T_B})}.
\end{align*}
This can be done efficiently, since $\mathcal{V}_{T_A}$ and $\mathcal{V}_{T_B}$ can be computed by simple linear algebraic methods.
\end{rem}

\newpage
\section{The discrimination protocol} \label{ConstructivePart}

The main goal of this section is to prove Theorem \ref{DiscriminationStroategyOutlineTheorem}. 
This is done in two steps. At first, we show how to discriminate between the identity channel and a compact set of channels, where some additional conditions are imposed on the channels under consideration. In particular, we obtain the following theorem.
\begin{thm} \label{MainDiscriminationTheoremFromIdentity}
For $\mathrm{dim}(\mathcal{H}) < \infty$, let $\mathcal{C} \subseteq \blt(\trcl(\mathcal{H}))$ be a closed set of channels and let $v \in \mathcal{H}$ be a unit vector such that for all $T \in \mathcal{C}$, the state $\ket{v}\bra{v}$ is the only state that is a fixed point of $T$.
Then there exists a constant $C$ and for every $N \in \mathbb{N}$ an $N$-step discrimination strategy $D$ and a two-valued POVM $\Pi$ such that 
\begin{align*}
    P_e(D, \Pi) \leq \frac{C}{N^2},
\end{align*}
where the discrimination error probability is w.r.t the sets $\{\idop\}$ and $\mathcal{C}$. \\Furthermore, if $T \in \mathcal{C}$ is a channel with transmission functional $\mathfrak{t}_{T}$ and $\mathfrak{t}_{T}(\ket{v}\bra{v}) = 0$, then the total transmission $\mathfrak{T}_T(D)$ is bounded by 
\begin{align*}
    \mathfrak{T}_T(D) \leq \frac{C\norm{\mathfrak{t}_T}}{N}.
\end{align*}
In particular, if $\mathfrak{t}_\idop = 0$ and for all $T \in \mathcal{C}$, $T$ is a channel with transmission functional $\mathfrak{t}_{T}$, with $\mathfrak{t}_{T}(\ket{v}\bra{v}) = 0$; and if $\sup_{T \in \mathcal{C}} \norm{\mathfrak{t}_T} < \infty$, then the sets $\{\idop\}$ and $\mathcal{C}$ can be discriminated in a transmission-free manner.
\end{thm}
\begin{proof} This statement is a direct consequence of Theorem \ref{FiniteDimKwiat} and the discussion in the paragraph "Description of the discrimination strategy". \end{proof}
\noindent The second step then is to show how to reduce the general case to Theorem \ref{MainDiscriminationTheoremFromIdentity}. This is the main content of Section \ref{ReductionProtocolSection}, in which we also prove Theorem \ref{DiscriminationStroategyOutlineTheorem}.

\subsection{Empty or not?}

In this part we study a special case of the general discrimination task. That is, we study the case where we want to discriminate between the identity channel (empty box) and a compact set of channels $\mathcal{C} \subseteq \blt(\trcl(\mathcal{H}))$, which does not contain the identity channel. We show that, under some conditions on the spectrum of the channels in $\mathcal{C}$ and on the transmission functionals, a Kwiat et al.-like strategy suffices to perform the task in a transmission-free manner, even if the underlying Hilbert space is infinite-dimensional. In the finite-dimensional case, our considerations reduce to Theorem \ref{MainDiscriminationTheoremFromIdentity}. Before we detail what we mean by a Kwiat et al.-like strategy, we give an overview of the additional conditions we impose on the channels in $\mathcal{C}$.
\paragraph{Outline of the assumptions}
Our first assumption is that there is a pure state $\ket{v}\bra{v} \in \states(\mathcal{H})$ (vacuum) that is a fixed point of all channels in $\mathcal{C}$ and that the transmission functionals satisfy $\mathfrak{t}_T(\ket{v}\bra{v}) = 0$ for all $T \in \mathcal{C}$. As a remark, note that if there were no state $\rho \in \states(\mathcal{H})$, with $\mathfrak{t}_T(\rho) = 0$ for all $T \in \mathcal{C}$, then, of course, the discrimination task is impossible. On the other hand, if there exists such a state $\rho$, then, by the spectral theorem and the linearity and positivity of $\mathfrak{t}_T$, there exists a pure state $\rho_v \in \states(\mathcal{H})$, with $\mathfrak{t}_T(\rho_v) = 0$ for all $T \in \mathcal{C}$. But then, if $\rho_v$ is not a fixed point of $T$, the discrimination task becomes trivial. Thus, assuming a pure fixed point for the current setting is not a strong assumption. 

Our second assumption is that all channels in $\mathcal{C}$ have a spectral gap. That is, if we exclude $1$ from the spectrum of $T$, then the remaining part must be contained in a disk of radius less than $1$ (remember that since $T$ is a channel, its spectral radius is $1$ and $1$ is part of the spectrum). In Remark \ref{SpectralConditionIsNecessary}, we show that the spectral gap assumption cannot be waived completely, if a Kwiat et al.-like protocol (defined below) should do the job. 

Our third assumption is that the spectral gap assumption is compatible in a certain sense with the discrimination strategy. Expression \eqref{SpectrumDoesNotMoveCondition} in the statement of Theorem  \ref{FundamentalTheoremPositivePart} makes this statement precise. A sufficient condition for the compatibility assumption to be fulfilled (given our second assumption) is that $1$ is a simple eigenvalue of every channel in $\mathcal{C}$. This is the content of Theorem \ref{GeneralTheoremCompactSet}. Furthermore, in the finite-dimensional case our second assumption is automatically fulfilled (given our first assumption), if $1$ is a simple eigenvalue of every channel in $\mathcal{C}$. This is the content of Theorem \ref{FiniteDimKwiat}.

Our fourth assumption concerns the relation between the channels in $\mathcal{C}$ and their associated transmission functionals. Note that the definition of a transmission functional (Definition \ref{DefinitionChannelWithTransmissionFunctional}) does not impose such a relation. For our current purpose, however, this is problematic since $\sup_{T \in \mathcal{C}} \norm{\mathfrak{t}_T}$ may be infinite. We will thus assume that $\sup_{T \in \mathcal{C}} \norm{\mathfrak{t}_T}$ is finite. This is a very mild assumption, since it is implied if $\mathfrak{t}_T$ depends continuously on $T$ (which is very reasonable on physical grounds). Furthermore, note that if $\mathfrak{t}_T$ is an \sq{interaction} functional, then, as a consequence of claim \ref{prop6} in Lemma \ref{MaximalVacuumSubspaceLemma}, we have $\sup_{T \in \mathcal{C}} \norm{\mathfrak{t}_T} \leq 1$. 

\paragraph{Description of the discrimination strategy}
The next step is to design a strategy that allows us to discriminate between the identity channel and $\mathcal{C}$. An important factor in designing a strategy is the amount of resources that are needed to implement it. To this end, we show that only a bare minimum is required. Let $H \in \blt(\mathcal{H})$ be a self-adjoint operator such that $v$ is not an eigenvector of $e^{-iH}$. In other words, we assume that $C_H := \abs{\braket{v}{e^{-iH} v}}$ is strictly less than $1$. Then our strategy is to repeat the $N$-step discrimination strategy, depicted in Figure \ref{Kwiat-likeStartegy}, a total of $K$ times.\begin{figure}[htbp]
    \centering
    \includegraphics[width=.9\textwidth]{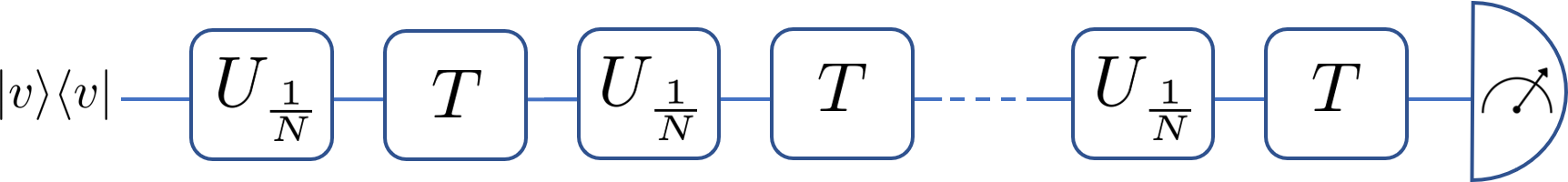}
    \caption{General form of a Kwiat et al.-like strategy} \label{Kwiat-likeStartegy}
\end{figure} More precisely, upon defining the 1-parameter family of channels $U_t : \trcl(\mathcal{H}) \rightarrow \trcl(\mathcal{H})$ by $U_t(\cdot) = e^{-iHt}\,\cdot \, e^{iHt}$, the discrimination strategy is given by the initial state $s_0 := \ket{v}\bra{v}$ and the set of channels $\Lambda$, with $\Lambda_i := U_\frac{1}{N}$ for $0 \leq i \leq N-1$ and $\Lambda_N := \idop$. After each execution of the discrimination strategy, we perform a measurement described by the two-valued POVM $\{P^\bot, \ket{v}\bra{v}\}$, where $P^\bot := \idmat - \ket{v}\bra{v}$. If all $K$ outcomes correspond to the second event, then we decide that the unknown channel is in $\mathcal{C}$ and otherwise we decide that the unknown channel is the identity. Of course, this protocol can be cast into the form of an $NK$-step discrimination strategy, by using an ancillary system and the principle of deferred measurement (see \cite{nielsen00}, p. 186). We call this strategy $D_{H,N,K}$. By Definition \ref{ErrorProbabilityDefn}, the error probability is then given by
\begin{align*}
P_e(D_{H, N, K}, \Pi) = \frac{1}{2} \left( \tr{\ket{v}\bra{v} \rho_N^\idop}^K + \sup_{T \in \mathcal{C}} \left\{ \tr{P^\bot\rho_N^T} \sum_{k = 0}^{K-1} \tr{\ket{v}\bra{v} \rho_N^T}^k \right\} \right),  
\end{align*}
where $\rho$ is the intermediate state map and where $\Pi$ denotes the measurement scheme described above. Explicitly, we have 
\begin{align*}
    \rho_N^\idop = U_{\frac{1}{N}}^N(\ket{v}\bra{v}) = e^{-iH}\ket{v}\bra{v}e^{iH} \quad \text{ and } \quad \rho_N^T = (T\circ U_{\frac{1}{N}})^N(\ket{v}\bra{v}).
\end{align*}
In general, this leads to the estimate
\begin{align} \label{EstimatedErrorProb}
    P_e(D_{H, N, K}, \Pi) \leq \frac{1}{2} \left(C_H^{2K} + K \sup_{T\in \mathcal{C}} \tr{P^\bot \rho_N^T} \right).
\end{align}
Now suppose that $P_M := \sup_{T\in \mathcal{C}} \tr{P^\bot \rho_N^T}$ approaches zero as $N \rightarrow \infty$ (we will show this below). Then, for given $\epsilon > 0$, we can choose $K := \left\lceil\frac{\ln(\epsilon)}{\ln(C_H)}\right\rceil$ and $N$ such that $K P_M < \epsilon$. It follows from \eqref{EstimatedErrorProb} that $P_e(D_{H, N, K}, \Pi) < \epsilon$. In other words, $P_e(D_{H, N, K}, \Pi)$ approaches zero if and only if $P_M$ does. Furthermore, for a channel $T \in \mathcal{C}$, the total transmission is given by
\begin{align*}
    \mathfrak{T}_T(D_{H, N, K}) = K \sum_{n = 0}^{N-1} \mathfrak{t}_T\left(\rho_n^{T}\right) = K\mathfrak{T}_T(D_{H, N, 1}).
\end{align*}
Thus, also $\mathfrak{T}_T(D_{H, N, K})$ approaches zero, if and only if $\mathfrak{T}_T(D_{H, N, 1})$ does.
In addition to that, we could always choose $H$ such that $\braket{v}{e^{-iH}v} = 0$. In that case, it suffices to set $K = 1$, which yields the simple expression
\begin{align*}
    P_e(D_{H, N, 1}, \Pi) = \frac{1}{2} \tr{P^\bot (T\circ U_{\frac{1}{N}})^N(\ket{v}\bra{v})}
\end{align*}
for the error probability. Hence, in order to find a strategy that discriminates between the identity channel and the set $\mathcal{C}$, we only need to show that the quantities $P_M$ and $\sup_{T\in \mathcal{C}}\mathfrak{T}_T(D_{H, N, 1})$ approach zero for $N \rightarrow \infty$. Moreover, since $\mathfrak{t}_{T}$ can be written in the form $\mathfrak{t}_{T}(\cdot) = \tr{\Theta_T \cdot}$, for some positive semi-definite operator $\Theta_T \in \blt(\mathcal{H})$ and since, by assumption $\mathfrak{t}_{T}(\ket{v}\bra{v}) = 0$, we can conclude that for $\rho \geq 0$, 
\begin{align*}
    \mathfrak{t}_{T}(\rho) \leq \norm{\mathfrak{t}_{T}} \tr{P^\bot \rho}.
\end{align*}
The important conclusion that we draw from the discussion above is that in order to prove Theorem \ref{MainDiscriminationTheoremFromIdentity}, it suffices to show (under the hypotheses of Theorem \ref{MainDiscriminationTheoremFromIdentity}) that for any self-adjoint $H \in \blt(\mathcal{H})$, there is a constant $C$ such that the inequalities 
\begin{align*}
    \tr{P^\bot \, (T \circ U_{\frac{1}{N}})^N(\ket{v}\bra{v})} \leq \frac{C}{N^2}\\
    \tr{P^\bot \, \sum_{n = 0}^{N-1} (U_{\frac{1}{N}} \circ T)^n(\ket{v}\bra{v})} \leq \frac{C}{N} 
\end{align*}
hold for all $N \in \mathbb{N}$. This is precisely the statement of Theorem \ref{FiniteDimKwiat}. Taking the validity of Theorem \ref{FiniteDimKwiat} for granted, we conclude that Theorem \ref{MainDiscriminationTheoremFromIdentity} holds.

\paragraph{Technical theorems} The remainder of this section is devoted to the proof of Theorem \ref{FiniteDimKwiat} and its infinite-dimensional versions. The following lemmas serve this purpose.

\begin{lem}[\cite{abramovich2002problems}, p. 202] \label{InvariantresolventLemma}
Let $T \in \blt(\mathcal{H})$, let $z \in \mathbb{C}$ be in the unbounded component of the resolvent $\rho(T)$, and let $X$ be a closed invariant subspace of $T$. Then $X$ is an invariant subspace of $(z-T)^{-1}$.
\end{lem}

\begin{lem} \label{SemisimplePeripheralSpectrum}
Let $T: \trcl(\mathcal{H}) \rightarrow \trcl(\mathcal{H})$ be a channel such that $1$ is in the discrete spectrum of $T$. Then, for any $n \in \mathbb{N}$ and any (rectifiable) path inside the resolvent set of $T$ that encloses $1$, and separates $1$ from $\sigma(T)\setminus\{1\}$, we have
\begin{align} \label{RieszProjectionIdentity}
  \frac{1}{2\pi i} \oint\limits_{\Gamma_1}\frac{z^n}{z-T}\,\mathrm{d}z = \frac{1}{2\pi i} \oint\limits_{\Gamma_1}\frac{1}{z-T}\,\mathrm{d}z, 
\end{align}
\end{lem}
\begin{proof}
See Appendix \ref{SemisimplicityLemmaAppendix}.
\end{proof}

\begin{lem} [Invariant subspace lemma] \label{InvariantSubsapceLemma}
Let $T : \trcl(\mathcal{H}) \rightarrow \trcl(\mathcal{H})$ be a channel, where $\mathcal{H}$ can be finite or infinite dimensional. Let $v \in \mathcal{H}$ be such that $\ket{v}\bra{v}$ is a fixed point of $T$ and set $V_v := \mathrm{span}\{v\}$. Then, the subspaces
\begin{align*}
   \blt_{v \bot} &:= \set{ \ket{v}\bra{\phi} }{ \phi \in V_v^\bot}\\
   \blt_{\bot v} &:= \set{ \ket{\phi}\bra{v} }{ \phi \in V_v^\bot}
\end{align*}
are invariant under $T$.
\end{lem}

\begin{proof}
 We prove that $\blt_{v \bot}$ is invariant. The invariance of $\blt_{\bot v}$ follows as $T$ is Hermiticity-preserving. Let $\{ K_i \}$ be a set of (non-zero) Kraus-operators of $T$. By assumption we have
\begin{align*}
\ket{v}\bra{v} &= T(\ket{v}\bra{v}) = \sum_i \tr{K_i^\dagger K_i} \frac{K_i \ket{v}\bra{v} K_i^\dagger}{\tr{K_i^\dagger K_i}},
\end{align*}
where the series converges in trace norm. 
As the pure state $\ket{v}\bra{v}$ is an extreme point of the closed and convex set of quantum states and the RHS is a convex combination of states, we must have that $K_i \ket{v}\bra{v} K_i^\dagger$ is proportional to $\ket{v}\bra{v}$. Henceforth, $v$ is an eigenvector of $K_i$ for all $i$. We denote the corresponding eigenvalue by $\lambda_i$. So for $\psi \in V_v^\bot$, we get
\begin{align*}
T(\ket{v}\bra{\psi}) = \sum_i K_i \ket{v} \bra{\psi}K_i^\dagger = \ket{v}\bra{\phi},
\end{align*}
where $\phi = \sum_i \overline{\lambda_i} K_i \,\psi$. As $T$ is trace-preserving, we have
\begin{align*}
0 = \tr{\ket{v}\bra{\phi}} = \tr{T(\ket{v}\bra{\psi})} = \tr{\ket{v}\bra{\phi}} = \braket{\phi}{v}.
\end{align*}
Hence, $\phi \in V_v^\bot$. This proves the claim.
\end{proof}
\noindent The following theorem is the main technical result. In fact, everything else in this section can (to some extent) be regarded as a corollary to this theorem.
\begin{thm} \label{FundamentalTheoremPositivePart}
Let $T: \trcl(\mathcal{H}) \rightarrow \trcl(\mathcal{H})$ be a channel such that $1$ is in the discrete spectrum of $T$, and let $v \in \mathcal{H}$ be a unit vector such that $\ket{v}\bra{v}$ is a fixed point of $T$. 
Furthermore, let $H \in \blt(\mathcal{H})$ be self-adjoint, $\tau > 0$ and $0 < \delta < 1$ such that
\begin{align} \label{SpectrumDoesNotMoveCondition}
    \sigma(U_t \circ T) \subseteq \mathbb{D}_{1-\delta}(0) \cup \{1\},
\end{align}
for $0 \leq t \leq \tau$, where $U_t:\trcl(\mathcal{H}) \rightarrow \trcl(\mathcal{H})$ is defined by $U_t(\cdot) := e^{-iHt} \cdot e^{iHt}$. Then, the inequalities
\begin{align}
    \tr{P^\bot \, (T \circ U_{\frac{1}{N}})^N(\ket{v}\bra{v})} \leq \frac{C}{N^2}\label{ProbabilityBound} \\ 
    \tr{P^\bot \, \sum_{n = 0}^{N-1} (U_{\frac{1}{N}} \circ T)^n(\ket{v}\bra{v})} \leq \frac{C}{N} \label{InfluenceBound}
\end{align}
hold for all $N \in \mathbb{N}$. Here, $P^\bot := \idmat - \ket{v}\bra{v}$ and
\begin{align*}
    C := \max\left\{\tau^{-2}, \;18 \delta^{-1} \norm{H}_{\blt(\mathcal{H})}^2 \max_{\substack{0 \leq t \leq \tau \\ z \in \Gamma}} \norm{(z-T)^{-1}} \norm{(z-U_tT)^{-1}}\right\} < \infty,
\end{align*}
where $\Gamma := \set{z \in \mathbb{C}}{ |z| = 1-\frac{\delta}{2}} \cup \set{z \in \mathbb{C}}{ |z-1| = \frac{\delta}{2}}$.
\end{thm}
\begin{proof}
We need to calculate the quantities \eqref{ProbabilityBound} and \eqref{InfluenceBound}. To do so, we employ the holomorphic functional calculus. For $0 \leq t \leq \tau$ and $n \in \mathbb{N}$, we have
\begin{align}
(U_tT)^n &= \frac{1}{2\pi i} \oint\limits_{|z - 1| = \frac{\delta}{2}}\frac{z^n}{z-U_tT}\,\mathrm{d}z + \frac{1}{2\pi i}\oint\limits_{|z| = 1-\frac{\delta}{2}}\frac{z^n}{z-U_tT}\,\mathrm{d}z \\
&= \frac{1}{2\pi i} \oint\limits_{|z - 1| = \frac{\delta}{2}}\frac{1}{z-U_tT}\,\mathrm{d}z + \frac{1}{2\pi i}\oint\limits_{|z| = 1-\frac{\delta}{2}}\frac{z^n}{z-U_tT}\,\mathrm{d}z, \label{RTNthPowerEq}
\end{align}
where we used Lemma \ref{SemisimplePeripheralSpectrum} to obtain the second line. Under the trace, we can (crudely) estimate this term as follows:
\begin{align}
\begin{split} \label{NThPowerUnderTheTrace}
   \abs{\tr{P^\bot (U_tT)^n(\ket{v}\bra{v})}} &\leq \frac{\delta}{2} \max_{|z-1| = \frac{\delta}{2}} \abs{\tr{P^\bot \frac{1}{z-U_tT}(\ket{v}\bra{v})}} \\&+ \left(1-\frac{\delta}{2}\right)^{n+1} \max_{|z| = 1-\frac{\delta}{2}} \abs{\tr{P^\bot \frac{1}{z-U_tT}(\ket{v}\bra{v})}} \\
   &\leq \max_{z \in \Gamma} \abs{\tr{P^\bot \frac{1}{z-U_tT}(\ket{v}\bra{v})}}
   \end{split}
\end{align}
In everything that follows, we will assume that $z \in \Gamma$. To proceed, we need two auxiliary calculations. 
First, we use the second resolvent identity (\cite{borthwick2020spectral}, p. 84) twice to obtain
\begin{align} \label{ResolventExpansion}
\begin{split}
\frac{1}{z-U_tT} &= \frac{1}{z-T} + \frac{1}{z-T}(U_t-\idop)\frac{T}{z-T} \\&+ \frac{1}{z-U_tT}(U_t-\idop)\frac{T}{z-T}(U_t-\idop)\frac{T}{z-T}.   
\end{split}
\end{align}
Second, an elementary application of Taylor's formula yields
\begin{align}
\norm{U_t - \idop} &\leq 2 \norm{H}_{\blt(\mathcal{H})} t \label{NormDUEstimate}\\
(U_t - \idop)(\rho) &= i[\rho, H]t + \mathfrak{U} t^2  \label{TaylorExp1}
\end{align}
with $||\mathfrak{U}|| \leq 2 \norm{H}_{\blt(\mathcal{H})}^2$.
When looking at \eqref{ResolventExpansion}, it is clear that the summands are of zeroth, first and second order in $t$, as $t \rightarrow 0$. The crucial step is to show that under the trace, the second term is $\mathcal{O}(t^2)$. Using \eqref{TaylorExp1}, we get
\begin{align}
\frac{1}{z-T}(U_t-\idop)\frac{T}{z-T}(\ket{v}\bra{v}) &= \frac{1}{z-1}\frac{1}{z-T}(U_t-\idop)(\ket{v}\bra{v}) \nonumber\\
\begin{split} \label{TermIsQuadraticEq2}
&= \frac{i t}{z-1} \frac{1}{z-T} (\ket{v}\bra{H v} - \ket{H v}\bra{v})
\\&+ \frac{t^2}{z-1} \frac{1}{z-T}(\mathfrak{U}(\ket{v}\bra{v})).  
\end{split}
\end{align}
It is easily verified, using the self-adjointness of $H$, that $\ket{v}\bra{H v} - \ket{H v}\bra{v} = \ket{v}\bra{\phi} - \ket{\phi}\bra{v}$, with $\phi := (H - \braket{v}{Hv})v$. Clearly, $\braket{\phi}{v} = 0$. Thus $\ket{\phi}\bra{v} \in \blt_{\bot v}$ and $\ket{v}\bra{\phi} \in \blt_{v \bot}$, where $\blt_{\bot v}$ and $\blt_{v \bot}$ are both invariant subspaces of $T$ (by Lemma \ref{InvariantSubsapceLemma}). As $z$ is in the unbounded component of the resolvent set of $T$, Lemma \ref{InvariantresolventLemma} implies that also $(z-T)^{-1}(\ket{\phi}\bra{v}) \in  \blt_{\bot v}$ and $(z-T)^{-1}(\ket{v}\bra{\phi}) \in  \blt_{v \bot}$. Thus, the first term in \eqref{TermIsQuadraticEq2} vanishes under the trace and we get
\begin{align}
    \abs{\tr{P^\bot \eqref{TermIsQuadraticEq2}}}  
    \leq t^2 \frac{2 \norm{H}_{\blt(\mathcal{H})}^2}{|z-1|} \norm{(z-T)^{-1}}. \label{EstimateResultInProof1}
\end{align}
So under the trace, this term is indeed quadratic in $t$. For the other two terms in \eqref{ResolventExpansion}, we have
\begin{align}
    \abs{\tr{P^\bot \frac{1}{z-T}(\ket{v}\bra{v})}} =  \frac{1}{|z-1|}\tr{P^\bot \ket{v}\bra{v})} = 0 \label{EstimateResultInProof2}
\end{align}
and
\begin{align}
    &\abs{\tr{P^\bot\;\frac{1}{z-U_tT}(U_t-\idop)\frac{T}{z-T}(U_t-\idop)\frac{T}{z-T}(\ket{v}\bra{v})}} \nonumber \\
    &\leq \frac{1}{|z-1|} \norm{(z-U_tT)^{-1}} \norm{U_t - \idop}^2 \norm{\frac{T}{z-T}} \nonumber \\
    &\leq t^2 \frac{4\norm{H}_{\blt(\mathcal{H})}^2}{|z-1|} \norm{(z-U_tT)^{-1}} \norm{(z-T)^{-1}}, \label{EstimateResultInProof3}
\end{align}
where we used the estimate \eqref{NormDUEstimate} and $\norm{T} = 1$ to obtain the last line.
We can now use the results obtained in \eqref{EstimateResultInProof1}, \eqref{EstimateResultInProof2} and \eqref{EstimateResultInProof3} to estimate the quantity of interest, \eqref{NThPowerUnderTheTrace}. We have
\begin{align}
    \eqref{NThPowerUnderTheTrace} &\leq 2 t^2  \norm{H}_{\blt(\mathcal{H})}^2 \max_{z \in \Gamma} \frac{\norm{(z-T)^{-1}}(1+2\norm{(z-U_tT)^{-1}})}{|z-1|}\nonumber \\
    &\leq t^2 \left(18 \delta^{-1} \norm{H}_{\blt(\mathcal{H})}^2 \max_{\substack{0 \leq t^\prime \leq \tau  \\ z \in \Gamma}} \norm{(z-T)^{-1}} \norm{(z-U_{t^\prime}T)^{-1}} \right)\nonumber \\&=: t^2 C_0, \label{BoundConstantInProof} 
\end{align}
To obtain the second estimate, we used that $\max_{z \in \Gamma} |z-1|^{-1} = 2\delta^{-1}$ and $\norm{(z-U_tT)^{-1}} \geq \norm{(z-U_tT)}^{-1} \geq (|z| + 1)^{-1} \geq \frac{2}{5}$.
Equation \eqref{BoundConstantInProof} is a bound for $t \leq \tau$. To prove the Theorem, we need a bound for all $t \geq 0$. To this end, we note that $\tr{P^\bot (U_tT)^n(\ket{v}\bra{v})} \leq 1$, since the expression represents a probability. We further define $C := \max(\tau^{-2}, C_0)$. If $t \leq \tau$, then by equation \eqref{BoundConstantInProof}, 
\begin{align*}
    \tr{P^\bot (U_tT)^n(\ket{v}\bra{v})} \leq t^2 C_0 \leq Ct^2.
\end{align*}
And if $t > \tau$, then
\begin{align*}
    \tr{P^\bot (U_tT)^n(\ket{v}\bra{v})} \leq 1 \leq \frac{t^2}{\tau^2} \leq Ct^2.
\end{align*}
Hence, 
\begin{align*}
    \tr{P^\bot (U_tT)^n(\ket{v}\bra{v})} \leq C t^2, 
\end{align*}
for all $t \geq 0$.
This is a bound independent of $n$. Inequality \eqref{InfluenceBound} is then easily obtained by setting $t := \frac{1}{N}$ and summing over all $n$, which yields an additional factor $N$. It remains to show inequality \eqref{ProbabilityBound}, in which $U_t$ and $T$ have switched order. Since $\ket{v}\bra{v}$ is a fixed point of $T$, we have $\tr{P^\bot (TU_t)^N(\ket{v}\bra{v})} = \tr{P^\bot T(U_tT)^N(\ket{v}\bra{v})}$. We set $\rho := (U_tT)^N(\ket{v}\bra{v})$ and $\phi := P^\bot \rho v$ and write
\begin{align*}
    \rho = \braket{v}{\rho v} \ket{v}\bra{v} + \ket{v}\bra{\phi} + \ket{\phi}\bra{v} + P^\bot \rho P^\bot.
\end{align*}
Clearly, $\ket{v}\bra{\phi} \in \blt_{v \bot}$ and $\ket{\phi}\bra{v} \in \blt_{\bot v}$. Hence, by Lemma \ref{InvariantSubsapceLemma}, we have $T(\ket{v}\bra{\phi}) \in \blt_{v \bot}$ and $T(\ket{\phi}\bra{v}) \in \blt_{\bot v}$. Thus,
\begin{align*}
    \tr{P^\bot T(\rho)} = \tr{P^\bot T(P^\bot \rho P^\bot)} \leq \tr{T(P^\bot \rho P^\bot)} = \tr{P^\bot \rho}.
\end{align*}
Hence,
\begin{align*}
    \tr{P^\bot (TU_t)^N(\ket{v}\bra{v})} \leq \frac{C}{N^2}.
\end{align*}
This finishes the proof.
\end{proof}

\begin{thm} \label{GeneralTheoremCompactSet}
Let $\mathcal{C} \subseteq \blt(\trcl(\mathcal{H}))$ be a compact set of channels, and let $v \in \mathcal{H}$ be a unit vector. Assume that
\begin{enumerate} 
    \item \label{deltaCondition} For all $T \in \mathcal{C}$, the quantity
\begin{align*}
    r_T := \sup_{z \in \sigma(T) \setminus \{1\}} |z|
\end{align*}
is strictly less than $1$. In other words, the spectral gap is non-zero.
\item For each $T \in \mathcal{C}$, the state $\ket{v}\bra{v}$ is a fixed point of $T$.
\item \label{SimplicityAssumption}For all $T \in \mathcal{C}$, the algebraic multiplicity\footnote{For an isolated point $\lambda \in \sigma(T)$, the algebraic multiplicity is the dimension of the range of the spectral projection.} of the isolated point $1 \in \sigma(T)$, is $1$. In other words, $1$ is a simple eigenvalue.
\end{enumerate}
Furthermore, let $H \in \blt(\mathcal{H})$ be self-adjoint and $U_t: \trcl(\mathcal{H}) \rightarrow \trcl(\mathcal{H})$ be defined by $U_t(\cdot) = e^{-iHt}\cdot e^{iHt}$. Then there exists a constant $C_\mathcal{C} < \infty$, such that 
\begin{align*}
    \tr{P^\bot \, (T \circ U_{\frac{1}{N}})^N(\ket{v}\bra{v})} &\leq \frac{C_\mathcal{C} \norm{H}^2_{\blt(\mathcal{H})}}{N^2} \\ 
    \tr{P^\bot \, \sum_{n = 0}^{N-1} (U_{\frac{1}{N}} \circ T)^n(\ket{v}\bra{v})} &\leq \frac{C_\mathcal{C} \norm{H}^2_{\blt(\mathcal{H})} }{N},
\end{align*}
for all $N \in \mathbb{N}$, where $P^\bot := \idmat - \ket{v}\bra{v}$.
\end{thm}
\begin{proof}
The basic strategy is to reduce the claim to an application of Theorem \ref{FundamentalTheoremPositivePart}. To this end, we basically need to show that conditions 1-3 imply that condition \eqref{SpectrumDoesNotMoveCondition} can be satisfied uniformly, i.e. that there exist $\tau > 0$ and $0 < \delta < 1$ such that \eqref{SpectrumDoesNotMoveCondition} is satisfied for all $T \in \mathcal{C}$.
The main tool to show this is the upper semi-continuity of the spectrum. To use that property, we import the following two theorems.
\begin{thm}[\cite{Kato}, p. 208]\label{SemicontinuitySpectrum}
For a Banach space $\mathcal{X}$, let $T, S \in \blt(\mathcal{X})$, and let $\Gamma$ be a compact subset of the resolvent set $\rho(T)$. \\If $\norm{T-S} < \min_{z \in \Gamma} \norm{(z-T)^{-1}}^{-1}$, then $\Gamma \subseteq \rho(S)$. Furthermore, for any open set $V \subseteq \mathbb{C}$, with $\sigma(T) \subset V$, there exists $\gamma > 0$, such that $\sigma(S) \subseteq V$ whenever $\norm{S-T} < \gamma$. 
\end{thm}
\begin{thm} [\cite{simon2015operator}, p. 67]\label{ProjectionTheorem}
For a Banach space $\mathcal{X}$, let $P, Q \in \blt(\mathcal{X})$ be bounded projections with $\norm{P - Q} < 1$. Then there exists an invertible operator $A \in \blt(\mathcal{X})$, such that $Q = APA^{-1}$. In particular $\mathrm{ran}(P)$ and $\mathrm{ran}(Q)$ are isomorphic.  
\end{thm}
\noindent 
To start, we show that not only $r_T < 1$ for all $T \in \mathcal{C}$, but that $\sup_{T\in \mathcal{C}} r_T < 1$.  
To this end, we show that the function $r : \mathcal{C} \rightarrow \mathbb{R}, \; T \mapsto r_T$ is upper semi-continuous. That is, we need to show that for every $T \in \mathcal{C}$ and every $\epsilon > 0$, there is a set $U \subseteq \mathcal{C}$, which is open in the relative topology on $\mathcal{C}$, such that $r_S \leq r_T + \epsilon$ for all $S \in U$. For fixed $T$ and $\epsilon > 0$, define $\epsilon^\prime := \min(\epsilon, \frac{1-r_{T}}{3})$ and the open set $V_{\epsilon^\prime} := B_{r_{T}+\epsilon^\prime}(0) \cup B_{\epsilon^\prime}(1) \subseteq \mathbb{C}$. By construction, $\sigma(T) \subseteq V_{\epsilon^\prime}$. Thus, Theorem \ref{SemicontinuitySpectrum} implies that there exists $\gamma > 0$ such that $\sigma(S) \subseteq V_{\epsilon^\prime}$, for all $S \in B_\gamma(T)$. Thus, for $S \in B_\gamma(T)$, the projection $P_S$ onto the spectral subspace associated with the spectral subset $\sigma(S) \cap B_{\epsilon^\prime}(1)$ is given by
\begin{align*}
    P_S := \frac{1}{2\pi i} \oint\limits_{|z - 1| = \frac{1-r_T}{2}}\frac{1}{z-T_n}\,\mathrm{d}z = P_T + \frac{1}{2\pi i} \oint\limits_{|z - 1| = \frac{1-r_T}{2}}\frac{1}{z-S}(S-T)\frac{1}{z-T}\,\mathrm{d}z, 
\end{align*}
where we used the second resolvent identity to obtain the last equation. A standard estimate yields
\begin{align*}
    \norm{P_S - P_T} \leq \frac{1-r_T}{2} \norm{S-T} \max_{|z-1| = \frac{1-r_T}{2}} \left\{\norm{(z-S)^{-1}}  \norm{(z-T)^{-1}} \right\}
\end{align*}
Since the set $S_0 := \overline{B_{\frac{\gamma}{2}}(T)} \cap \mathcal{C}$ is compact, the constant
\begin{align*}
    C_0 :=\max_{\substack{ |z-1| = \frac{1-r_T}{2}\\S \in S_0}} \left\{\norm{(z-S)^{-1}}  \norm{(z-T)^{-1}} \right\}
\end{align*}
is finite. We set $\gamma^\prime := \min\{\frac{\gamma}{2}, \frac{1}{(1-r_T) C_0}\}$ and $U := B_{\gamma^\prime}(T) \cap \mathcal{C}$. By construction, $U$ is open in the relative topology on $\mathcal{C}$ and we have $\sigma(S) \subseteq V_{\epsilon^\prime}$ and $\norm{P_S - P_T} \leq \frac{1}{2} < 1$, for all $S \in U$.
By assumption $\ref{SimplicityAssumption}$, $\mathrm{ran}(P_T)$ is 1-dimensional. Thus, by Theorem \ref{ProjectionTheorem}, also $\mathrm{ran}(P_S)$ is one-dimensional, for $S \in U$. Thus, there can be only one point in $\sigma(S) \cap B_{\epsilon^\prime}(1)$ and this point must be $1$, as $1$ is in the spectrum of every channel. Hence, for $S \in U$, we have $\sigma(S)\setminus\{ 1 \} \subseteq B_{r_T+\epsilon^\prime}(0)$. So $r(S) = r_{S} \leq r_{T} + \epsilon^\prime \leq r_{T} + \epsilon = r(T) + \epsilon$. In other words, $r$ is upper semi-continuous. 
The upper semi-continuous function $r$ assumes its maximum on the compact set $\mathcal{C}$. This maximum cannot be equal to $1$, as this would contradict assumption \ref{deltaCondition}. Thus $\max_{T\in \mathcal{C}} r_T  < 1$, as claimed. \\ In preparation for the application of Theorem \ref{FundamentalTheoremPositivePart}, we define the joint spectral gap
\begin{align}
    \delta_J := 1 - \max_{T \in \mathcal{C}} r(T).
\end{align}
We have $0 < \delta_J < 1$ and
\begin{align*}
    \sigma(T) \subseteq \mathbb{D}_{1-\delta_J}(0) \cup \{1\},
\end{align*}
for all $T \in \mathcal{C}$.
We define $\Gamma := \mathbb{D}_{1+\frac{\delta_J}{3}}(0)\setminus(B_{\frac{\delta_J}{3}}(1) \cup B_{1-\frac{2\delta_J}{3}}(0))$, which is a compact subset of $\rho(T)$ for all $T\in \mathcal{C}$, and we set
\begin{align*}
    \tau := \frac{1}{7\norm{H}_{\blt(\mathcal{H})}} \min_{\substack{T \in \mathcal{C} \\z\in \Gamma}} \norm{(z-T)^{-1}}^{-2},
\end{align*}
which is non-zero, as the minimization is over a strictly positive function on a compact set. For this particular choice of $\tau$, we will show that 
\begin{align*}
    \sigma(U_tT) \subseteq D_{1-\frac{2\delta_J}{3}}(0) \cup \{1\}
\end{align*}
for $0 \leq t \leq \tau$ and then use Theorem \ref{FundamentalTheoremPositivePart}. From now on, let $0 \leq t \leq \tau$ and $T\in \mathcal{C}$. Using the Taylor estimate \eqref{NormDUEstimate} and the definition of $\tau$, yields 
\begin{align}\label{T-UtEstimate}
    \norm{T - U_tT} &\leq \norm{U_t-\idop} \norm{T} \leq 2\norm{H}_{\blt(\mathcal{H})}t \nonumber\\
    &\leq \frac{2}{7} \min_{\substack{T \in \mathcal{C} \\z\in \Gamma}} \norm{(z-T)^{-1}}^{-2}.
\end{align}
This inequality has two important implications. 
First, for $z \in \Gamma$ we have $\norm{(z-T)^{-1}}^{-1} \leq \norm{z-T} \leq |z| + 1 \leq \frac{7}{3}$. Hence $\eqref{T-UtEstimate} < \min_{\substack{T \in \mathcal{C} \\z\in \Gamma}} \norm{(z-T)^{-1}}^{-1}$ and we can apply Theorem \ref{SemicontinuitySpectrum}, which tells us that $\Gamma \subseteq \rho(U_tT)$ for all $T\in \mathcal{C}$ and $0 \leq t \leq \tau$. Equivalently, 
\begin{align*}
    \sigma(U_tT) \subseteq \mathbb{D}_{1-\frac{2\delta_J}{3}}(0) \cup \mathbb{D}_{\frac{\delta_J}{3}}(1). 
\end{align*}
Thus we only have to show that $\sigma(U_tT) \cap \mathbb{D}_{\frac{\delta_J}{3}}(1) = \{1\}$.\\
Second, $\norm{(U_tT - T)(z-T)^{-1}} \leq \frac{2}{7}\min_{\substack{T \in \mathcal{C} \\z\in \Gamma}} \norm{(z-T)^{-1}}^{-1} \leq \frac{2}{3}$. Thus, the series
\begin{align*}
    \frac{1}{z-T} \sum_{k = 0}^\infty \left[(U_tT-T)(z-T)^{-1}\right]^k = (z-U_tT)^{-1}
\end{align*}
converges. A term by term estimate yields 
\begin{align} \label{UxEstimate}
    \norm{(z-U_tT)^{-1}} \leq 3\norm{(z-T)^{-1}}
\end{align}
Let $P_t := \frac{1}{2\pi i} \oint\limits_{|z - 1| = \frac{\delta_J}{3}}\frac{1}{z-U_tT} \,\mathrm{d}z$ be the spectral projection, then 
\begin{align*}
    \norm{P_t - P_0} &= \norm{ \frac{1}{2\pi i} \oint\limits_{|z - 1| = \frac{\delta_J}{3}}\frac{1}{z-U_tT} - \frac{1}{z-T}\,\mathrm{d}z } 
    \\ &\leq \frac{\delta_J}{3} \max_{|z - 1| = \frac{\delta_J}{3}} \norm{(z-U_tT)^{-1} - (z-T)^{-1}}
    \\ &= \frac{\delta_J}{3} \max_{|z - 1| = \frac{\delta_J}{3}} \norm{(z-U_tT)^{-1} (U_tT-T) (z-T)^{-1}}
    \\ &\leq \delta_J \norm{U_tT-T} \max_{z \in \Gamma} \norm{(z-T)^{-1}}^{2}
    \\ &\leq \frac{2\delta_J}{7} < 1, 
\end{align*}
where we used the second resolvent identity to obtain the third line, \eqref{UxEstimate} for the forth line and \eqref{T-UtEstimate} for the fifth line.
Hence, by Theorem \ref{ProjectionTheorem}, the dimension of $\mathrm{ran}(P_t)$ equals the dimension of $\mathrm{ran}(P_0)$, for all $0 \leq t \leq \tau$ and the latter dimension is $1$. Thus $\sigma(U_tT) \cap \mathbb{D}_{\frac{\delta_J}{3}}(1)$ contains exactly one point, which must be $1$, as $U_tT$ is a channel. In conclusion, we have
\begin{align*}
    \sigma(U_tT) \subseteq \mathbb{D}_{1 - \delta}(0) \cup \{1\},
\end{align*}
for all $T\in \mathcal{C}$ and $0 \leq t \leq \tau$, with $\delta := \frac{2\delta_J}{3}$. 
Finally, a direct application of Theorem \ref{FundamentalTheoremPositivePart} proves the claim. We can also get an explicit bound for $C_\mathcal{C}$. To this end, we need to bound the constant that appears in Theorem \ref{FundamentalTheoremPositivePart}. We have 
\begin{align*}
    \tau^{-2} = 49 \norm{H}_{\blt(\mathcal{H})}^2 \max_{\substack{T \in \mathcal{C}\\z\in \Gamma}} \norm{(z-T)^{-1}}^4
\end{align*}
and, by \eqref{UxEstimate}, the second term can be bounded by
\begin{align} \label{EstimateSecondConstantTerm}
 36 \delta_J^{-1} \norm{H}_{\blt(\mathcal{H})}^2 \max_{\substack{T \in \mathcal{C}\\ z \in \Gamma}} \norm{(z-T)^{-1}}^2.
\end{align}
Furthermore, by the spectral mapping theorem, the spectral radius of $(z-T)^{-1}$ is given by $(\inf_{s \in \sigma(T)} \norm{z-s})^{-1} = (\mathrm{dist}(z, \sigma(T)))^{-1}$. Since the norm of any operator is an upper bound for the spectral radius, we have
\begin{align*}
    \max_{\substack{T \in \mathcal{C}\\z\in \Gamma}} \norm{(z-T)^{-1}} \geq \max_{\substack{T \in \mathcal{C}\\z\in \Gamma}} \left\{ \mathrm{dist}(z, \sigma(T))^{-1} \right\} \geq 3\delta_J^{-1} \geq 3.
\end{align*}
By applying this bound to \eqref{EstimateSecondConstantTerm}, we see that $\tau^{-2} \geq \eqref{EstimateSecondConstantTerm}$. Thus, we can choose
\begin{align*}
    C_\mathcal{C} := 49 \max_{\substack{T\in \mathcal{C} \\ z \in \Gamma}} \norm{(z-T)^{-1}}^4 < \infty.
\end{align*}
\end{proof}
\begin{thm} \label{FiniteDimKwiat}
For $\mathrm{dim}(\mathcal{H}) < \infty$, let $\mathcal{C}$ be a closed set of channels $T : \trcl(\mathcal{H}) \rightarrow \trcl(\mathcal{H})$ and let $v \in \mathcal{H}$ be a unit vector such that for every $T \in \mathcal{C}$, the state $\ket{v}\bra{v}$ is the only state that is a fixed point of $T$.\\
Furthermore, let $H \in \blt(\mathcal{H})$ be self-adjoint and $U_t: \trcl(\mathcal{H}) \rightarrow \trcl(\mathcal{H})$ be defined by $U_t(\cdot) = e^{-iHt}\cdot e^{iHt}$. Then there exists a constant $C_\mathcal{C} < \infty$, such that for all $N \in \mathbb{N}$,
\begin{align}
    \tr{P^\bot \, (T \circ U_{\frac{1}{N}})^N(\ket{v}\bra{v})} &\leq \frac{C_\mathcal{C} \norm{H}^2_{\blt(\mathcal{H})} }{N^2} \\ 
    \tr{P^\bot \, \sum_{n = 0}^{N-1} (U_{\frac{1}{N}} \circ T)^n(\ket{v}\bra{v})} &\leq \frac{C_\mathcal{C} \norm{H}^2_{\blt(\mathcal{H})} }{N}, \label{FiniteInfucenceCound}
\end{align}
where $P^\bot := \idmat - \ket{v}\bra{v}$.
\end{thm}
\begin{proof}
The claim follows from Theorem \ref{GeneralTheoremCompactSet} and from results by Burgarth and Giovannetti \cite{Burgarth_2007}. In particular, in their terminology, a channel $T$ is called $\textit{ergodic}$, if there is a unique state that is a fixed point of $T$. And (according to Theorem 7 in \cite{Burgarth_2007}), $T$ is called \textit{mixing}, if $1$ is the only eigenvalue with modulus $1$ and the eigenvalue $1$ is simple. Thus, in particular, if $T$ is mixing, then the spectral gap is non-zero. Theorem 8 in \cite{Burgarth_2007} says that ergodic channels are mixing, if the unique state that is a fixed point is pure. By assumption, every $T \in \mathcal{C}$ is ergodic and the only state that is a fixed point is the pure state $\ket{v}\bra{v}$. Thus, all $T \in \mathcal{C}$ are mixing and the conditions in Theorem \ref{GeneralTheoremCompactSet} are automatically satisfied. This proves the claim.
\end{proof}
\begin{rem} \label{SpectralConditionIsNecessary}
In the previous theorem, it is important that $\ket{v}\bra{v}$ is the only state that is a fixed point. To demonstrate this, we define the Hamiltonian on a qubit system, $\mathcal{H}_Q := \mathrm{span}\{v, q_1\}$, as $H := \frac{\pi}{2} \sigma_y$, where $\sigma_y$ is the Pauli matrix\footnote{In coordinates, $\sigma_y := \begin{pmatrix}
0 & -i \\
i & 0
\end{pmatrix}$ and $e^{-i H t} = \begin{pmatrix}
\cos(\theta) & -\sin(\theta) \\
\sin(\theta) & \cos(\theta)
\end{pmatrix}$, with $\theta := \frac{\pi t}{2}$}. So, $U_t(\cdot) := e^{-iHt} \cdot e^{iHt}$. The channel 
$T : \trcl(\mathcal{H}_Q) \rightarrow \trcl(\mathcal{H}_Q)$ is then defined by
\begin{align*}
    T(\cdot) := \tr{\ket{v}\bra{v}\, \cdot\,}\ket{v}\bra{v} + \tr{\ket{q_1}\bra{q_1}\, \cdot\,}\ket{q_1}\bra{q_1},
\end{align*}
It is not hard to verify by induction that 
\begin{align*}
   (U_{\frac{1}{N}}\circ T)^n = U_{\frac{1}{N}}\left( \frac{1}{2}(1+\cos^n(2\theta)) \ket{v}\bra{v} + \frac{1}{2}(1-\cos^n(2\theta)) \ket{q_1}\bra{q_1}\right),  
\end{align*}
where $\theta := \frac{\pi}{2N}$.
The formula for the sum of the geometric progression yields
\begin{align*}
    \sum_{n = 0}^{N-1} (U_{\frac{1}{N}} \circ T)^n(\ket{v}\bra{v}) = U_{\frac{1}{N}}\left(\frac{1}{2}\left(N + \lambda\right)\ket{v}\bra{v} + \frac{1}{2}\left(N - \lambda\right) \ket{q_1}\bra{q_1} \right),
\end{align*}
with $\lambda := \frac{1-\cos^N(2\theta)}{2\sin^2(\theta)}$. It is an exercise in elementary calculus (or a query in your favourite computer algebra system) to show that 
\begin{align*}
    \lim_{N\rightarrow\infty}(N-\lambda) = \frac{\pi^2}{4}
\end{align*}
Since $U_{\frac{1}{N}} \rightarrow \idop$, when $N \rightarrow \infty$, it follows that the quantity on the RHS of \eqref{FiniteInfucenceCound} does not vanish as $N \rightarrow \infty$. In particular, our example shows that the Kwiat et al.-like protocol cannot be applied naively. Thus, the reduction process described in the next section is needed in some cases. 
\end{rem}
\begin{rem}
If the channel in Theorem \ref{FiniteDimKwiat} is a qubit channel ($\mathcal{H} = \mathrm{span}\{v, p\}$), then one can determine the precise asymptotics in a rather tedious calculation. We only state the result, which is that if $H := \frac{\pi}{2} \sigma_y$, then
\begin{align*}
    \lim_{N\rightarrow \infty} N^2 \tr{P^\bot \, (T \circ U_{\frac{1}{N}})^N(\ket{v}\bra{v})} &= \lim_{N\rightarrow \infty} N \tr{P^\bot \, \sum_{n = 0}^{N-1} (U_{\frac{1}{N}} \circ T)^n(\ket{v}\bra{v})} \\ &= \frac{\pi^2}{4}\frac{1- \abs{\tau_0}^2}{(1-\tau) \abs{1-\tau_0}^2},
\end{align*} 
where $\tau := \tr{P^\bot T(P^\bot)}$ and $\tau_0 := \tr{\ket{p}\bra{v} \,T(\ket{v}\bra{p})}$. \\This result contains as a special case the result for semi-transparent objects \cite{PhysRevA.74.054301, PhysRevA.96.062129}. 
\end{rem}
\begin{rem} It is a direct consequence of the results in the next section that the $N^{-1}$ form of the bound is optimal.  
\end{rem}

\subsection{The reduction protocol} \label{ReductionProtocolSection}
In this section, in which we will assume that all Hilbert spaces are finite-dimensional, we want to transform our given channel in such a way that the Kwiat et al.-like strategy, which was described in the previous section, can be applied. The general idea is that instead of inserting the unknown channel directly into the circuit of Figure \ref{Kwiat-likeStartegy}, we preprocess and postprocess the states that go in and out of the channel. 
\begin{figure}[htbp]
    \centering
    \includegraphics[width=.9\textwidth]{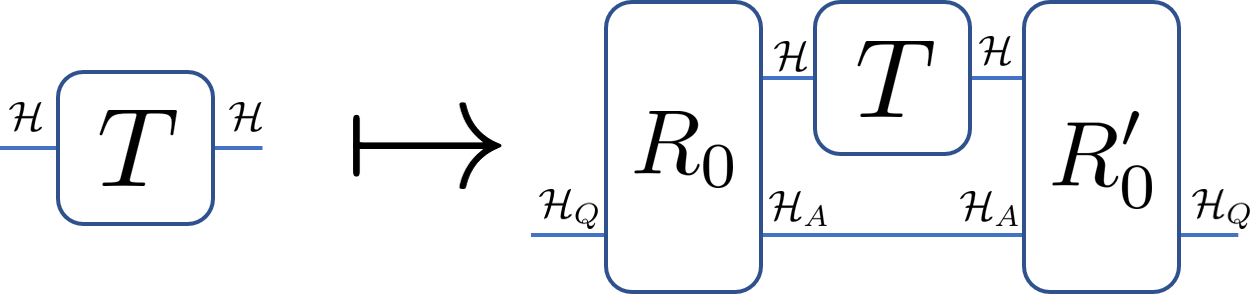}
    \caption{General transformation scheme: a superchannel} \label{SuperchannelFigure}
\end{figure}
In other words, we replace the channel $T$ in Figure \ref{Kwiat-likeStartegy} by the construction that is depicted on the RHS of Figure \ref{SuperchannelFigure}. In Figure \ref{SuperchannelFigure}, $\mathcal{H}_Q$ and $\mathcal{H}_A$ are Hilbert spaces and $R_0 : \trcl(\mathcal{H}_Q) \rightarrow \trcl(\mathcal{H}\otimes \mathcal{H}_A)$ and $R_0^\prime : \trcl(\mathcal{H}\otimes \mathcal{H}_A) \rightarrow \trcl(\mathcal{H}_Q)$ are channels. The resulting transformation can be viewed as a map $R : \blt(\trcl(\mathcal{H})) \rightarrow \blt(\trcl(\mathcal{H}_Q))$, defined by $R(T) := R_0^\prime (T\otimes \idop) R_0$. Maps of this kind are usually called superchannels \cite{Chiribella_2008}. Clearly, if $T$ is a channel with transmission functional $\mathfrak{t}_T$, then $R(T)$ is a channel with transmission functional $\mathfrak{t}_{R(T)} := \mathfrak{t}_{T} \circ \mathrm{tr}_A \circ R_0$. We say that the superchannel $R$ transforms the transmission functional $\mathfrak{t}_T$ to $\mathfrak{t}_{R(T)}$. For consistency reasons, we also remark the following: As is shown in \cite{Chiribella_2008}, for any superchannel $S : \blt(\trcl(\mathcal{H})) \rightarrow \blt(\trcl(\mathcal{H}_Q))$, there exists a Hilbert space $\mathcal{H}_{A^\prime}$ and channels $S_0 : \trcl(\mathcal{H}_Q) \rightarrow \trcl(\mathcal{H}\otimes \mathcal{H}_{A^\prime})$ and $S_0^\prime : \trcl(\mathcal{H}\otimes \mathcal{H}_{A^\prime}) \rightarrow \trcl(\mathcal{H}_Q)$ such that $S(T) = S_0^\prime (T\otimes \idop) S_0$, for all $T \in \blt(\trcl(\mathcal{H}))$. Of course, the choice of $\mathcal{H}_{A^\prime}$, $S_0$ and $S_0^\prime$ is not unique. The transformation of the transmission functional, however, is unique. To see this, assume that we apply $S$ to the map $T_B$, defined by $T_B(\cdot) = \tr{B\,\cdot\,} \rho_0$, where $\rho_0 \in \states(\mathcal{H})$ and $B\in \blt(\mathcal{H})$ are arbitrary. Since $S_0^\prime$ is trace-preserving, we have for $\sigma \in \blt(\mathcal{H}_Q)$, that $\tr{S(T)(\sigma)} = \tr{(T\otimes \idop)S_0(\sigma)} = \tr{B \ptr{A^\prime}{S_0(\sigma)}}$. Since $B$ and $\sigma$ were arbitrary, it follows that $\mathrm{tr}_{A^\prime}\circ S_0$ is independent of the choice of $\mathcal{H}_{A^\prime}$, $S_0$ and $S_0^\prime$. Hence the transformation of the transmission functional is independent of the particular implementation of a superchannel. Formally the replacement described above yields a transformation of the discrimination strategy. That is, given a discrimination strategy $D = (\mathcal{H}_Q, \mathcal{H}_Z, \mathcal{H}_i, \mathcal{H}_o, s_0, \Lambda)$, with $\Lambda = \{\Lambda_1, \Lambda_2, \dots, \Lambda_N\}$, then we obtain the transformed discrimination strategy $D^R := (\mathcal{H}, \mathcal{H}_A\otimes\mathcal{H}_Z, \mathcal{H}_i, \mathcal{H}_o, s_0, \Lambda_R)$, with $\Lambda^R_0 := (R_0 \otimes \idop_Z)\Lambda_0$, $\Lambda^R_N := \Lambda_N(R_0^\prime \otimes \idop_Z)$ and $\Lambda^R_n := (R_0 \otimes \idop_Z)\Lambda_n(R_0^\prime \otimes \idop_Z)$, for $1 \leq n \leq N-1$. 

The task of this section is to show the existence of a superchannel such that the general discrimination task reduces to the one described in the last section. It will be evident from the proof of the following theorem that such a superchannel can be implemented by using only one ancillary qubit and classical resources.  Furthermore, we will show in Remark \ref{OneNeedsAncillaryRemark} that in general the implementation of such a superchannel is impossible without using an ancillary qubit.
\begin{thm}[Reduction superchannel] \label{ReductionSuperchannelTheorem}
For $\mathrm{dim}(\mathcal{H}) < \infty$, let $T : \trcl(\mathcal{H}) \rightarrow \trcl(\mathcal{H})$ be a channel and let $\mathcal{V} \subseteq \mathcal{H}$ be a subspace such that $T$ is isometric on $\mathcal{V}$. Furthermore, let $v \in \mathcal{V}$ be a unit vector. Then there exists a 2-dimensional Hilbert space $\mathcal{H}_Q$, with orthonormal basis $\{q_0, q_1\}$ and a superchannel $R : \blt(\trcl(\mathcal{H})) \rightarrow  \blt(\trcl(\mathcal{H}_Q))$ with the following properties: 
\begin{enumerate}
    \item If $T^\prime \in \blt(\trcl(\mathcal{H}))$ satisfies $T\vert_{\trcl(\mathcal{V})} = T^\prime\vert_{\trcl(\mathcal{V})}$, then $R(T^\prime) = \idop$.
    \item If $T^\prime \in \blt(\trcl(\mathcal{H}))$ is a channel such that $T\vert_{\trcl(\mathcal{V})} \neq T^\prime\vert_{\trcl(\mathcal{V})}$, then the only state that is a fixed point of $R(T^\prime)$, is $\ket{q_0}\bra{q_0}$.
    \item If $T^\prime \in \blt(\trcl(\mathcal{H}))$ is a channel with transmission functional $\mathfrak{t}_{T^\prime}$ and $\mathfrak{t}_{T}(\ket{v}\bra{v}) = 0$, then the transformed transmission functional $\mathfrak{t}_{R(T^\prime)}$ is given by
    \begin{align*}
     \mathfrak{t}_{R(T^\prime)}(\cdot) =  \bigg\{\begin{array}{lr}
         \frac{1}{2} \mathfrak{t}_{T^\prime}(\frac{P^\bot}{d-1}) \tr{\ket{q_1}\bra{q_1}\,\cdot\,}  & \text{if } d > 1\\
       0 & \text{if } d = 1
        \end{array},
    \end{align*}
    where $d := \mathrm{dim}(\mathcal{V})$ and where $P^\bot$ denotes the orthogonal projection onto $\set{\psi \in \mathcal{V}}{\braket{\psi}{v} = 0}$.
\end{enumerate}
\end{thm}
\noindent Before we prove the theorem, let us explore its consequences. First, we establish the analog of Theorem \ref{DiscriminationStroategyOutlineTheorem} for the transmission functional model.
\begin{cor} \label{transmissionFreeDiscriminationCor}
For $\mathrm{dim}(\mathcal{H}) < \infty$, let $\mathcal{C}_A, \mathcal{C}_B \subseteq \blt(\trcl(\mathcal{H}))$ be two closed sets of channels. Furthermore, let $\mathcal{V}$ be a subspace of $\mathcal{H}$ and let $v \in \mathcal{V}$ be a unit vector such that
\begin{enumerate}
    \item For all $T \in \mathcal{C}_A \cup \mathcal{C}_B$, $T$ is a channel with transmission functional $\mathfrak{t}_{T}$.
    \item For all $T \in \mathcal{C}_A$, $T$ is isometric on $\mathcal{V}$.
    \item For all $T \in \mathcal{C}_A$, $\mathfrak{t}_{T}\vert_{\trcl(\mathcal{V})} = 0$.
    \item For all $T \in \mathcal{C}_B$, $\mathfrak{t}_{T}(\ket{v}\bra{v}) = 0$.
    \item $\sup_{T \in \mathcal{C}_B} \norm{\mathfrak{t}_T\vert_{\trcl(\mathcal{V})}} < \infty$
    \item The set $\mathcal{C}_A\vert_{\trcl(\mathcal{V})} := \set{T\vert_{\trcl(\mathcal{V})}}{T \in \mathcal{C}_A}$ contains exactly one element.
    \item $\mathcal{C}_A\vert_{\trcl(\mathcal{V})}$ and $\mathcal{C}_B\vert_{\trcl(\mathcal{V})} := \set{T\vert_{\trcl(\mathcal{V})}}{T \in \mathcal{C}_B}$ are disjoint.
\end{enumerate}
Then there exist a constant $C$ and for every $N \in \mathbb{N}$, an $N$-step discrimination strategy $D$ and a two-valued POVM $\Pi$ such that 
\begin{gather*}
    P_e(D, \Pi) \leq \frac{C}{N^2}\\
    \mathfrak{T}_{T_A}(D) = 0 \quad \text{ and } \quad \mathfrak{T}_{T_B}(D) \leq \frac{C}{N} 
\end{gather*}
for all $T_A \in \mathcal{C}_A$ and all $T_B \in \mathcal{C}_B$, where the discrimination error probability is w.r.t. the sets $\mathcal{C}_A$ and $\mathcal{C}_B$. Hence, the sets $\mathcal{C}_A$ and $\mathcal{C}_B$ can be discriminated in a transmission-free manner.
\end{cor}
\begin{proof}
We combine Theorem \ref{MainDiscriminationTheoremFromIdentity} and Theorem \ref{ReductionSuperchannelTheorem}. Fix some $T_A \in \mathcal{C}_A$. From Theorem \ref{ReductionSuperchannelTheorem} (with $T = T_A$), we obtain the map $R$, with the properties (1), (2) and (3). We want to apply Theorem \ref{MainDiscriminationTheoremFromIdentity} with $\mathcal{C} := R(\mathcal{C}_B)$. Since $\mathcal{C}_B$ is (as a closed subset of the compact set of channels) compact and $R$ is continuous, $\mathcal{C}$ is compact and hence closed. Furthermore, since by assumption 7, the sets $\mathcal{C}_A\vert_{\trcl(\mathcal{V})}$ and $\mathcal{C}_B\vert_{\trcl(\mathcal{V})}$ are disjoint, we have $T^\prime\vert_{\trcl(\mathcal{V)}} \neq T_A\vert_{\trcl(\mathcal{V})}$, for all $T^\prime \in \mathcal{C}_B$. Hence, property (2) implies that for all $T \in \mathcal{C}$, the state $\ket{q_0}\bra{q_0}$ is the only state that is a fixed point of $T$. In particular, $\idop \notin \mathcal{C}$. Furthermore, assumption 6 implies that $T^\prime\vert_{\trcl(\mathcal{V)}} = T_A\vert_{\trcl(\mathcal{V})}$, for all $T^\prime \in \mathcal{C}_A$. Hence, by property (1), $R(\mathcal{C_A}) = \{\idop\}$. Thus, Theorem \ref{MainDiscriminationTheoremFromIdentity} yields a discrimination strategy $\tilde{D}$ and a two-valued POVM such that $P_e(\tilde{D}, \Pi) \leq \tilde{C}N^{-2}$, for some constant $\tilde{C}$. By construction, $P_e(\tilde{D}, \Pi)$ is the discrimination probability w.r.t. the sets $\mathcal{C}$ and $\{\idop\}$, but since we have for $T^\prime \in \mathcal{C}_A \cup \mathcal{C}_B$ that $R(T^\prime) \in \{\idop\}$ iff $T^\prime \in \mathcal{C}_A$ and $R(T^\prime) \in \mathcal{C}$ iff $R(T^\prime) \in \mathcal{C}_B$, it follows that $P_e(\tilde{D}^R, \Pi) = P_e(\tilde{D}, \Pi)$, where $\tilde{D}^R$ is the transformed discrimination strategy, as defined in the main text. For $T^\prime \in \mathcal{C}_A$, condition 3 and property (3) imply that the transformed transmission functional $\mathfrak{t}_{R(T^\prime)} = 0$. Thus $\mathfrak{T}_{T^\prime}(\tilde{D}^R) = 0$. Furthermore, for $T^\prime \in \mathcal{C}_B$ with transmission functional $\mathfrak{t}_{T^\prime}$, property (3) implies that the norm of the transformed transmission functional satisfies $\norm{\mathfrak{t}_{R(T^\prime)}} = \frac{1}{2} \mathfrak{t}_{T^\prime}\left(\frac{P^\bot}{d-1}\right) \leq \frac{1}{2}\norm{\mathfrak{t}_{T^\prime}\vert_{\trcl(\mathcal{V})}}$. Since we have $\mathfrak{T}_{T^\prime}(\tilde{D}^R) = \mathfrak{T}_{R(T^\prime)}(\tilde{D})$, Theorem \ref{MainDiscriminationTheoremFromIdentity} implies that $\mathfrak{T}_{T^\prime}(D^R) \leq  \frac{\tilde{C}\norm{\mathfrak{t}_{T^\prime}\vert_{\trcl(\mathcal{V})}}}{2N}$. We finish the proof by identifying $D$ with $\tilde{D}^R$ and defining 
\begin{align*}
    C := \max\left[\tilde{C}, \frac{\tilde{C}}{2}\sup_{T^\prime \in \mathcal{C}_B}\norm{\mathfrak{t}_{T^\prime}\vert_{\trcl(\mathcal{V})}}  \right] < \infty.
\end{align*}
\end{proof}

As a direct consequence of the previous result, we get the validity of Theorem \ref{DiscriminationStroategyOutlineTheorem}. 

\begin{proof} \textit{(Theorem \ref{DiscriminationStroategyOutlineTheorem})}
We interpret every channel $T$ with \sq{interaction} functional $\mathfrak{i}_{T}$ as channel with transmission functional $\mathfrak{i}_{T}$. By Lemma \ref{ComparisionIFTF}, it suffices to check conditions 1-7 of Corollary \ref{transmissionFreeDiscriminationCor}. 1, 2, 6 and 7 follow by assumption and 3, 4 and 5 follow directly from Lemma \ref{MaximalVacuumSubspaceLemma} (\ref{prop6}).
\end{proof}

The remainder of this section is devoted to the proof of Theorem \ref{ReductionSuperchannelTheorem}. 
We will show that the transformation depicted in Figure \ref{ReductionSuperchannelFigure} has the desired properties. We will define this superchannel precisely in the proof of Theorem \ref{ReductionSuperchannelTheorem}.  An important part is the so called twirling operation, which we study here for a special group. 

\begin{figure}[htbp]
\centering
\includegraphics[width=.95\textwidth]{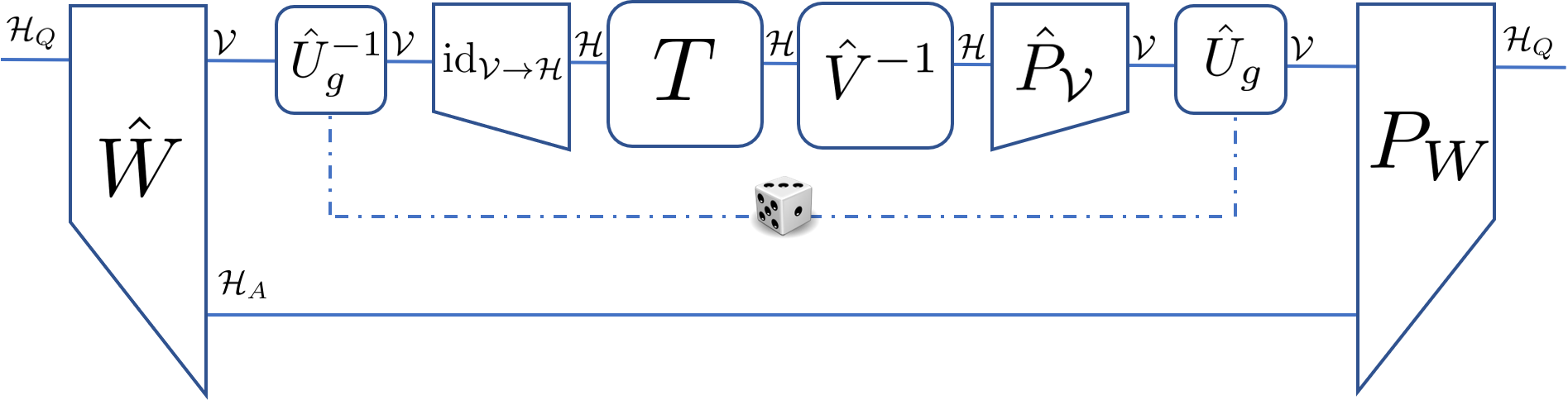}
 \caption{The reduction superchannel for $\mathrm{dim}(\mathcal{V}) > 1$.} \label{ReductionSuperchannelFigure}
\end{figure}

\begin{lem}[Twirling] \label{TwirlingLemma}
For $2 \leq d := \mathrm{dim}(\mathcal{H}) < \infty$, let $v \in \mathcal{H}$ be a unit vector and set $V_v := \mathrm{span}\{v\}$. We define the group
\begin{align} \label{TwirlingGroup}
    G := \set{g =  \idmat_{V_v} \oplus U_g \in \blt(V_v\oplus V_v^\bot)}{U_g \in \blt(V_v^\bot) \text{ is unitary}}
\end{align}
and the twirling superchannel $S : \blt(\trcl(\mathcal{H})) \rightarrow \blt(\trcl(\mathcal{H}))$ by
\begin{align*}
S(T) = \int \hat{U}_g \circ T \circ \hat{U}_g^{-1} \,\mathrm{d}\mu_G(g),
\end{align*}
where $\mu_G$ is the Haar measure on $G$ and $\hat{U}_g : \trcl(\mathcal{H}) \rightarrow \trcl(\mathcal{H})$ is the quantum channel obtained by conjugation with the group element $g \in G$, i.e. $\hat{U}_g(\cdot) = g \cdot g^{-1}$. Then the following statements hold.
\begin{itemize}
\item Let $\psi \in V_v^\bot$ be any unit vector and $\phi := \frac{1}{\sqrt{2}}\left( v + \psi \right)$. If $T : \trcl(\mathcal{H}) \rightarrow \trcl(\mathcal{H})$ is a channel and $\ket{\phi}\bra{\phi}$ is a fixed point of $S(T)$ then $T = \idop$. Conversely, $S(\idop) = \idop$ and thus $\ket{\phi}\bra{\phi}$ is a fixed point of $S(\idop)$. 
\item For a functional $\mathfrak{t} : \trcl(\mathcal{H}) \rightarrow \mathbb{C}$, we have 
\begin{align} \label{FunctionalTwirlingExpression}
    \int \mathfrak{t} \circ \hat{U}_g^{-1} \,\mathrm{d}\mu_G(g) = \mathfrak{t}\left(\frac{P^\bot}{d-1}\right)\, \tr{P^\bot \,\cdot\, } + \mathfrak{t}(\ket{v}\bra{v}) \,\tr{\ket{v}\bra{v}\,\cdot\,}.
\end{align}
\end{itemize}
\end{lem}
\begin{rem}
The integration over the Haar measure can be replaced by a unitary t-design \cite{PhysRevA.80.012304}. We can thus implement the superchannel $S$ without using an ancillary quantum system.
\end{rem}
\begin{proof}
We start by showing that the range of $S$ is spanned by the following seven operators.
\begin{equation} \label{TwirlingSpan}
\begin{gathered} 
\begin{aligned}
   &\tr{\ket{v}\bra{v} \,\cdot\, } \ket{v}\bra{v} \qquad&& \tr{P^\bot \,\cdot\, }\ket{v}\bra{v} \qquad&& \tr{\ket{v}\bra{v} \,\cdot\,} \frac{P^\bot}{d-1} \\ 
   &\tr{P^\bot \,\cdot\, } \frac{P^\bot}{d-1} && P^\bot \cdot \ket{v}\bra{v} && \ket{v}\bra{v} \cdot P^\bot
   \end{aligned}\\
   P^\bot \cdot P^\bot - \tr{P^\bot \,\cdot \,}\frac{P^\bot}{d-1}.
\end{gathered}
\end{equation}
Using the definition of the Haar measure, we obtain that the range of $S$ consists precisely those operators $T : \trcl(\mathcal{H}) \rightarrow \trcl(\mathcal{H})$ that commute with $\hat{U}_g$ for all $g \in G$. We calculate the commutant on the level of Choi matrices. To this end, we identify $\trcl(\mathcal{H})$ with $\mathcal{H} \otimes \mathcal{H}$ via the Choi isomorphism ($\ket{h_i}\bra{h_j} \leftrightarrow h_i \otimes h_j$), where $h_0, h_1, \dots h_{d-1}$ is an orthonormal basis of $\mathcal{H}$ such that $h_0 = v$. The operator corresponding to $\hat{U}_g$ is $g \otimes \overline{g}$, where the complex conjugation is w.r.t the aforementioned basis. We can rewrite this operator as
\begin{align*} 
    g\otimes\overline{g} &= (\idmat_{V_v} \oplus U_g) \otimes  (\idmat_{V_v} \oplus \overline{U}_g) \\&= (\idmat_{V_v} \otimes \idmat_{V_v}) \oplus (\idmat_{V_v} \otimes \overline{U}_g) \oplus (U_g \otimes \idmat_{V_v}) \oplus (U_g\otimes \overline{U}_g)
\end{align*}
The maps $g \mapsto \idmat_{V_v} \otimes \idmat_{V_v}$, $g \mapsto \idmat_{V_v} \oplus \overline{U}_g$ and $g \mapsto U_g \otimes \idmat_{V_v}$ are inequivalent irreducible representations of $G$. If $d = 2$, the representation $g \mapsto (U_g\otimes \overline{U}_g)$ is the trivial 1-dimensional representation. A simple consequence of Schur's lemma is that the commutant then is $2^2 + 1^2 + 1^2 = 6$ dimensional (see \cite{sternberg1995group}, p. 60 for the dimension formula). For $d=2$, also the span of the operators in \eqref{TwirlingSpan} is $6$-dimensional ($P^\bot \cdot P^\bot - \tr{P^\bot \cdot}\frac{P^\bot}{d-1} = 0$). So in this case, we have proven the claim. If $d \geq 3$, then the representation  $g \mapsto (U_g\otimes \overline{U}_g)$ is the direct sum of the trivial 1-dimensional representation and an irreducible $((d-1)^2 - 1)$-dimensional representation (see \cite{VolWernerEntanglementMeasures}). Hence, the dimension of the commutant is $2^2+1^2+1^2+1^2 = 7$. Also the dimension of the span of the operators in \eqref{TwirlingSpan} is $7$-dimensional. 
This proves that the range of $S$ is indeed given by the span of the operators in \eqref{TwirlingSpan}.

For our first claim, we clearly have $S(\idop) = \idop$. 
Conversely, let $T$ be a channel  such that $\ket{\phi}\bra{\phi}$ is a fixed point of $S(T)$. Let $\alpha_1, \alpha_2, \dots ,\alpha_7$ be the coefficients of an expansion of $S(T)$ in terms of the operators in \eqref{TwirlingSpan}. Note that for $d = 2$, this expansion is not unique but can be made that way by demanding $\alpha_7 := 1$. As $\ket{\phi}\bra{\phi}$ is a fixed point of $S(T)$, we have
\begin{align*}
\begin{split}
    \ket{\phi}\bra{\phi} &= \frac{1}{2}\left( \ket{v}\bra{v} + \ket{v}\bra{\psi} + \ket{\psi}\bra{v} + \ket{\psi}\bra{\psi} \right)\\
    &= S(T)(\ket{\phi}\bra{\phi})\\
    &\begin{aligned} \;=\frac{1}{2} \bigg( (\alpha_1 + \alpha_2)\ket{v}\bra{v} + (\alpha_3 + \alpha_4 - \alpha_7)\frac{P^\bot}{d-1} + \alpha_5 \ket{\psi}\bra{v} &+ \alpha_6 \ket{v}\bra{\psi} \\&+ \alpha_7 \ket{\psi}\bra{\psi} \bigg).
    \end{aligned}
\end{split}
\end{align*}
By comparing the second and the last expression, it follows that $\alpha_1 + \alpha_2 = 1$ and $\alpha_5 = \alpha_6 = 1$. If $d = 2$, then $P^\bot = \ket{\psi}\bra{\psi}$ and $\alpha_3 + \alpha_4 = 1$. Otherwise, we have $\alpha_7 = 1$ and $\alpha_3 + \alpha_4 - \alpha_7 = 0$, hence also $\alpha_3 + \alpha_4 = 1$.   Furthermore, 
\begin{align} \label{TraceEquation}
    S(T)(\ket{v}\bra{v}) &= \alpha_1 \ket{v}\bra{v} + \alpha_3 \frac{P^\bot}{d-1}\\
    S(T)(\ket{\psi}\bra{\psi}) &= \alpha_2 \ket{v}\bra{v} + (\alpha_4 - \alpha_7) \frac{P^\bot}{d-1} + \alpha_7 \ket{\psi}\bra{\psi}. \nonumber
\end{align}
As $S(T)$ is trace-preserving, we obtain $\alpha_1 + \alpha_3 = 1$ and $\alpha_3 + \alpha_4 = 1$. Our equations imply that $\alpha_2 = 1-\alpha_1, \alpha_3 = 1-\alpha_1$ and $\alpha_4 = \alpha_1$. Positivity of $S(T)$ in \eqref{TraceEquation} implies that $\alpha_1 \geq 0$ and $\alpha_3 \geq 0$. Thus $0\leq \alpha_1 \leq 1$. We want to show that complete positivity of $S(T)$ even implies $\alpha_1 = 1$. To this end, we define $\mathcal{H}_A := \mathrm{span}\{v, \psi\}$ and $\Omega^+, \Omega^- \in \mathcal{H}_A \otimes \mathcal{H}$ by
\begin{align*}
    \Omega^+ := v\otimes v + \psi \otimes \psi \qquad \Omega^- := v\otimes v - \psi \otimes \psi.
\end{align*}
As $S(T)$ is completely positive, we have
\begin{align*}
    0 &\leq \braket{\Omega^-}{(\idop_A\otimes S(T))(\ket{\Omega^+}\bra{\Omega^+})\; \Omega^-}\\
    &= \braket{\Omega^-}{\left(\ket{v}\bra{v}\otimes \left(\alpha_1 \ket{v}\bra{v} + (1-\alpha_1)\frac{P^\bot}{d-1} \right)\right)\;\Omega^-} 
    \\&+ \braket{\Omega^-}{\left(\ket{\psi}\bra{v}\otimes \ket{\psi}\bra{v}\right)\;\Omega^-} + \braket{\Omega^-}{\left(\ket{v}\bra{\psi}\otimes \ket{v}\bra{\psi}\right)\;\Omega^-}
    \\ &+\braket{\Omega^-}{\left(\ket{\psi}\bra{\psi}\otimes \left((1-\alpha_1) \ket{v}\bra{v} + \alpha_1\frac{P^\bot}{d-1} + \ket{\psi}\bra{\psi} - \frac{P^\bot}{d-1} \right)\right)\;\Omega^-} 
    \\&= \alpha_1 - 2 + \frac{\alpha_1 - 1}{d-1} + 1
    \\&= d \frac{\alpha_1 - 1}{d-1}.
\end{align*}
Thus, $\alpha_1 \geq 1$. This further implies that $\alpha_1 = 1, \alpha_2 = 0, \alpha_3 = 0$ and $\alpha_4 = 1$. Together with the earlier result that $\alpha_5= \alpha_6 = \alpha_7 = 1$, we obtain
\begin{align*}
    S(T) &= \tr{\ket{v}\bra{v} \,\cdot\, }\ket{v}\bra{v} + \tr{P^\bot \,\cdot\, } \frac{P^\bot}{d-1}+ P^\bot \cdot \ket{v}\bra{v}+\ket{v}\bra{v} \cdot P^\bot\\&+P^\bot \cdot P^\bot - \tr{P^\bot \,\cdot \,}\frac{P^\bot}{d-1}
    \\&= \idop. 
\end{align*}
Thus we have shown that if $\ket{\phi}\bra{\phi}$ is a fixed point of $S(T)$, then $S(T) = \idop$. To see that this also implies that $T = \idop$, we note that $S(T)$ is a convex combination of the channels $\hat{U}_g \circ T \circ \hat{U}^{-1}_g$. But as the identity is an extremal element of the convex set of quantum channels, $\hat{U}_g \circ T \circ \hat{U}^{-1}_g$ must be proportional to the identity $\mu_G$-almost everywhere. In particular, $\hat{U}_g \circ T \circ \hat{U}^{-1}_g = \idop$, for some $g \in G$. Thus $T = \idop$. This proves the first claim.

It remains to prove the second claim. For $\mathfrak{t}(\cdot) = \tr{L\, \cdot\,}$ and $\rho \in \trcl(\mathcal{H})$, we have
\begin{align*}
    S^\prime(\mathfrak{t})(\rho) := \int \mathfrak{t} \circ \hat{U}_g^{-1}(\rho) \,\mathrm{d}\mu_G(g) = \tr{\int g L g^{-1} \,\mathrm{d}\mu_G(g) \; \rho}.
\end{align*}
By the definition of the Haar measure, the integral must commute with all $g \in G$. The representation $g \mapsto \idmat_{V_v} \oplus U_g$ is the sum of two inequivalent irreducible representations of $G$. Thus the commutant is $2$-dimensional. It is easy to check that $P^\bot$ and $\ket{v}\bra{v}$ are in the commutant. Thus 
\begin{align*} 
    \int g L g^{-1} \,\mathrm{d}\mu_G(g) = \lambda_1 P^\bot + \lambda_2 \ket{v}\bra{v}, 
\end{align*}
for some $\lambda_1, \lambda_2 \in \mathbb{C}$. Therefore, we can write 
\begin{align*} \label{CombinationFunctional}
    S^\prime(\mathfrak{t})(\rho) = \lambda_1 \tr{P^\bot \rho} + \lambda_2 \tr{\ket{v}\bra{v} \rho}.
\end{align*}
Substituting $P^\bot$ and $\ket{v}\bra{v}$ for $\rho$, yields $\lambda_1 = (d-1)^{-1} \,S^\prime(\mathfrak{t})(P^\bot)$ and $\lambda_2 = S^\prime(\mathfrak{t})(\ket{v}\bra{v})$. As $P^\bot$ and $\ket{v}\bra{v}$ commute with all $g \in G$, we have
\begin{align*}
    S^\prime(\mathfrak{t})(P^\bot) &= \mathfrak{t} \left( \int g^{-1} P^\bot g \,\mathrm{d}\mu_G(g) \right) = \mathfrak{t}(P^\bot)\\
    S^\prime(\mathfrak{t})(\ket{v}\bra{v}) &= \mathfrak{t} \left( \int g^{-1} \ket{v}\bra{v} g \,\mathrm{d}\mu_G(g) \right) = \mathfrak{t}(\ket{v}\bra{v})
\end{align*}
We plug this into \eqref{CombinationFunctional}, and obtain the desired result, equation \eqref{FunctionalTwirlingExpression}. Thus we have proven our last claim. 
\end{proof}

\noindent We are now ready to prove Theorem \ref{ReductionSuperchannelTheorem}. 

\begin{proof}
As already mentioned, the proof consists of an explicit construction of the superchannel $R$. The construction is depicted in Fig. \ref{ReductionSuperchannelFigure}. We start by defining the components of this circuit from left to right.
For the definition of the first component, we define $\mathcal{H}_A$ to be a two-dimensional Hilbert space with orthonormal basis $\{a_0, a_1 \}$. The channel $\hat{W} : \trcl(\mathcal{H}_Q) \rightarrow \trcl(\mathcal{V} \otimes \mathcal{H}_A)$ is defined by $\hat{W}(\cdot) = W\cdot W^\dagger$, with isometry $W : \mathcal{H}_Q \rightarrow \mathcal{V} \otimes \mathcal{H}_A$ defined by
\begin{align*}
    W q_0 &= v\otimes a_0\\
    W q_1 &= \bigg\{\begin{array}{lr}
        \frac{1}{\sqrt{2}}\left( v + \psi \right) \otimes a_1,  & \text{if } \mathrm{dim}(\mathcal{V}) > 1\\
        v\otimes a_1, & \text{if } \mathrm{dim}(\mathcal{V}) = 1
        \end{array},
\end{align*}
where $\psi \in \mathcal{V}$ is any unit vector that is orthogonal to $v$. This channel is designed in order to exhibit the second conclusion of Lemma \ref{TwirlingLemma}. 

The second component is the twirling operation $S : \blt(\trcl(\mathcal{V})) \rightarrow \blt(\trcl(\mathcal{V}))$, which is a superchannel on its own and which we only define for $\mathrm{dim}(\mathcal{V}) > 1$. This operation, is depicted by the two unitary channels $\hat{U}_g$ and $\hat{U}_g^{-1}$ connected by a dashed line and acts as 
\begin{align}
S(\cdot) := \int \hat{U}_g \circ (\cdot) \circ \hat{U}_g^{-1} \,\mathrm{d}\mu_G(g), 
\end{align}
where $\mu_G$ is the Haar measure on the compact group $G$, defined by (cf. Lemma \ref{TwirlingLemma})
\begin{align*}
    G := \set{g =  \idmat_{V_v} \oplus U_g \in \blt(V_v\oplus V_v^\bot)}{U_g \in \blt(V_v^\bot) \text{ is unitary}},
\end{align*}
with $V_v := \mathrm{span}\{v\}$. The channels $\hat{U}_g, \hat{U}^{-1}_g : \trcl(\mathcal{V}) \rightarrow \trcl(\mathcal{V})$ are defined by 
\begin{align*}
    \hat{U}_g(\cdot) &:= (\idmat_{V_v} \oplus U_g) (\cdot) (\idmat_{V_v} \oplus U_g^\dagger) &&\text{and}&& \hat{U}_g^{-1}(\cdot) := (\idmat_{V_v} \oplus U_g^\dagger) (\cdot) (\idmat_{V_v} \oplus U_g).
\end{align*}
The channel $\idop_{\mathcal{V} \rightarrow \mathcal{H}} : \trcl(\mathcal{V}) \rightarrow \trcl(\mathcal{H}), \rho \mapsto \rho$ embeds $\trcl(\mathcal{V})$ into $\trcl(\mathcal{H})$.

To define the channel $\hat{V}^{-1} : \trcl(\mathcal{H}) \rightarrow \trcl(\mathcal{H})$, we use that by assumption, $T$ is isometric on $\mathcal{V}$. This means that there exists an isometry $\tilde{V} : \mathcal{V} \rightarrow \mathcal{H}$ such that $T\vert_{\trcl(\mathcal{V})}(\cdot) = \tilde{V} \cdot \tilde{V}^\dagger$. This isometry can be extended (in a non-unique way) to a unitary and therefore invertible operation $V : \mathcal{H} \rightarrow \mathcal{H}$. We then define 
\begin{align*}
    \hat{V}^{-1}(\cdot) := V^\dagger \cdot V.
\end{align*}

We define the channel $\hat{P}_\mathcal{V} : \trcl(\mathcal{H}) \rightarrow \trcl(\mathcal{V})$ by 
\begin{align*}
    \hat{P}_\mathcal{V}(\cdot) := P_\mathcal{V} \cdot P_\mathcal{V}^\dagger + \tr{(\idmat-P_\mathcal{V}^\dagger P_\mathcal{V}) (\cdot)} \ket{v}\bra{v},
\end{align*}
where $P_\mathcal{V} : \mathcal{H} \rightarrow \mathcal{V}$ is the orthogonal projection onto $\mathcal{V}$.
To finish the channel definitions, we define the channel $P_W : \trcl(\mathcal{V}\otimes\mathcal{H}_A) \rightarrow \trcl(\mathcal{H}_Q)$ by
\begin{align*} 
    P_W(\cdot) := W^\dagger \cdot W + \tr{(\idmat - W W^\dagger) (\cdot)} \ket{q_0}\bra{q_0}.
\end{align*}
We can now define the superchannel $R$. If $\mathrm{dim}(\mathcal{V}) > 1$, we define 
\begin{align}
R(\cdot) := P_W\circ \left(\left[\int \hat{U}_g \circ \hat{P}_\mathcal{V} \circ \hat{V}^{-1} \circ (\cdot) \circ \idop_{\mathcal{V}\rightarrow \mathcal{H}} \circ \hat{U}_g^{-1} \,\mathrm{d}\mu_G(g)\right] \otimes \idop_A \right) \circ \hat{W},  
\end{align}
and if $\mathrm{dim}(\mathcal{V}) = 1$, we define 
\begin{align} \label{RForDim1}
  R(\cdot) := P_W \circ \hat{V}^{-1} \circ (\cdot) \circ \idop_{\mathcal{V}\rightarrow \mathcal{H}} \circ \hat{W} 
\end{align}

With the definition in palace, it only remains to show that the superchannel $R$ has the claimed properties.  
To prove the first claim, let $T^\prime \in \blt(\trcl(\mathcal{H}))$ such that $T\vert_{\trcl(\mathcal{V})} = T^\prime\vert_{\trcl(\mathcal{V})}$. For $\mathrm{dim}(\mathcal{V}) > 1$, we use that by construction $\hat{V}^{-1}\circ T^\prime|_{\trcl(\mathcal{V})} = \idop_\mathcal{V}$ and that operators in $\trcl(\mathcal{V})$ are fixed points of $P_\mathcal{V}$. We get 
\begin{align*}
    R(T^\prime) &= P_W\circ \left(\left[\int \hat{U}_g \circ \idop_\mathcal{V} \circ \hat{U}_g^{-1} \,\mathrm{d}\mu_G(g)\right] \otimes \idop_A \right) \circ \hat{W},  \\
    &= P_W \circ \hat{W}\\
    &= \idop_Q.
\end{align*}
By means of a similar argument, it follows that the claim also holds for $\mathrm{dim}(\mathcal{V}) = 1$.
To prove the second claim, we start by showing that $\ket{q_0}\bra{q_0}$ is a fixed point of $R(T^\prime)$, for every channel $T^\prime$. For $\mathrm{dim}(\mathcal{V}) > 1$, we have
\begin{align*}
    R(T^\prime)(\ket{q_0}\bra{q_0}) &= P_W\circ \left(S(\hat{P}_\mathcal{V} \circ \hat{V}^{-1} \circ T^\prime \circ \idop_{\mathcal{V} \rightarrow \mathcal{H}}) \otimes \idop_A\right) \circ W(\ket{q_0}\bra{q_0}) \\
    &= P_W \left(S(\hat{P}_\mathcal{V} \circ \hat{V}^{-1} \circ T^\prime \circ \idop_{\mathcal{V} \rightarrow \mathcal{H}})(\ket{v}\bra{v}) \otimes \ket{a_0}\bra{a_0} \right)\\
    &= \ket{q_0}\bra{q_0},
\end{align*}
where the last line follows as $P_W$ maps every state of the form $\sigma \otimes \ket{a_0}\bra{a_0}$ to $\ket{q_0}\bra{q_0}$. 
An analogous argument yields that $\ket{q_0}\bra{q_0}$ is also a fixed point of $R(T^\prime)$, if $\mathrm{dim}(\mathcal{V}) = 1$. Conversely, assume that $T^\prime \in \blt(\trcl(\mathcal{H}))$ is a channel such that $T\vert_{\trcl(\mathcal{V})} \neq T^\prime\vert_{\trcl(\mathcal{V})}$ and $\rho \in \states(\mathcal{H}_Q)$ is a fixed point of $R(T^\prime)$. We will prove that $\rho = \ket{q_0}\bra{q_0}$. We do so by first showing that if $\rho \neq \ket{q_0}\bra{q_0}$, then $\ket{q_1}\bra{q_1}$ is also a fixed point of $R(T^\prime)$, which will lead to a contradiction. By part $1$ of the theorem, $\ket{q_0}\bra{q_0}$ is a fixed point of $R(T^\prime)$. Hence, Lemma \ref{InvariantSubsapceLemma} implies that $\mathrm{span}\{ \ket{q_0}\bra{q_1} \}$ and $\mathrm{span}\{ \ket{q_1}\bra{q_0} \}$ are invariant subspaces of $R(T^\prime)$. Thus,
\begin{align*}
    \braket{q_0}{R(T^\prime)(\ket{q_0}\bra{q_1})\,q_0} = \braket{q_0}{R(T^\prime)(\ket{q_1}\bra{q_0})\,q_0} = 0.
\end{align*}
We then have
\begin{align*}
    \braket{q_0}{\rho \,q_0} &= \braket{q_0}{R(T^\prime)(\rho)\, q_0}\\
    &= \sum_{i,j = 0}^1 \braket{q_i}{\rho\,q_j} \braket{q_0}{R(T^\prime)(\ket{q_i}\bra{q_j})\, q_0} \\
    &= \sum_{i = 0}^1 \braket{q_i}{\rho\,q_i} \braket{q_0}{R(T^\prime)(\ket{q_i}\bra{q_i})\, q_0} \\
    &= \braket{q_0}{\rho \,q_0} + \braket{q_1}{\rho\,q_1} \braket{q_0}{R(T^\prime)(\ket{q_1}\bra{q_1})\, q_0}.
\end{align*}
Hence,
\begin{align*}
    \braket{q_1}{\rho\,q_1} \braket{q_0}{R(T^\prime)(\ket{q_1}\bra{q_1})\, q_0} = 0.
\end{align*}
If $\braket{q_1}{\rho\, q_1} = 0$, then positivity of $\rho$ implies that $\rho = \ket{q_0}\bra{q_0}$, which contradicts the assumption that $\rho \neq \ket{q_0}\bra{q_0}$. It follows that
\begin{align*}
    \braket{q_0}{R(T^\prime)(\ket{q_1}\bra{q_1})\, q_0} = 0.
\end{align*}
Positivity of $R(T^\prime)(\rho)$ yields $R(T^\prime)(\ket{q_1}\bra{q_1}) = \ket{q_1}\bra{q_1}$, which shows that $\ket{q_1}\bra{q_1}$ is a fixed point of $R(T^\prime)$. We will now show that this leads to a contradiction. With the abbreviations $\tilde{S} := S(\hat{P}_\mathcal{V} \circ \hat{V}^{-1} \circ T^\prime \circ \idop_{\mathcal{V} \rightarrow \mathcal{H}})$ and $\phi := \frac{1}{\sqrt{2}}(v + \psi)$, we get
\begin{align*}
    \ket{q_1}\bra{q_1} &= R(T^\prime)(\ket{q_1}\bra{q_1}) \\
    &= P_W \left( \tilde{S}(\ket{\phi}\bra{\phi}) \otimes \ket{a_1}\bra{a_1} \right) \\
    &= \tr{\ket{\phi}\bra{\phi} \,\tilde{S}(\ket{\phi}\bra{\phi}) } \ket{q_1}\bra{q_1} + \tr{(\idmat - W W^\dagger)\, \tilde{S}(\ket{\phi}\bra{\phi})} \ket{q_0}\bra{q_0}.
\end{align*}
Comparing the last with the first line implies that $\tr{\ket{\phi}\bra{\phi} \,\tilde{S}(\ket{\phi}\bra{\phi}) } = 1$. We observe the the latter equation says that the Cauchy-Schwarz inequality (w.r.t. the Hilbert-Schmidt inner product) is satisfied with equality. Thus $\tilde{S}(\ket{\phi}\bra{\phi}) = \ket{\phi}\bra{\phi}$. Lemma \ref{TwirlingLemma} then implies 
\begin{align*}
    \hat{P}_\mathcal{V} \circ \hat{V}^{-1} \circ T^\prime \circ \idop_{\mathcal{V} \rightarrow \mathcal{H}} = \idop_\mathcal{V}.
\end{align*}
Note that $P_\mathcal{V}$ is the sum of the two completely positive trace non-increasing maps, $P_1(\cdot) := P_\mathcal{V} \cdot P_\mathcal{V}$ and $P_2(\cdot) := \tr{(\idmat-P_\mathcal{V}^\dagger P_\mathcal{V}) (\cdot)} \ket{v}\bra{v}$. Thus, with the appropriate normalization, the extremal point of the convex set of completely positive maps, $\idop_\mathcal{V}$ can be written as a convex combination of, $P_i \circ \hat{V}^{-1} \circ T^\prime \circ \idop_{\mathcal{V} \rightarrow \mathcal{H}}$. Thus,
\begin{align} \label{PreContradictionIdentity}
    \hat{V}^{-1} \circ T^\prime \circ \idop_{\mathcal{V} \rightarrow \mathcal{H}} = \idop_{\mathcal{V} \rightarrow \mathcal{H}}
\end{align}
As $\hat{V}^{-1}$ is invertible and $T^\prime \circ \idop_{\mathcal{V} \rightarrow \mathcal{H}}$, identity \eqref{PreContradictionIdentity} is equivalent to
\begin{align*}
    T^\prime\vert_{\trcl(\mathcal{V})} = \hat{V}\vert_{\trcl(\mathcal{V})}.
\end{align*}
By construction of $\hat{V}$, the RHS equals $T\vert_{\trcl(\mathcal{V})}$. But this contradicts the assumption that $T\vert_{\trcl(\mathcal{V})} \neq T^\prime\vert_{\trcl(\mathcal{V})}$. Thus $\ket{q_1}\bra{q_1}$ cannot be a fixed point of $R(T^\prime)$. Consequently, $\rho = \ket{q_0}\bra{q_0}$, which proves that $\ket{q_0}\bra{q_0}$ is the only state that is a fixed point of $R(T^\prime)$. This proves the second claim. 
To prove the third claim, we must calculate how our protocol transforms the transmission functional. For $\mathrm{dim}(\mathcal{V}) = 1$, we get directly from the definition \eqref{RForDim1} that $\mathfrak{t}_{R(T)}(\cdot) = \tr{\cdot} \mathfrak{t}_T(\ket{v}\bra{v}) = 0$.
For $\mathrm{dim}(\mathcal{V}) > 1$, the transmission functional $\mathfrak{t}_T$ transforms to $\mathfrak{t}_{R(T)}$, given by
\begin{align} \label{DefTransformFunct}
    \mathfrak{t}_{R(T)} := \int \mathfrak{t}_T \circ \idop_{\mathcal{V} \rightarrow \mathcal{H}} \circ \hat{U}_g^{-1} \circ \mathrm{tr}_A \circ \hat{W} \,\mathrm{d}\mu_G(g)
\end{align}
To evaluate \eqref{DefTransformFunct}, we use \eqref{FunctionalTwirlingExpression} and get
\begin{align*}
    \mathfrak{t}_{R(T)}(\cdot) = \mathfrak{t}_T\circ \idop_{\mathcal{V}\rightarrow \mathcal{H}}\left(\frac{P^\bot}{d-1}\right) \tr{P^\bot \ptr{A}{\hat{W}(\cdot)}} 
\end{align*}
A direct calculation then yields the claim. 

\end{proof}
\begin{rem} \label{OneNeedsAncillaryRemark}
With our protocol, we achieved a transformation from channels on $\mathcal{H}$ to qubit channels with certain properties. This was achieved by using classical communication and one ancillary qubit. To demonstrate that our implementation of this transformation uses the quantum resources in the most economic way possible, we show that in general one cannot use only classical communication to implement a transformation which has the desired properties. To this end, we consider the following procedure. First, we use an instrument to transform the state and to obtain classical information. Then we apply the channel, which should be transformed. Afterwards, we apply some quantum channel, where the choice of the channel may depend on the classical information that we obtained in the first step. Our instrument described by a collection of non-zero quantum operations $I_1, I_2, \dots, I_N$, such that $\sum_i I_i$ is trace-preserving. We denote the associated channels that are applied in the last step by $\Lambda_1, \Lambda_2, \dots, \Lambda_N$. Our protocol then implements the following transformation
\begin{align*}
    T \mapsto \sum_i \Lambda_i \circ T \circ I_i.
\end{align*}
Assume that the channel $T$ of the Theorem \ref{ReductionSuperchannelTheorem} is the identity and $\mathrm{dim}(\mathcal{H}) = \mathrm{dim}(\mathcal{V}) = 2$. Our first requirement is that $\idop \mapsto \idop$. Thus
\begin{align*}
    \idop = \sum_i \Lambda_i\circ I_i
\end{align*}
Since $\idop$ is an extreme point of the convex set of quantum operations, there must be non-negative coefficients $p_1, p_2, \dots, p_N$,  such that 
\begin{align*}
    \Lambda_i\circ I_i = p_i\cdot \idop, \text{ for } i = 1, 2, \dots, N. 
\end{align*}
This implies that $\Lambda_i$ and $I_i$ must be proportional to a unitary conjugation, i.e. $\Lambda_i(\cdot) = U_i^\dagger \cdot U_i$ and $I_i(\cdot) = p_i U_i \cdot U_i^\dagger$, for some unitary operator $U_i$. Our second requirement is that (since $\mathcal{V} = \mathcal{H}$), every channel except $\idop$ must be transformed to a state whose only fixed point is $\ket{q_0}\bra{q_0} =: P_0$. In particular, for the pinching channel, defined by $T_P(\cdot) = P_0\cdot P_0 + P_1\cdot P_1$, with $P_1 := \idmat - P_0$, we have
\begin{align*}
    P_0 = \sum_i \sum_{j=0}^1 p_i (U_i^\dagger P_j U_i) P_0 (U_i^\dagger P_j U_i).
\end{align*}
Since $P_0$ is an extremal point of the convex set $\set{\rho \geq 0}{\tr{\rho} \leq 1}$, we get that 
\begin{align*}
    (U_i^\dagger P_j U_i) P_0 (U_i^\dagger P_j U_i) = \lambda_{ij} P_0,
\end{align*}
for some $\lambda_{ij} \geq 0$. From this, we conclude that either $U_i^\dagger P_j U_i = P_0$ or $U_i^\dagger P_j U_i = P_1$. But then the application of the transformed channel to $P_1$ yields
\begin{align*}
    \sum_i \sum_{j=0}^1 p_i (U_i^\dagger P_j U_i) P_1 (U_i^\dagger P_j U_i) = P_1.
\end{align*}
Thus, $P_0$ is not the only state that is a fixed point of the transformed channel. Hence, to achieve our transformation, an ancillary system is needed.  
\end{rem}

\newpage

\section{No-go results} \label{NOGOSection}

In this section, we consider the case for which we claimed in our main theorem that it is impossible to discriminate two channels in an \sq{interaction-free} manner. There are two major results in this section: Theorem \ref{QuantitativeNoGo} which claims an inequality between the error probability and the \sq{interaction} probability; and Theorem \ref{RateLimitTheorem}, which claims that, under a certain condition, the best achievable rate (in terms of the number of channel uses, $N$) for the \sq{interaction} probability is proportional to $N^{-1}$. Both theorems will turn out to be consequences of our main technical results: Proposition \ref{No-GoInequality} and Proposition \ref{NoGoTheorem}. The proof techniques for these results are inspired by the techniques used in two papers by Mitchison, Massar and Pironio \cite{PhysRevA.63.032105, PhysRevA.64.062303}, who proved an analogous no-go result for the special case of a semitransparent object. Before we state the first proposition, we define a quantity that will appear as proportionality constant in the results of this section. As this may seem complicated, we want to stress that in all relevant cases, $C_{\mathcal{V}, \mathcal{W}}^{(T_A^\downarrow, T_B^\downarrow)}$ can be bounded by $2$. 

\begin{defn} \label{ProportionalityConstantDef}
For $\mathrm{dim}(\mathcal{H}) < \infty$, let $T_A^\downarrow, T_B^\downarrow: \trcl(\mathcal{H}) \rightarrow \trcl(\mathcal{H})$ be two linear maps, let $\mathcal{V}$ be a linear subspace of $\mathcal{H}$ and let $\mathcal{W} = \{\mathcal{W}_1, \mathcal{W}_2, \dots, \mathcal{W}_K\}$ be a collection of mutually orthogonal subspaces of $\mathcal{V}^\bot$ with the property that $\mathcal{V}^\bot = \mathcal{W}_1 \oplus \mathcal{W}_2 \oplus \dots \oplus \mathcal{W}_K$. Furthermore, let $P$ and $P_1, P_2, \dots, P_K$ be the orthogonal projections onto $\mathcal{V}$ and $\mathcal{W}_1, \mathcal{W}_2, \dots \mathcal{W}_K$. \\We define the quantity $C_{\mathcal{V}, \mathcal{W}}^{(T_A^\downarrow, T_B^\downarrow)}$ to be the infimum of the (possibly empty) set of real numbers $r$ with the property that there exists a finite-dimensional Hilbert space $\mathcal{H}_E$, isometries $V_A, V_B : \mathcal{H} \rightarrow \mathcal{H}_E\otimes\mathcal{H}$ and orthogonal projections $P_A, P_B : \mathcal{H}_E \rightarrow \mathcal{H}_E$ such that\footnote{ $\norm{\cdot}$ is the operator norm on $\blt(\mathcal{H})$.} 
\begin{subequations}
\begin{align}
    r &= \max_{1\leq k \leq K} \norm{P_k ({V_A}^\dagger (P_A P_B \otimes \idmat) V_B - \idmat) P_k} \label{DefCabFactor}\\
    V_A P &= V_B P \label{StinespringEqualityCondition}\\
    T_X^\downarrow(\cdot) &= \ptr{E}{(P_X\otimes\idmat)V_X \cdot V_X^\dagger} \label{StinespringDilationCondition}, 
\end{align}
\end{subequations}
for $X \in \{A, B\}$.
\end{defn}

We are now ready to state the first important proposition, which establishes, for a single channel use, an uncertainty relation between the \sq{information-gain} (RHS of \eqref{InformationInequality}) about the identity of the channel (is it $T_A$ or $T_B$?) and a quantity that depends on the probability that if we would measure the input states, we would find that they are supported in the orthogonal complement of a subspace $\mathcal{V}$. This subspace will later on be chosen to be a maximum vacuum subspace.

\begin{prop}[Information-interaction tradeoff] \label{No-GoInequality}
For $\mathrm{dim}(\mathcal{H}) < \infty$, let $T_A^\downarrow, T_B^\downarrow : \trcl(\mathcal{H}) \rightarrow \trcl(\mathcal{H})$ be quantum operations and let $\mathcal{V}$ be a subspace of $\mathcal{H}$ such that $T_A^\downarrow\vert_{\trcl(\mathcal{V)}}$ is trace-preserving and $T_A^\downarrow\vert_{\trcl(\mathcal{V)}} = T_B^\downarrow\vert_{\trcl(\mathcal{V)}}$. Let $\mathcal{W} = \{\mathcal{W}_1, \mathcal{W}_2, \dots, \mathcal{W}_K\}$ be a collection of mutually orthogonal subspaces of $\mathcal{V}^\bot$, such that $\mathcal{V}^\bot = \mathcal{W}_1 \oplus \mathcal{W}_2 \oplus \dots \oplus \mathcal{W}_K$. Denote the orthogonal projections onto these subspaces by $P_1, P_2, \dots, P_K$. Then $C_{\mathcal{V}, \mathcal{W}}^{(T_A^\downarrow, T_B^\downarrow)} \leq 2$ and
\begin{align} \label{InformationInequality}
\sqrt{F}(\rho, \sigma) - \sqrt{F}(T^\downarrow_A(\rho), T^\downarrow_B(\sigma)) &\leq C_{\mathcal{V}, \mathcal{W}}^{(T_A^\downarrow, T_B^\downarrow)} \sum_{k=1}^K \sqrt{\tr{P_k \rho} \tr{P_k \sigma}},
\end{align}
for all $\rho, \sigma \geq 0$.
\end{prop}

Before proving the proposition, let us remark that Proposition \ref{FidelityInequality}, is a direct consequence thereof. 

\begin{proof} \textit{(Proposition \ref{FidelityInequality})}
This follows directly from the fact that the fidelity can be characterized in terms of the minimum over measurements of expressions of the form given on the RHS of \eqref{InformationInequality} (see \cite{nielsen00}, p. 412).
\end{proof}

\begin{proof} \textit{(Proposition \ref{No-GoInequality})}
We first establish that $C_{\mathcal{V}, \mathcal{W}}^{(T_A^\downarrow, T_B^\downarrow)} \leq 2$. Let $P$, $P^\bot$ be the orthogonal projections onto $\mathcal{V}$ and $\mathcal{V}^\bot$. By applying the triangular inequality and the sub-multiplicativity of the operator norm to the definition of $C_{\mathcal{V}, \mathcal{W}}^{(T_A^\downarrow, T_B^\downarrow)}$, it follows that if there exist $\mathcal{H}_E, V_A, V_B, P_A$ and $P_B$ with the properties of Definition \ref{ProportionalityConstantDef}, then $C_{\mathcal{V}, \mathcal{W}}^{(T_A^\downarrow, T_B^\downarrow)} \leq 2$. Therefore, we start our proof by showing the existence of the aforementioned quantities.
It is a basic property of completely positive trace non-increasing maps (see \cite{nielsen00}, p. 365) that there exist finite-dimensional Hilbert spaces $\mathcal{H}_{E_A}$ and $\mathcal{H}_{E_B}$, isometries $\tilde{V}_A : \mathcal{H} \rightarrow \mathcal{H}_{E_A} \otimes \mathcal{H}$ and $\tilde{V}_B : \mathcal{H} \rightarrow \mathcal{H}_{E_B} \otimes \mathcal{H}$ and orthogonal projections $\tilde{P}_A : \mathcal{H}_{E_A} \rightarrow \mathcal{H}_{E_A}$ and $\tilde{P}_B : \mathcal{H}_{E_B} \rightarrow \mathcal{H}_{E_B}$, such that $T_A^\downarrow(\cdot) = \ptr{E_A}{(\tilde{P}_A \otimes \idmat) \tilde{V}_A \cdot \tilde{V}_A^\dagger}$ and $T_B^\downarrow(\cdot) = \ptr{E_B}{(\tilde{P}_B \otimes \idmat) \tilde{V}_B \cdot \tilde{V}_B^\dagger}$. By enlarging the smaller of the two ancillary Hilbert spaces and identifying two orthonormal basis, we can achieve that $\mathcal{H}_{E_A}$ and $\mathcal{H}_{E_B}$ are the same space, $\mathcal{H}_E$. By assumption, $T_A\vert_{\trcl(\mathcal{V})}$ and $T_B\vert_{\trcl(\mathcal{V})}$ are trace-preserving. It follows that $(\tilde{P}_A\otimes \idmat) \tilde{V}_A|_\mathcal{V}$ and $(\tilde{P}_A\otimes \idmat) \tilde{V}_B|_\mathcal{V}$ are isometries and thus $(\tilde{P}_A\otimes \idmat) \tilde{V}_A|_\mathcal{V} = \tilde{V}_A|_\mathcal{V}$ and $(\tilde{P}_B\otimes \idmat) \tilde{V}_B|_\mathcal{V} = \tilde{V}_B|_\mathcal{V}$. Hence, $\tilde{V}_A|_\mathcal{V}$ and $\tilde{V}_B|_\mathcal{V}$ are Stinespring isometries of the same channel and thus are related by a unitary operator on $\mathcal{H}_E$. Precisely, there exists a unitary operator $W : \mathcal{H}_E \rightarrow \mathcal{H}_E$ such that $\tilde{V}_B|_\mathcal{V} = (W\otimes \idmat)\tilde{V}_A|_\mathcal{V}$. Equivalently, $\tilde{V}_BP = (W\otimes \idmat)\tilde{V}_A P$. It is then easy to verify that the operators $V_A := (W\otimes \idmat)\tilde{V}_A$, $V_B := \tilde{V}_B$ and $P_A := W \tilde{P}_A W^{-1}, P_B := \tilde{P}_B$ satisfy the requirements \eqref{StinespringDilationCondition} and \eqref{StinespringEqualityCondition}. In particular, we have 
\begin{align} \label{NoGoEquivalent}
    (P_A\otimes \idmat)V_AP &= V_AP = V_BP = (P_B\otimes\idmat)V_BP.
\end{align}
This finishes the proof of the first part of the proposition. For the second part, we fix $V_A, V_B, P_A$ and $P_B$ such that the conditions \eqref{StinespringDilationCondition} and \eqref{StinespringEqualityCondition} are satisfied. In particular, this implies that \eqref{NoGoEquivalent} holds. To prove the inequality, we proceed as follows: for two positive operators $\rho, \sigma \geq 0$, Uhlmann's theorem implies that there exists a finite-dimensional Hilbert space $\mathcal{H}_Q$ and two vectors $\psi, \phi \in \mathcal{H}_Q \otimes \mathcal{H}$ (purifications) such that $\ptr{Q}{\ket{\psi}\bra{\psi}} = \rho$ and $\ptr{Q}{\ket{\phi}\bra{\phi}} = \sigma$ and $\sqrt{F}(\rho, \sigma) = |\braket{\psi}{\phi}|$. We further note that $(\idmat_Q \otimes (P_A\otimes\idmat) V_A)\ket{\psi}$ and $(\idmat_Q \otimes (P_B\otimes\idmat) V_B)\ket{\phi}$ are purifications of $T_A^\downarrow(\rho)$ and $T_B^\downarrow(\sigma)$. Hence, Uhlmann's theorem implies that
\begin{align}
    \sqrt{F}(T_A^\downarrow(\rho), T_B^\downarrow(\sigma)) \geq \abs{\braket{(\idmat_\mathcal{Q}\otimes (P_A\otimes \idmat)V_A)\psi}{(\idmat_\mathcal{Q}\otimes (P_B\otimes \idmat)V_B)\phi}} \label{ApplicationOfUhlmannsTheorem}
\end{align}
By inserting $\idmat_\mathcal{Q}\otimes P + \idmat_\mathcal{Q}\otimes P^\bot$ (which is equal to the identity) and expanding the scalar product, we obtain
\begin{subequations} \label{ExpansionAll}
\begin{align}
\text{RHS of \eqref{ApplicationOfUhlmannsTheorem}} &= |\braket{(\idmat_{Q}\otimes (P_A\otimes \idmat)V_AP) \psi}{(\idmat_Q\otimes (P_B\otimes \idmat)V_BP) \phi}  \\&+ 
\braket{(\idmat_Q\otimes (P_A\otimes \idmat)V_AP^\bot) \psi}{(\idmat_Q\otimes (P_B\otimes \idmat)V_BP) \phi} \label{Expansionb} \\&+
\braket{(\idmat_Q\otimes (P_A\otimes \idmat)V_AP) \psi}{(\idmat_Q\otimes (P_B\otimes \idmat)V_BP^\bot) \phi} \label{Expansionc}\\&+
\braket{(\idmat_Q\otimes (P_A\otimes \idmat)V_AP^\bot) \psi}{(\idmat_Q\otimes (P_B\otimes \idmat)V_BP^\bot) \phi} |.
\end{align}
\end{subequations}
It is not hard to see from \eqref{NoGoEquivalent} that the terms \eqref{Expansionb} and \eqref{Expansionc} vanish. Explicitly, we have
\begin{align*}
\eqref{Expansionb} &= \braket{(\idmat_Q\otimes (P_A\otimes \idmat)V_AP^\bot) \psi}{(\idmat_Q\otimes (P_B\otimes \idmat)V_BP) \phi} \\ &=\braket{(\idmat_Q\otimes (P_A\otimes \idmat)V_AP^\bot) \psi}{(\idmat_Q\otimes (P_A\otimes \idmat)V_AP) \phi} \\&= \braket{(\idmat_Q\otimes V_A P^\bot) \psi}{(\idmat_Q\otimes (P_A\otimes \idmat)V_AP) \phi} \\&= \braket{(\idmat_Q\otimes V_AP^\bot) \psi}{(\idmat_Q\otimes V_AP) \phi} \\&= \braket{\psi}{(\idmat_Q\otimes P^\bot P) \phi} \\&= 0
\end{align*}
and similarly for \eqref{Expansionc}. Adding and subtracting $\braket{(\idmat_Q\otimes P^\bot)\psi}{(\idmat_Q\otimes P^\bot)\phi}$ and using the inverse triangular inequality, yields
\begin{align} \label{NoGoSecondEstimate}
\begin{split}
\text{\eqref{ExpansionAll}} &\geq |\braket{(\idmat_Q\otimes P)\psi}{(\idmat_Q\otimes P)\phi} + \braket{(\idmat_Q\otimes P^\bot)\psi}{(\idmat_Q\otimes P^\bot)\phi}| \\&- |\braket{(\idmat_Q\otimes (P_A\otimes \idmat)V_AP^\bot) \psi}{(\idmat_Q\otimes (P_B\otimes \idmat)V_BP^\bot) \phi}\\&\quad-\braket{(\idmat_Q\otimes P^\bot)\psi}{(\idmat_Q\otimes P^\bot)\phi}|.
\end{split}
\end{align}
We further use that $P^\bot P = 0$ (thus $\sqrt{F}(\rho, \sigma) =  |\braket{(\idmat_Q\otimes P)\psi}{(\idmat_Q\otimes P)\phi} + \braket{(\idmat_Q\otimes P^\bot)\psi}{(\idmat_Q\otimes P^\bot)\phi}|$) and some rearrangement to arrive at
\begin{align} \label{NextEstimateNoGo}
\eqref{NoGoSecondEstimate} &= \sqrt{F}(\rho, \sigma) - |\braket{(\idmat_Q\otimes P^\bot)\psi}{(\idmat_Q\otimes P^\bot(V_A^\dagger (P_A P_B\otimes \idmat)V_B
- \idmat) P^\bot)\phi}|
\end{align}
As by assumption, $P^\bot = \sum_k P_k$ and $P_kP_l = 0$ for $k\neq l$, we get
\begin{align} \label{LastineqProofNoGo}
    \eqref{NextEstimateNoGo} &\geq \sqrt{F}(\rho, \sigma) - \sum_{k=1}^K |\braket{(\idmat_Q\otimes P_k)\psi}{(\idmat_Q\otimes P_k(V_A^\dagger (P_A P_B\otimes \idmat)V_B
- \idmat) P_k)\phi}| \nonumber
\\&\begin{aligned}
\geq  \sqrt{F}(\rho, \sigma) - \sum_{k=1}^K \bigg\{ \norm{P_k(V_A^\dagger (P_A P_B\otimes \idmat)V_B - \idmat) P_k} & \\  \norm{(\idmat_Q\otimes P_k)\psi}& \norm{(\idmat_Q\otimes P_k)\phi} \bigg\} \end{aligned}\nonumber \\
&= \sqrt{F}(\rho, \sigma) - \sum_{k=1}^K \norm{P_k(V_A^\dagger (P_A P_B\otimes \idmat)V_B - \idmat) P_k}  \sqrt{\tr{P_k\rho}\tr{P_k\sigma}},
\end{align}
where we used the Cauchy-Schwarz inequality and the sub-multiplicativity of the matrix norm to get from the first to the second line. For the last line, we used that  
\begin{align*}
\norm{\idmat_Q\otimes P_k\psi}^2 &= \braket{\psi}{(\idmat_Q\otimes P_k) \psi} = \tr{(\idmat_Q\otimes P_k) \ket{\psi}\bra{\psi}} = \tr{P_k \ptr{Q}{\ket{\psi}\bra{\psi}}} \\&= \tr{P_k\rho}.
\end{align*}
As the only constraints that $V_A, V_B, P_A, P_B$ and $\mathcal{H}_E$ have to satisfy are the ones of Definition \ref{ProportionalityConstantDef}, we conclude that 
\begin{align*}
    \eqref{LastineqProofNoGo} \geq \sqrt{F}(\rho, \sigma) - C_{\mathcal{V}, \mathcal{W}}^{(T_A^\downarrow, T_B^\downarrow)} \sum_{k=1}^K \sqrt{\tr{P_k\rho}\tr{P_k\sigma}}
\end{align*}
This proves the claim.
\end{proof}
Proposition \ref{No-GoInequality} does not allow for ancillary systems. In the following proposition, which is an iterated refinement of the preceding one, we show that this problem can be solved by applying Proposition \ref{No-GoInequality} to $T^\downarrow \otimes \idop$.

\begin{prop}[Technical no-go theorem] \label{NoGoTheorem}
For $\mathrm{dim}(\mathcal{H}) < \infty$, let $T_A^\downarrow, T_B^\downarrow : \trcl(\mathcal{H}) \rightarrow \trcl(\mathcal{H})$ be two completely positive trace non-increasing maps. Let $\mathcal{V}$ be a subspace of $\mathcal{H}$ such that $T_A^\downarrow\vert_{\trcl(\mathcal{V)}}$ is trace-preserving and $T_A^\downarrow\vert_{\trcl(\mathcal{V)}} = T_B^\downarrow\vert_{\trcl(\mathcal{V)}}$. Let $\mathcal{W} = \{\mathcal{W}_1, \mathcal{W}_2, \dots, \mathcal{W}_K\}$ be a collection of mutually orthogonal subspaces of $\mathcal{V}^\bot$, such that $\mathcal{V}^\bot = \mathcal{W}_1 \oplus \mathcal{W}_2 \oplus \dots \oplus \mathcal{W}_K$. Denote the orthogonal projections onto these subspaces by $P_1, P_2, \dots, P_K$.
Furthermore, let $T_A, T_B : \trcl(\mathcal{H}) \rightarrow \trcl(\mathcal{H})$ be completely positive maps such that $T_A - T_A^\downarrow$ and $T_B - T_B^\downarrow$ are also completely positive. Then for every finite-dimensional $N$-step discrimination strategy $D = (\mathcal{H}, \mathcal{H}_Z, s_0, \Lambda)$, we have
\begin{align} \label{NoGoFidelityInequality}
   1 - \sqrt{F}(\rho_N^{T_A}, \rho_N^{T_B})  \leq C_{\mathcal{V}, \mathcal{W}}^{(T_A^\downarrow, T_B^\downarrow)} \sum_{i = 0}^{N-1}\sum_{k = 1}^{K} \sqrt{\tr{P_k \ptr{Z}{\rho_i^{T_A^\downarrow}}}\cdot\tr{P_k \ptr{Z}{\rho_i^{T_B^\downarrow}}}},
\end{align}
where $\rho$ is the intermediate state map of $D$. Furthermore,  $C_{\mathcal{V}, \mathcal{W}}^{(T_A^\downarrow, T_B^\downarrow)} \leq 2$.
\end{prop}
\begin{cor}
For $\mathrm{dim}(\mathcal{H}) < \infty$, let $T_A^\downarrow, T_B^\downarrow : \trcl(\mathcal{H}) \rightarrow \trcl(\mathcal{H})$ be two completely positive trace non-increasing maps. Let $\mathcal{V}$ be a subspace of $\mathcal{H}$ such that $T_A^\downarrow\vert_{\trcl(\mathcal{V)}}$ is trace-preserving and $T_A^\downarrow\vert_{\trcl(\mathcal{V)}} = T_B^\downarrow\vert_{\trcl(\mathcal{V)}}$. Then
\begin{align}
1 - \sqrt{F}(\rho_N^{T_A}, \rho_N^{T_B})  \leq C_{\mathcal{V}, \mathcal{W}}^{(T_A^\downarrow, T_B^\downarrow)} \sum_{i = 0}^{N-1}\sum_{k = 1}^{K} \sqrt{\tr{P_k \ptr{Z}{\rho_i^{T_A^\downarrow}}}\cdot\tr{P_k \ptr{Z}{\rho_i^{T_B^\downarrow}}}},
\end{align}
\end{cor}
\begin{proof} To reduce the overhead in notation, we define $\rho_i := \rho_i^{T_A}$, $\rho_i^\downarrow := \rho_i^{T_A^\downarrow}$ and $\sigma_i := \rho_i^{T_B}$, $\sigma_i^\downarrow := \rho_i^{T_B^\downarrow}$. We start to prove the proposition by showing that
\begin{align} \label{FildelityInequality}
1-\sqrt{F}(\rho_N, \sigma_N) \leq 1 - \sqrt{F}(\rho_N^\downarrow, \sigma_N^\downarrow).
\end{align}
This inequality follows from the strong concavity of the fidelity and the observation that $\rho_N - \rho_N^\downarrow \geq 0$ and $\sigma_N - \sigma_N^\downarrow \geq 0$. The latter statement follows inductively, as $\rho_0 - \rho_0^\downarrow = 0 \geq 0$ and
\begin{align*}
    \rho_{i+1} - \rho^\downarrow_{i+1} &= \Lambda_i((T_A\otimes \idop)(\rho_{i}) -  (T_A^\downarrow\otimes \idop)(\rho^\downarrow_{i}) )\\
    &= \Lambda_i((T_A\otimes \idop)(\rho_i - \rho_i^\downarrow) +((T_A-T_A^\downarrow)\otimes \idop)(\rho_i^\downarrow)) \\
    &\geq 0.
\end{align*}
The last line follows, as by induction $\rho_i - \rho_i^\downarrow \geq 0$ and $T_A-T_A^\downarrow$ is, by assumption, completely positive. Replacing $\rho$ by $\sigma$ and $A$ by $B$ in the argument above, shows that also $\sigma_N - \sigma_N^\downarrow \geq 0$. We write $\Delta\rho := \rho_N-\rho^\downarrow_N$ and $\Delta\sigma := \sigma_N-\sigma^\downarrow_N$ and use the strong concavity (see \cite{nielsen00}, p. 414) and the non-negativity of the fidelity, to obtain the following inequality 
\begin{align*}
    \sqrt{F}(\rho_N, \sigma_N) &= \sqrt{F}(\rho_N^\downarrow + \Delta\rho, \sigma_N^\downarrow + \Delta\sigma)\\
    &\geq  \sqrt{F}(\rho_N^\downarrow, \sigma_N^\downarrow) +  \sqrt{F}(\Delta\rho, \Delta\sigma) \\
    &\geq \sqrt{F}(\rho_N^\downarrow, \sigma_N^\downarrow),
\end{align*}
which is equivalent to \eqref{FildelityInequality}.
To prove \eqref{NoGoFidelityInequality}, it remains to show that
\begin{align} \label{NonTPInequality}
   1 - \sqrt{F}(\rho_N^\downarrow, \sigma_N^\downarrow) \leq C^{(T_A^\downarrow, T_B^\downarrow)}_{\mathcal{V}, \mathcal{W}} \sum_{i = 0}^{N-1}\sum_{k = 0}^{K} \sqrt{\tr{P_k \ptr{Z}{\rho_i^\downarrow}}\cdot\tr{P_k \ptr{Z}{\sigma_i^\downarrow}}}. 
\end{align}
To this end, notice that if $T^\downarrow_A\vert_{\trcl(\mathcal{V})} = T^\downarrow_B\vert_{\trcl(\mathcal{V})}$, then $(T^\downarrow_A\otimes \idop)\vert_{\trcl(\mathcal{V}\otimes\mathcal{H}_Z)} = (T^\downarrow_B\otimes\idop)\vert_{\trcl(\mathcal{V}\otimes\mathcal{H}_Z)}$. Hence, $T_A^\prime :=  (T^\downarrow_A \otimes \idop)$, $T_B^\prime := (T^\downarrow_B \otimes \idop)$, $\mathcal{V}^\prime := \mathcal{V} \otimes \mathcal{H}_Z$ and $\mathcal{W}^\prime := \{\mathcal{W}_1 \otimes \mathcal{H}_Z, \dots, \mathcal{W}_K \otimes \mathcal{H}_Z\}$ satisfy the assumptions of Proposition \ref{No-GoInequality}. Furthermore, as the fidelity is non-decreasing under the channel $\Lambda_i$ (see \cite{nielsen00}, p. 414), we have
\begin{align*}
    \sqrt{F}(\rho^\downarrow_i, \sigma_i^\downarrow) - \sqrt{F}(\rho^\downarrow_{i+1}, \sigma^\downarrow_{i+1}) &= \sqrt{F}(\rho^\downarrow_i, \sigma^\downarrow_i) - \sqrt{F}(\Lambda_i\circ T_A^\prime(\rho^\downarrow_{i}), \Lambda_i\circ T_B^\prime(\sigma^\downarrow_{i})) 
    \\ & \leq \sqrt{F}(\rho^\downarrow_i, \sigma^\downarrow_i) - \sqrt{F}(T_A^\prime(\rho^\downarrow_{i}), T_B^\prime(\sigma^\downarrow_{i})).
\end{align*}
We want to apply Proposition \ref{No-GoInequality} to the RHS of this expression. To do this correctly, we should notice that the projections, appearing in \eqref{InformationInequality}, project onto $\mathcal{W}_k\otimes \mathcal{H}_Z$, hence are equal to $P_k \otimes \idmat$. Also, if $V_A, V_B, P_A$ and $P_B$ satisfy the conditions \eqref{StinespringEqualityCondition} and \eqref{StinespringDilationCondition} then $V_A\otimes\idmat, V_B\otimes\idmat, P_A\otimes\idmat$ and $P_B\otimes\idmat$ satisfy the conditions \eqref{StinespringEqualityCondition} and \eqref{StinespringDilationCondition} for $T_A^\prime$ and $T_B^\prime$. If we plug this into \eqref{DefCabFactor} and use that in general $\norm{X\otimes\idmat} = \norm{X}$, we obtain
\begin{align*}
    C_{\mathcal{V}^\prime, \mathcal{W}^\prime}^{(T_A^\downarrow\otimes \idop, T_B^\downarrow\otimes \idop)} \leq C_{\mathcal{V}, \mathcal{W}}^{(T_A^\downarrow, T_B^\downarrow)}
\end{align*}
Using these observations, we get
\begin{align*}
\sqrt{F}(\rho^\downarrow_i, \sigma^\downarrow_i) - \sqrt{F}(\rho^\downarrow_{i+1}, \sigma^\downarrow_{i+1}) &\leq C_{\mathcal{V}, \mathcal{W}}^{(T_A^\downarrow, T_B^\downarrow)} \sum_{k = 1}^K \sqrt{\tr{(P_k\otimes \idmat) \rho^\downarrow_i}\tr{(P_k \otimes \idmat) \sigma^\downarrow_i}} \\
&= C_{\mathcal{V}, \mathcal{W}}^{(T_A^\downarrow, T_B^\downarrow)} \sum_{k = 1}^K\sqrt{\tr{P_k \ptr{Z}{\rho^\downarrow_i}}\tr{P_k \ptr{Z}{\sigma^\downarrow_i}}},
\end{align*}
Equivalently,
\begin{align*}
\sqrt{F}(\rho^\downarrow_{i+1}, \sigma^\downarrow_{i+1}) \geq \sqrt{F}(\rho^\downarrow_i, \sigma^\downarrow_i) -  C_{\mathcal{V}, \mathcal{W}}^{(T_A^\downarrow, T_B^\downarrow)} \sum_{k = 1}^K\sqrt{\tr{P_k \ptr{Z}{\rho^\downarrow_i}}\tr{P_k \ptr{Z}{\sigma^\downarrow_i}}}.
\end{align*}
If we iterate this inequality, we obtain
\begin{align*}
    \sqrt{F}(\rho^\downarrow_N, \sigma^\downarrow_N) \geq \sqrt{F}(\rho^\downarrow_0, \sigma^\downarrow_0) - C_{\mathcal{V}, \mathcal{W}}^{(T_A^\downarrow, T_B^\downarrow)} \sum_{i = 0}^{N-1}\sum_{k = 1}^K \sqrt{\tr{P_k \ptr{Z}{\rho^\downarrow_i}}\tr{P_k \ptr{Z}{\sigma^\downarrow_i}}}. 
\end{align*}
Using that $\sqrt{F}(\rho^\downarrow_0, \sigma^\downarrow_0) = \sqrt{F}(s_0, s_0) = 1$ and some rearrangement establishes \eqref{NonTPInequality} and completes the proof of the theorem.
\end{proof}

\noindent To connect this technical result with the main results of this section, we need two auxilliary lemmas. 

\begin{lem} \label{RelationErrorFidelity}
For $\mathrm{dim}(\mathcal{H}) < \infty$, let $T_A, T_B : \trcl(\mathcal{H}) \rightarrow \trcl(\mathcal{H})$ be two channels and let $D$ be a finite-dimensional $N$-step discrimination strategy and $\Pi$ be a two-valued POVM. Then 
\begin{align*}
    \frac{(1 - 2 P_e(D, \Pi))^2}{2} \leq  1-\sqrt{F}(\rho_N^{T_A}, \rho_N^{T_B}),
\end{align*}
where $\rho$ is the intermediate state map of $D$.
\end{lem}
\begin{proof}
By definition, 
\begin{align*}
P_e(D,\Pi) = \frac{1}{2} \left[ \tr{\pi_B \rho_N^{T_A}} + \tr{\pi_A \rho_N^{T_B}} \right].
\end{align*}
If we minimize over the possible two-valued POVMs $\Pi^\prime$, the famous Holevo-Helstrom formula reads
\begin{align*}
    P^m_e(D) := \min_{\Pi^\prime} P_e(D, \Pi^\prime) = \frac{1}{2}\left[1 - \frac{1}{2}\norm{\rho_N^{T_A} - \rho_N^{T_B}}_1 \right].
\end{align*}
Since $0 \leq P_e(D, \Pi) \leq \frac{1}{2}$, we have $1 - 2 P_e(D, \Pi) \geq 0$. Thus,
\begin{align} \label{BasicInequalityEasy}
    \frac{(1 - 2 P_e(D, \Pi))^2}{2} \leq \frac{(1 - 2 P_e^m(D))^2}{2}.
\end{align}
By the Fuchs-van de Graaf inequality (see \cite{nielsen00}, p. 416),
\begin{align*}
    \frac{1}{2}\norm{\rho - \sigma}_1 \leq \sqrt{1-\sqrt{F}(\rho, \sigma)^2}.
\end{align*}
Thus,
\begin{align*}
    \frac{(1 - 2 P_e^m(D))^2}{2} &= \frac{\left( \frac{1}{2} \norm{\rho_N^{T_A} - \rho_N^{T_B}}_1 \right)^2}{2}\\ &\leq \frac{1-\sqrt{F}(\rho_N^{T_A}, \rho_N^{T_B})^2}{2} 
    \\&= (1-\sqrt{F}(\rho_N^{T_A}, \rho_N^{T_B})) \frac{1+\sqrt{F}(\rho_N^{T_A}, \rho_N^{T_B})}{2} \\&\leq 1-\sqrt{F}(\rho_N^{T_A}, \rho_N^{T_B}).
\end{align*}
Together with \eqref{BasicInequalityEasy}, this proves the claim.
\end{proof}

\begin{lem} \label{CPLemma}
For $\mathrm{dim}(\mathcal{H}) < \infty$, let $T_A, T_B : \trcl(\mathcal{H}) \rightarrow \trcl(\mathcal{H})$ be two channels with vacuum $v \in \mathcal{H}$. Let $\mathcal{V}_{T_A}$ and $\mathcal{V}_{T_B}$ be the respective maximal vacuum subspaces and let $T_A^\downarrow$ and $T_B^\downarrow$ be as in Definition \ref{InteractionProbabilityDef} (Eq. \ref{DefinitionDownarrowChannel}). Furthermore, let $\mathcal{V}$ be a subspace such that $v \in \mathcal{V}$ and $\mathcal{V} \subseteq \mathcal{V}_{T_A} \cap \mathcal{V}_{T_B}$. Let $\mathcal{W} = \{\mathcal{W}_1, \mathcal{W}_2, \dots, \mathcal{W}_K\}$ be a collection of mutually orthogonal subspaces of $\mathcal{V}^\bot$, such that $\mathcal{V}^\bot = \mathcal{W}_1 \oplus \mathcal{W}_2 \oplus \dots \oplus \mathcal{W}_K$. Denote the orthogonal projections onto these subspaces by $P_1, P_2, \dots, P_K$. Then
\begin{align*}
\frac{(1 - 2 P_e(D, \Pi))^2}{2} \leq C_{\mathcal{V}, \mathcal{W}}^{(T_A^\downarrow, T_B^\downarrow)} \sum_{i = 0}^{N-1}\sum_{k = 1}^{K} \sqrt{\tr{P_k \ptr{Z}{\rho_i^{T_A^\downarrow}}}\cdot\tr{P_k \ptr{Z}{\rho_i^{T_B^\downarrow}}}},
\end{align*}
for all finite-dimensional $N$-step discrimination strategies $D = (\mathcal{H}, \mathcal{H}_Z, s_0, \Lambda)$ and all two-valued POVMs, $\Pi$.
\end{lem}
\begin{proof}
By Lemma \ref{RelationErrorFidelity}, we have for any finite-dimensional $N$-step discrimination strategy $D$ and any two-valued POVM, $\Pi$, that
\begin{align} \label{RHSFirstStepProofNo-Go}
    \frac{(1 - 2 P_e(D, \Pi))^2}{2} \leq  1-\sqrt{F}(\rho_N^{T_A}, \rho_N^{T_B}).
\end{align}
We want to apply Proposition \ref{NoGoTheorem} to the RHS of this inequality. To this end, we have to define the quantities appearing in that proposition. We identify $T_A, T_B, \mathcal{V}$ and $\mathcal{W}$ with the objects that bare the same name. In the following let $X \in \{A, B\}$. We define $T_X^\downarrow$ as in Definition \ref{InteractionProbabilityDef} and need to check that $T_X - T^\downarrow_X$ is completely positive and that $T_X^\downarrow\vert_{\trcl(\mathcal{V})}$ is trace-preserving.
To this end, fix a Stinespring isometry $V_X : \mathcal{H} \rightarrow \mathcal{H}_E \otimes \mathcal{H}$ of $T_X$. Then $T^\downarrow_X$ is defined by
\begin{align*}
T_X^\downarrow(\cdot) = \ptr{E}{(P_v^{(X)}\otimes \idmat) V_X \cdot V_X^\dagger},
\end{align*}
where $P_v^{(X)}$ is the projection onto the support of $\ptr{\mathcal{H}}{V_X \ket{v}\bra{v} V_X^\dagger}$. It follows immediately from this expression that $T_X-T_X^\downarrow$ is completely positive. To see that $T_X^\downarrow\vert_{\trcl(\mathcal{V}_{T_X})}$ is trace-preserving, note that by Definition \ref{MaximalVacuumSubspaceDef} 
\begin{align*}
    \mathcal{V}_{T_X} = V_X^{-1} \left[ \mathrm{supp}(\ptr{\mathcal{H}}{ V_X \ket{v}\bra{v} V_X^\dagger}) \otimes \mathcal{H} \right].
\end{align*}
Thus, for any\footnote{Remember that for a subspace $\mathcal{V}_0 \subseteq \mathcal{H}$, the operators in $\trcl(\mathcal{V}_0)$ are those that can be written in the form $\sum_{i, j} \alpha_{i j} \ket{\psi_i}\bra{\psi_i}$, with $\alpha_{i j} \in \mathbb{C}$ and $\psi_i \in \mathcal{V}_0$. } $\rho \in \trcl(\mathcal{V}_{T_X})$, 
\begin{align*}
    V_X \rho V_X^\dagger \in \trcl(\mathrm{supp}(\ptr{\mathcal{H}}{ V_X \ket{v}\bra{v} V_X^\dagger}) \otimes \mathcal{H}).
\end{align*}
As $P_v^{(X)}\otimes \idmat$ is the projection onto $\mathrm{supp}(\ptr{\mathcal{H}}{ V_X \ket{v}\bra{v} V_X^\dagger}) \otimes \mathcal{H}$, we have
\begin{align*}
    T_X^\downarrow\vert_{\trcl(\mathcal{V}_{T_X})}(\cdot) &= \ptr{E}{(P_v^{(X)}\otimes \idmat) V_X \cdot V_X^\dagger} = \ptr{E}{V_X \cdot V_X^\dagger} = T_X\vert_{\trcl(\mathcal{V}_{T_X})}(\cdot).
\end{align*}
Thus, $T_X^\downarrow\vert_{\trcl(\mathcal{V}_{T_X})}$ is trace-preserving, as $T_X\vert_{\trcl(\mathcal{V}_{T_X})}$ is. As $\mathcal{V}$ is a subspace of $\mathcal{V}_{T_X}$, also $T_X^\downarrow\vert_{\trcl(\mathcal{V})}$ is trace-preserving. This is what we have claimed. As all assumptions are satisfied, we can invoke Proposition \ref{NoGoTheorem}, which directly yields the desired inequality.
\end{proof}

\noindent The next result has already been stated in the results section. 

\begin{thm}[No-go theorem] \label{QuantitativeNoGo}
For $\mathrm{dim(\mathcal{H})} < \infty$, let $T_A, T_B : \trcl(\mathcal{H}) \rightarrow \trcl(\mathcal{H})$ be two channels with vacuum $v\in \mathcal{H}$.
If there exists no subspace $\mathcal{V} \subseteq \mathcal{H}$ such that $v \in \mathcal{V}$, at least one of the channels $T_A$ or $T_B$ is isometric on $\mathcal{V}$ and $T_A\vert_{\trcl(\mathcal{V})} \neq T_B\vert_{\trcl(\mathcal{V})}$, then there exists a constant $C < \infty$, such that 
\begin{align*} 
    (1-2P_e(D, \Pi))^2 \leq  C \sqrt{P_I^{T_A}(D)\cdot P_I^{T_B}(D)} \leq C \max(P_I^{T_A}(D), P_I^{T_B}(D)),
\end{align*}
for all finite-dimensional $N$-step discrimination strategies $D$ and all two-valued POVMs, $\Pi$.  
Hence, $T_A$ and $T_B$ cannot be discriminated in an \sq{interaction-free} manner.
\end{thm}
\begin{rem} \label{DifferentCharacterizationsRemark}
The assumption: "The statement that $T_A$ or $T_B$ is isometric on a subspace $\mathcal{V}$, with $v \in \mathcal{V}$, already implies that $T_A\vert_{\trcl(\mathcal{V})} = T_B\vert_{\trcl(\mathcal{V})}$" can be rephrased in two equivalent ways. The first one is that the conditions \ref{MainProp1}, \ref{MainProp2} and \ref{MainProp3} in the Main Theorem (section \ref{ResultsSection}) cannot be fulfilled simultaneously.  
The second reformulation is that for the maximum vacuum subspaces $\mathcal{V}_{T_A}$ and $\mathcal{V}_{T_B}$, we have $\mathcal{V}_{T_A} = \mathcal{V}_{T_B}$ and $T_A\vert_{\trcl(\mathcal{V}_{T_A})} = T_B\vert_{\trcl(\mathcal{V}_{T_B})}$. The equivalence follows directly from the characterization of maximal vacuum subspaces in Lemma \ref{MaximalVacuumSubspaceLemma} \ref{prop4}. This second reformulation is not only important in the proof, but also if one wants to check this criterion, as $\mathcal{V}_{T_A}$ and $\mathcal{V}_{T_B}$ are efficiently computable directly from Definition \ref{MaximalVacuumSubspaceDef}. 
\end{rem}
\begin{proof}
We use the second characterization in Remark \ref{DifferentCharacterizationsRemark}. That is, $\mathcal{V}_{T_A} = \mathcal{V}_{T_B}$ and $T_A\vert_{\trcl(\mathcal{V}_{T_A})} = T_B\vert_{\trcl(\mathcal{V}_{T_B})}$. We set $\mathcal{V} := \mathcal{V}_{T_A}$ and let $T_A^\downarrow$ and $T_B^\downarrow$ be as in Definition \ref{InteractionProbabilityDef}. Furthermore, we define $\mathcal{W} := \{\mathcal{W}_1\}$, with $\mathcal{W}_1 := \mathcal{V^\bot}$.
Then, by Lemma \ref{CPLemma}, we have
\begin{align} \label{RHSSecondtStepProofNo-Go}
    (1 - 2 P_e(D, \Pi))^2 \leq 2C_{\mathcal{V}, \mathcal{W}}^{(T_A^\downarrow, T_B^\downarrow)} \sum_{i = 0}^{N-1} \sqrt{\tr{P^\bot \ptr{Z}{\rho_i^{T_A^\downarrow}}}\cdot\tr{P^\bot \ptr{Z}{\rho_i^{T_B^\downarrow}}}},
\end{align}
where $P^\bot$ is the orthogonal projection onto $\mathcal{W}_1 = \mathcal{V}^\bot$. 
As $\mathcal{V}$ is the maximum vacuum subspace of $T_A$ and $T_B$, Lemma \ref{MaximalVacuumSubspaceLemma} \ref{prop5} implies that for $X \in \{A, B\}$, there is a constant $C_{T_X} > 0$ such that $\mathfrak{i}_{T_X}(\rho) \geq C_{T_X} \tr{P^\bot\, \rho}$ for all $\rho \geq 0$. As $\ptr{Z}{\rho_i^{T_X^\downarrow}} \geq 0$, we get
\begin{align*}
    \eqref{RHSSecondtStepProofNo-Go} &\leq \frac{2C_{\mathcal{V}, \mathcal{W}}^{(T_A^\downarrow, T_B^\downarrow)}}{\sqrt{C_{T_A} C_{T_B}}} \sum_{i = 0}^{N-1} \sqrt{\mathfrak{i}_{T_A}\left(\ptr{Z}{\rho_i^{T_A^\downarrow}}\right)\,\mathfrak{i}_{T_B}\left(\ptr{Z}{\rho_i^{T_B^\downarrow}}\right)} \\
    &\leq \frac{2C_{\mathcal{V}, \mathcal{W}}^{(T_A^\downarrow, T_B^\downarrow)}}{\sqrt{C_{T_A} C_{T_B}}} \sqrt{\left( \sum_{i = 0}^{N-1} \mathfrak{i}_{T_A}\left(\ptr{Z}{\rho_i^{T_A^\downarrow}}\right)\right)\,\left(\sum_{i = 0}^{N-1}\mathfrak{i}_{T_B}\left(\ptr{Z}{\rho_i^{T_B^\downarrow}}\right)\right)} \\
    &= \frac{2C_{\mathcal{V}, \mathcal{W}}^{(T_A^\downarrow, T_B^\downarrow)}}{\sqrt{C_{T_A} C_{T_B}}} \sqrt{P_I^{T_A}(D)\cdot P_I^{T_B}(D)},
\end{align*}
where we used the Cauchy-Schwarz inequality (on $\mathbb{C}^{N}$) to obtain the second line and the definition of the \sq{interaction} probability in the last line. We note that the last inequality in the statement of the theorem is trivial. Thus, by setting $C := \frac{2C_{\mathcal{V}, \mathcal{W}}^{(T_A^\downarrow, T_B^\downarrow)}}{\sqrt{C_{T_A} C_{T_B}}}$, we have proven the claim.  
\end{proof}

\noindent The following Theorem is the technical version of the result stated in the results section.  

\begin{thm}[Rate limit theorem] \label{RateLimitTheorem}
For $\mathrm{dim(\mathcal{H})} < \infty$, let $T_A, T_B : \trcl(\mathcal{H}) \rightarrow \trcl(\mathcal{H})$ be two channels with vacuum $v\in \mathcal{H}$. Let $\mathcal{V}_{T_A}$ and $\mathcal{V}_{T_B}$ be the respective maximal vacuum subspace of $T_A$ and $T_B$. Set $\mathcal{V} := \mathcal{V}_{T_A} \cap \mathcal{V}_{T_B}$. Suppose that $T_A\vert_{\trcl(\mathcal{V})} = T_B\vert_{\trcl(\mathcal{V})}$ and that $\mathcal{V}^\bot \cap \mathcal{V}_A$ and $\mathcal{V}^\bot \cap \mathcal{V}_B$ are orthogonal. \\Then there exists a constant $C > 0$ such that
\begin{align*}
    \max(P_I^{T_A}(D), P_I^{T_B}(D)) \geq C\,\frac{(1-2P_e(D,\Pi))^4}{N},
\end{align*}
for all finite-dimensional $N$-step discrimination strategies $D$, and any two-valued POVM, $\Pi$.
\end{thm}

\begin{proof} The proof is similar to the one of the No-go theorem. Let $T_A^\downarrow$ and $T_B^\downarrow$ be as in Definition \ref{InteractionProbabilityDef} and set $\mathcal{V} := \mathcal{V}_{T_A} \cap \mathcal{V}_{T_B}$. Furthermore, define $\mathcal{W} := \{\mathcal{W}_1, \mathcal{W}_2, \mathcal{W}_3 \}$, with $\mathcal{W}_1 := \mathcal{V}^\bot \cap \mathcal{V}_{T_A}$, $\mathcal{W}_2 := \mathcal{V}^\bot \cap \mathcal{V}_{T_B}$ and $\mathcal{W}_3 := (\mathcal{W}_1\oplus\mathcal{W}_2)^\bot \cap \mathcal{V}^\bot$. Clearly, $\mathcal{W}_1, \mathcal{W}_2$ and $\mathcal{W}_3$ are mutually orthogonal and their direct sum is $\mathcal{V}^\bot$. Furthermore, $\mathcal{W}_2\oplus \mathcal{W}_3 = \mathcal{V}_{T_A}^\bot$ and $\mathcal{W}_1\oplus \mathcal{W}_3 = \mathcal{V}_{T_B}^\bot$. Thus, by Lemma \ref{CPLemma}, we have
\begin{align} \label{RHSSecondtStepProofNo-Go}
    (1 - 2 P_e(D, \Pi))^2 \leq 2C_{\mathcal{V}, \mathcal{W}}^{(T_A^\downarrow, T_B^\downarrow)} \sum_{i = 0}^{N-1}\sum_{k= 1}^{3} \sqrt{\tr{P_k \ptr{Z}{\rho_i^{T_A^\downarrow}}}\cdot\tr{P_k \ptr{Z}{\rho_i^{T_B^\downarrow}}}},
\end{align} 
where for $k \in \{1, 2, 3\}$, $P_k$ is the orthogonal projection onto $\mathcal{W}_k$. Using the Cauchy-Schwarz inequality (on $\mathbb{C}^{3}$), and the fact that probabilities are less than one, and afterwards the Cauchy-Schwarz inequality on $\mathbb{C}^N$, we get
\begin{align}\label{NoGoRateFthird}
    \eqref{RHSSecondtStepProofNo-Go} &\leq \sqrt{12}C_{\mathcal{V}, \mathcal{W}}^{(T_A^\downarrow, T_B^\downarrow)} \sum_{i = 0}^{N-1} \sqrt{\tr{(P_2+P_3) \ptr{Z}{\rho_i^{T_A^\downarrow}}} + \tr{(P_1+P_3) \ptr{Z}{\rho_i^{T_B^\downarrow}}}}\nonumber \\
    &\leq \sqrt{12N} C_{\mathcal{V}, \mathcal{W}}^{(T_A^\downarrow, T_B^\downarrow)} \sqrt{\sum_{i = 0}^{N-1} \tr{P_{\mathcal{V}_{T_A}}^\bot \ptr{Z}{\rho_i^{T_A^\downarrow}}} + \tr{P_{\mathcal{V}_{T_B}}^\bot \ptr{Z}{\rho_i^{T_B^\downarrow}}}},
\end{align}
where $P_{\mathcal{V}_{T_A}}^\bot$ and $P_{\mathcal{V}_{T_B}}^\bot$ are the projections onto $\mathcal{V}_{T_A}^\bot$ and $\mathcal{V}_{T_B}^\bot$.
Lemma \ref{MaximalVacuumSubspaceLemma}, \ref{prop5} implies that for $X \in \{A, B\}$, there is a constant $C_{T_X} > 0$ such that $\mathfrak{i}_{T_X}(\rho) \geq C_{T_X} \tr{P_{\mathcal{V}_{T_X}}^\bot\, \rho}$ for all $\rho \geq 0$. As $\ptr{Z}{\rho_i^{T_X^\downarrow}} \geq 0$, we get
\begin{align*}
    \eqref{NoGoRateFthird}&\leq \sqrt{12N} C_{\mathcal{V}, \mathcal{W}}^{(T_A^\downarrow, T_B^\downarrow)} \sqrt{C_{T_A}^{-1}\sum_{i = 0}^{N-1}\mathfrak{i}_{T_A}\left(\ptr{Z}{\rho_i^{T_A^\downarrow}}\right) + C_{T_B}^{-1}\sum_{i = 0}^{N-1}\mathfrak{i}_{T_B}\left(\ptr{Z}{\rho_i^{T_B^\downarrow}}\right)}\\
        &\leq C_\mathcal{V}^{(T_A^\downarrow, T_B^\downarrow)} \sqrt{\frac{24}{\min(C_{T_A}, C_{T_B})}}  \sqrt{N\,\max(P_I^{T_A}(D), P_I^{T_B}(D))}
\end{align*}
Taking the square and defining $C := \frac{\min(C_{T_A}, C_{T_B})}{24 {C_{\mathcal{V}, \mathcal{W}}^{(T_A^\downarrow, T_B^\downarrow)}}^2}$ proves the claim. 
\end{proof}

\newpage
\section{Conclusion and open problems}

In our work, we have characterized, when it is possible and impossible to discriminate quantum channels in an \sq{interaction-free} manner. Although this answers the question, what can be done perfectly with \sq{interaction-free} measurements there are still many questions left open. One question that is in direct succession of our work is, under which conditions two channels can be discriminated such that the \sq{interaction} probability decays faster than $~N^{-1}$. 
Another question would ask for a more quantitative treatment, i.e., even though one might not be able to discriminate two channels in an \sq{interaction-free} manner, there still might be a significant quantum advantage over classical strategies. A \sq{big} question concerns the influence of noise and decoherence. We may note that noise may influence what can or cannot be done in both directions, since the noise can also be on the Daemon's side and hence make his detection skills weaker. Before the no-go results for semitransparent objects were established \cite{PhysRevA.63.032105, PhysRevA.64.062303}, one anticipated application of \sq{interaction-free} measurement was to eliminate the exposure of humans to radiation in medical applications such as x-ray scans. This is not possible. However, our no-go theorem does not touch the case of asymmetric \sq{interaction-free} discrimination. That is, we may allow that one of the two objects to be discriminated gets destroyed (for example by simply setting its transmission functional to zero). This might even be a desirable effect. For example in a medical context, we would love to design a procedure such that a tumor gets destroyed, while the healthy tissue stays intact. 
\\\\
\noindent \textbf{Acknowledgment}
M.H. was supported by the Bavarian excellence network {\sc enb} via the \mbox{International} PhD Programme of Excellence {\em Exploring Quantum Matter} ({\sc exqm}).
M.M.W. acknowledges funding by the Deutsche Forschungsgemeinschaft (DFG, German Research Foundation) under Germany’s Excellence Strategy EXC-2111 390814868.
\\\\
\noindent \textbf{Remark} 
This work started as M.H.'s (unpublished) Master thesis \cite{MasterThesis}. However, the contents of \cite{MasterThesis} and the present work are largely disjoint. In particular, key results which include the main result, the \sq{interaction} model, the optimization of the needed resources and the determination of the rate of the best discrimination strategy were missing.

\appendix

\newpage
\section{} \label{SemisimplicityLemmaAppendix}

\begin{lem}[Semi-simplicity of the peripheral spectrum] 
Let $T: \trcl(\mathcal{H}) \rightarrow \trcl(\mathcal{H})$ be a channel such that $1$ is in the discrete spectrum of $T$. Then, for any $n \in \mathbb{N}$ and any (rectifiable) path inside the resolvent set of $T$ that encloses $1$, and separates $1$ from $\sigma(T)\setminus\{1\}$, we have
\begin{align} \label{RieszProjectionIdentity}
  \frac{1}{2\pi i} \oint\limits_{\Gamma_1}\frac{z^n}{z-T}\,\mathrm{d}z = \frac{1}{2\pi i} \oint\limits_{\Gamma_1}\frac{1}{z-T}\,\mathrm{d}z, 
\end{align}
\end{lem}
\begin{proof}
For brevity, we denote the Riesz-Projection on the RHS of \eqref{RieszProjectionIdentity} by $P$. As $1$ is in the discrete spectrum of $T$, Corollary 2.3.6 in \cite{simon2015operator} says that $T P = \frac{1}{2\pi i} \oint\limits_{\Gamma_1}\frac{z}{z-T}\,\mathrm{d}z = P + N$, where $N$ is a nilpotent operator that commutes with $P$. Hence $T = P + N + T_0$, where $T_0 := (\idop-P)T(\idop-P)$.
By the analytic functional calculus, we have
\begin{align*}
    (P+N)^n = \left(\frac{1}{2\pi i} \oint\limits_{\Gamma_1}\frac{z}{z-T}\,\mathrm{d}z\right)^n = \frac{1}{2\pi i} \oint\limits_{\Gamma_1}\frac{z^n}{z-T}\,\mathrm{d}z.
\end{align*}
If $N = 0$, then the claim follows, since $P$ is a projection ($P^n = P$). To this end, assume that $N \neq 0$. Since $N$ is nilpotent, there exists and integer $D$ such that $N^{D} \neq 0$ and $N^{D+1} = 0$. As $N \neq 0$, we have $D \geq 1$. Choose $\rho \in \trcl(\mathcal{H})$ such that $N^D(\rho) \neq 0$ and $P(\rho) = \rho$. Note that $P T_0 = T_0P = 0$. Thus $T^n(\rho) = (P+N)^n(\rho) + T_0^n(\rho) = (P+N)^n(\rho)$. In particular, since $T$ is a channel, $\norm{T^n} = 1$ and thus 
\begin{align} \label{boundednessEq}
    \norm{(P+N)^n(\rho)} \leq \norm{\rho}.
\end{align} 
For $n \geq D$, we have 
\begin{align*}
    (P+N)^n(\rho) = \sum_{i = 0}^D {n \choose i} N^i(\rho).
\end{align*}
Furthermore, the vectors $\rho, N(\rho), N^2(\rho), \dots N^D(\rho)$ are linearly independent. The coordinate function of $N(\rho)$ is ${n\choose 1}$, which is unbounded for $n \rightarrow \infty$. Since the coordinate function can be extended to a continuous linear functional on $\trcl(\mathcal{H})$ (Hahn-Banach), the unboundedness contradicts \eqref{boundednessEq}. Hence, $N = 0$.  
\end{proof}

\bibliography{references} 
\bibliographystyle{ieeetr}



\end{document}